\documentclass{article}
\usepackage[margin=1in]{geometry}
\usepackage{amsthm,amsmath,amssymb,enumerate, graphicx,subcaption}
\usepackage{setspace}
\usepackage{thmtools}
\usepackage{thm-restate}
\usepackage{cleveref,color}
\usepackage[]{algorithm2e}
\usepackage{natbib}
\usepackage[mathscr]{euscript}
\usepackage[shortlabels]{enumitem}

\declaretheorem[name=Theorem]{thm}

\declaretheorem[name=Lemma, sibling=thm]{lemma}

\allowdisplaybreaks

\DeclareMathOperator*{\inprob}{\stackrel{P}{\longrightarrow}}

\DeclareMathOperator*{\indist}{\stackrel{d}{\longrightarrow}}
\DeclareMathOperator*{\bounded}{O_P}
\DeclareMathOperator*{\fasterthan}{o_P}
\DeclareMathOperator*{\fasterthandet}{o}

\DeclareMathOperator*{\argmin}{argmin}
\DeclareMathOperator*{\argmax}{argmax}

\DeclareMathOperator*{\expit}{expit}

\newcommand{\s}[1]{\mathscr{#1}}
\renewcommand{\d}[1]{\mathbb{#1}}

\newcommand{\n}[1]{\mathrm{#1}}

\usepackage{natbib}

\title{Causal isotonic regression }
\author{Ted\ Westling\\ Department of Mathematics and Statistics\\ University of Massachusetts Amherst \\ twestling@math.umass.edu \and Peter Gilbert\\ Vaccine and Infectious Disease Division\\ Fred Hutchinson Cancer Research Center \\ pgilbert@scharp.org \and Marco Carone\\ Department of Biostatistics \\ University of Washington \\ mcarone@uw.edu }
\date{}

\begin{document}

\maketitle

\begin{abstract}
In observational studies, potential confounders may distort the causal relationship between an exposure and an outcome. However, under some conditions, a causal dose-response curve can be recovered using the $G$-computation formula. Most classical methods for estimating such curves when the exposure is continuous rely on restrictive parametric assumptions, which carry significant risk of model misspecification. Nonparametric estimation in this context is challenging because in a nonparametric model these curves cannot be estimated at regular rates. Many available nonparametric estimators are sensitive to the selection of certain tuning parameters, and performing valid inference with such estimators can be difficult. In this work, we propose a nonparametric estimator of a causal dose-response curve known to be monotone. We show that our proposed estimation procedure generalizes the classical least-squares isotonic regression estimator of a monotone regression function. Specifically, it does not involve tuning parameters, and is invariant to strictly monotone transformations of the exposure variable. We describe theoretical properties of our proposed estimator, including its irregular limit distribution and the potential for doubly-robust inference. Furthermore, we illustrate its performance via numerical studies, and use it to assess the relationship between BMI and immune response in HIV vaccine trials.
\end{abstract}

\doublespacing

\section{Introduction}

\subsection{Motivation and literature review}

Questions regarding the causal effect of an exposure on an outcome are ubiquitous in science. If investigators are able to carry out an experimental study in which they randomly assign a level of exposure to each participant and then measure the outcome of interest, estimating a causal effect is generally straightforward. However, such studies are often not feasible, and data from observational studies must be relied upon instead. Assessing causality is then more difficult, in large part because of potential confounding of the relationship between exposure and outcome. Many nonparametric methods have been proposed for drawing inference about a causal effect using observational data when the exposure of interest is either binary or categorical -- these include, among others, inverse probability weighted (IPW) estimators \citep{horvitz1952sampling}, augmented IPW estimators  \citep{scharfstein1999adjusting, bang2005doubly}, and targeted minimum loss-based estimators (TMLE) \citep{vanderlaan2011tmle}. 

In practice, many exposures are continuous, in the sense that they may take any value in an interval. A common approach to dealing with such exposures is to simply discretize the interval into two or more regions, thus returning to the categorical exposure setting. However, it is frequently of scientific interest to learn the causal dose-response curve, which describes the causal relationship between the exposure and outcome across a continuum of the exposure. Much less attention has been paid to continuous exposures. \cite{robins2000msm}  and \cite{zhang2016quantitative} studied this problem using parametric models, and \cite{neugebauer2007jspi} considered inference on parameters obtained by projecting a causal dose-response curve onto a parametric working model. Other authors have taken a nonparametric approach instead. \cite{rubin2006msm} and \cite{diaz2011super} discussed nonparametric estimation using flexible data-adaptive algorithms. \cite{kennedy2016continuous} proposed an estimator based on local linear smoothing. Finally, \cite{van2018non} recently presented a general framework for inference on parameters that fail to be smooth enough as a function of the data-generating distribution and for which regular root-$n$ estimation theory is therefore not available. This is indeed the case for the causal dose-response curve, and \cite{van2018non} discussed inference on such a parameter as a particular example.

Despite a growing body of literature on nonparametric estimation of causal dose-response curves, to the best of our knowledge, existing methods do not permit valid large-sample inference and may be sensitive to the selection of certain tuning parameters. For instance, smoothing-based methods are often sensitive to the choice of a kernel function and bandwidth, and these estimators typically possess non-negligible asymptotic bias, which complicates the task of performing valid inference.


In many settings, it may be known that the causal dose-response curve is monotone in the exposure.  For instance, exposures such as daily exercise performed, cigarettes smoked per week, and air pollutant levels are all known to have monotone relationships with various health outcomes. In such cases, an extensive literature suggests that monotonicity may be leveraged to derive estimators with desirable properties -- the monograph of \cite{groene2014shape} provides a comprehensive overview. For example, in the absence of confounding, isotonic regression may be employed to estimate the causal dose-response curve \citep{barlow1972order}. The isotonic regression estimator does not require selection of a kernel function or bandwidth, is invariant to strictly increasing transformations of the exposure, and upon centering and scaling by  $n^{-1/3}$, converges in law pointwise to a symmetric limit distribution with mean zero \citep{brunk1970regression}. The latter property is useful since it facilitates asymptotically valid pointwise inference.

Nonparametric inference on a monotone dose-response curve when the exposure-outcome relationship is confounded is more difficult to tackle and is the focus of this manuscript.   To the best of our knowledge, this problem has not been comprehensively studied  before.

\subsection{Parameter of interest and its causal interpretation}\label{sec:param}

The prototypical data unit we consider is $O = (Y, A,W)$, where $Y$ is a response, $A$ a continuous exposure, and $W$ a vector of covariates. The support of the true data-generating distribution $P_0$ is denoted by $\s{O} = \s{Y} \times \s{A} \times \s{W}$, where $\s{Y} \subseteq \d{R}$,  $\s{A} \subseteq \d{R}$ is an interval, and $\s{W} \subseteq \d{R}^p$.  Throughout,  the use of subscript $0$ refers to evaluation at or under $P_0$. For example, we write $\theta_0$ and $F_0$ to denote $\theta_{P_0}$ and $F_{P_0}$, respectively, and $E_0$ to denote expectation under $P_0$. 

Our parameter of interest is the so-called \emph{$G$-computed regression function} from $\mathscr{A}$ to $\mathbb{R}$, defined as \[a\mapsto \theta_0(a):=E_0\left[E_0\left(Y\mid A=a,W\right)\right]\ ,\] where the outer expectation is with respect to the marginal distribution $Q_0$ of $W$. In some scientific contexts, $\theta_0(a)$ may have a causal interpretation. Adopting the Neyman-Rubin potential outcomes framework, for each $a \in \s{A}$, we denote by $Y(a)$ a unit's potential outcome under exposure level $A =a$. The causal parameter $m_0(a) := E_0\left[Y(a)\right]$ corresponds to the average outcome under assignment of the entire population to exposure level $A=a$. The resulting curve $m_0:\s{A}\rightarrow \d{R}$ is what we formally define as the \emph{causal dose-response curve}. Under varying sets of causal conditions, $m_0(a)$ may be identified with  functionals of the observed data distribution, such as the unadjusted regression function $r_0(a) := E_0\left(Y \mid A =a\right)$ or the $G$-computed regression function $\theta_0(a)$.

Suppose that (i) each unit's potential outcomes are independent of all other units' exposures; and (ii) the observed outcome $Y$ equals the potential outcome $Y(A)$ corresponding to the exposure level $A$ actually received. Identification of $m_0(a)$ further depends on the relationship between $A$ and $Y(a)$. If (i) and (ii) hold, and in addition, (iii) $A$ and $Y(a)$ are independent, and (iv) the marginal density of $A$ is positive at $a$, then $m_0(a) = r_0(a)$.  Condition (iii) typically only holds in experimental studies (e.g., randomized trials). In observational studies, there are often common causes of $A$ and $Y(a)$ -- so-called \emph{confounders} of the exposure-outcome relationship -- that induce dependence. In such cases, $m_0(a)$ and $r_0(a)$ do not generally coincide. However, if $W$ contains a sufficiently rich collection of confounders, it may still be possible to identify $m_0(a)$ from the observed data. If (i) and (ii) hold, and in addition, (v) $A$ and $Y(a)$ are conditionally independent given $W$, and (vi) the conditional density of $A$ given $W$ is almost surely positive at $A=a$, then $m_0(a)= \theta_0(a)$. This is a fundamental result in causal inference \citep{robins1986, gill2001}. Whenever $m_0(a) = \theta_0(a)$, our methods can be interpreted as drawing inference on the causal dose-response parameter $m_0(a)$.

We note that the definition of the counterfactual outcome $Y(a)$ presupposes that the intervention setting $A=a$ is uniquely defined.  In many situations, this stipulation requires careful thought. For example, in Section~\ref{bmi} we consider an application in which body mass index (BMI) is the exposure of interest. There is an ongoing scientific debate about whether such an exposure leads to a meaningful causal interpretation, since it is not clear what it means to intervene on BMI.

Even if the identifiability conditions stipulated above do not strictly hold or the scientific question is not causal in nature, when $W$ is associated with both $A$ and $Y$, $\theta_0(a)$ often has a more appealing interpretation than the unadjusted regression function $r_0(a)$.  Specifically, $\theta_0(a)$ may be interpreted as the average value of $Y$ in a population with exposure fixed at $A =a$  but otherwise characteristic of the study population with respect to $W$. Because $\theta_0(a)$ involves both adjustment for $W$ and marginalization with respect to a single reference population that does not depend on the value $a$, the comparison of $\theta_0(a)$ over different values of $a$ is generally more meaningful than for $r_0(a)$.

When $P_0(A = a) = 0$, the parameter $P \mapsto \theta_P(a)$ is not pathwise differentiable at $P_0$ with respect to the nonparametric model \citep{diaz2011super}. Heuristically, due to the continuous nature of $A$, $\theta_P(a)$ corresponds to a local feature of $P$. As a result, regular root-$n$ rate estimators cannot be expected, and standard methods for constructing efficient estimators of pathwise differentiable parameters in nonparametric and semiparametric models (e.g., estimating equations, one-step estimation, targeted minimum loss-based estimation) cannot be used directly to target and obtain inference on $\theta_0(a)$.



\subsection{Contribution and organization of the article}

We denote by $F_P:\mathscr{A}\rightarrow\mathbb{R}$ the distribution function of $A$ under $P$, by $\mathscr{F}_\theta$ the class of non-decreasing real-valued functions on $\mathscr{A}$, and by $\mathscr{F}_F$ the class of strictly increasing and continuous distribution functions supported on $\mathscr{A}$. The statistical model we will work in is $\mathscr{M}:=\{P:\theta_P\in\mathscr{F}_\theta,F_P\in \mathscr{F}_F\}$, which consists of the collection of  distributions for which $\theta_P$ is non-decreasing over $\mathscr{A}$ and the marginal distribution of $A$ is continuous with positive Lebesgue density over $\mathscr{A}$.

In this article, we study nonparametric estimation and inference on the $G$-computed regression function $a\mapsto \theta_0(a)= E_0\left[E_0\left(Y \mid A = a, W\right)\right]$ for use when $A$ is a continuous exposure and $\theta_0$ is known to be monotone. Specifically,  our goal is to make inference about $\theta_0(a)$ for $a\in\mathscr{A}$ using independent observations $O_1,O_2,\ldots,O_n$ drawn from $P_0\in\mathscr{M}$.  This problem is an extension of classical isotonic regression to the setting in which the exposure-outcome relationship is confounded by recorded covariates -- this is why we refer to the method proposed as \emph{causal isotonic regression}.  As mentioned above, to the best of our knowledge, nonparametric estimation and inference on a monotone $G$-computed regression function has not been comprehensively studied before. In what follows, we:%

\begin{enumerate}
\item show that our proposed estimator generalizes the unadjusted isotonic regression estimator to the more realistic scenario in which there is confounding by recorded covariates;
\item investigate finite-sample and asymptotic properties of the proposed estimator, including invariance to strictly increasing transformations of the exposure, doubly-robust consistency, and doubly-robust convergence in distribution to a non-degenerate limit;
\item derive practical methods for constructing pointwise confidence intervals, including intervals that have valid doubly-robust calibration;
\item illustrate numerically the practical performance of the proposed estimator.
\end{enumerate}
%
%

We note that in \cite{westling2018monotone}, we studied estimation of $\theta_0$ as one of several examples of a general approach to monotonicity-constrained inference.  Here, we provide a comprehensive examination of estimation of a monotone dose-response curve. In particular, we establish novel theory and methods that have important practical implications. First, we provide conditions under which the estimator converges in distribution even when one of the nuisance estimators involved in the problem is inconsistent. This contrasts with the results in  \cite{westling2018monotone}, which required that both nuisance parameters be estimated consistently. We also propose two estimators of the scale parameter arising in the limit distribution, one of which requires both nuisance estimators to be consistent, and the other of which does not. Second, we demonstrate that our estimator is invariant to strictly monotone transformations of the exposure. Third, we study the joint convergence of our proposed estimator at two points, and use this result to construct confidence intervals for causal effects. Fourth, we study the behavior of our estimator in the context of discrete exposures. Fifth, we propose an alternative estimator based on cross-fitting of the nuisance estimators, and demonstrate that this strategy removes the need for empirical process conditions required in \cite{westling2018monotone}. Finally, we investigate the behavior of our estimator in comprehensive numerical studies, and compare its behavior to that of the local linear estimator of \cite{kennedy2016continuous}.


The remainder of the article is organized as follows. In Section~\ref{estimator}, we concretely define the proposed estimator. In Section~\ref{theoretical}, we study theoretical properties of the proposed estimator. In Section~\ref{inference}, we propose methods for pointwise inference. In Section~\ref{numerical}, we perform numerical studies to assess the performance of the proposed estimator, and in Section~\ref{bmi}, we use this procedure to investigate the relationship between BMI and immune response to HIV vaccines using data from several randomized trials. Finally, we provide concluding remarks in Section~\ref{discussion}. Proofs of all theorems are provided in Supplementary Material.

\section{Proposed approach}\label{estimator}



\subsection{Review of isotonic regression}

Since the proposed estimator of $\theta_0(a)$ builds upon isotonic regression, we briefly review the classical least-squares isotonic regression estimator of  $r_0(a)$. The isotonic regression $r_n$ of $Y_1,Y_n, \dotsc, Y_n$ on $A_1,A_2, \dotsc, A_n$ is the minimizer in $r$ of $\sum_{i=1}^n [Y_i - r(A_i)]^2$ over all  monotone non-decreasing functions. This minimizer can be obtained via the Pool Adjacent Violators Algorithm \citep{ayer1955empirical, barlow1972order}, and can also be represented in terms of greatest convex minorants (GCMs). The GCM of a bounded function $f$ on an interval $[a,b]$ is defined as the supremum over all convex functions $g$ such that $g \leq f$. Letting $F_n$ be the empirical distribution function of $A_1,A_2, \dotsc, A_n$, $r_n(a)$ can be shown to equal the left derivative, evaluated at $F_n(a)$, of the GCM over the interval $[0,1]$ of the linear interpolation of the so-called \emph{cusum diagram} \[\bigg{\{} \tfrac{1}{n}\bigg{(}i, \sum_{j=0}^i Y_{(i)}^*\bigg{)} : i = 0, 1, \dotsc, n\bigg{\}}\ ,\] where $Y_{(0)}^* := 0$ and $Y_{(i)}^*$ is the value of $Y$ corresponding to the observation with $i^{th}$ smallest value of $A$.

The isotonic regression estimator $r_n$ has many attractive properties. First, unlike smoothing-based estimators, isotonic regression does not require the choice of a kernel function, bandwidth, or any other tuning parameter. Second, it is invariant to strictly increasing  transformations of $A$. Specifically, if $H : \s{A} \to \d{R}$ is a strictly increasing function, and $r_n^*$ is the isotonic regression of $Y_1, Y_2,\dotsc, Y_n$ on $H(A_1),H(A_2), \dotsc, H(A_n)$, it follows that $r_n^* = r_n \circ H^{-1}$. Third, $r_n$ is uniformly consistent on any strict subinterval of $\s{A}$. Fourth, $n^{1/3}[r_n(a) - r_0(a)]$ converges in distribution to $\left[ 4 r_0'(a) \sigma_0^2(a) / f_0(a) \right]^{1/3} \d{W}$ for any interior point $a$ of $\s{A}$ at which $r_0'(a)$, $f_0(a) := F_0'(a)$ and $\sigma_0^2(a) := E_0\left\{ [Y-r_0(a)]^2 \mid A = a\right\}$ exist, and are positive and continuous in a neighborhood of $a$. Here, $\d{W} := \argmax_{u \in \d{R}} \{ Z_0(u) - u^2\}$, where $Z_0$ denotes a two-sided Brownian motion originating from zero, and is said to follow \emph{Chernoff's distribution}. Chernoff's distribution has been extensively studied: among other properties, it is a log-concave and symmetric law centered at zero, has moments of all orders, and can be approximated by a $N(0, 0.52)$ distribution \citep{chernoff1964, groeneboom2001jcgs}. It appears often in the limit distribution of monotonicity-constrained estimators.

\subsection{Definition of proposed estimator}\label{sec:defn}

For any given $P \in \s{M}$, we  define the outcome regression pointwise as $\mu_P(a,w) := E_{P}\left(Y  \mid A=a, W=w\right)$,  and the normalized exposure density as $g_P(a, w) := \pi_P(a \mid w) / f_P(a)$, where $\pi_P(a \mid w)$ is the evaluation at $a$ of the conditional density function of $A$ given $W=w$ and $f_P$ is the marginal density function of $A$ under $P$. Additionally, we define the pseudo-outcome $\xi_{\mu, g, Q}(y, a, w)$ as
\[ \xi_{\mu, g, Q}(y, a, w) := \frac{y - \mu(a, w)}{g(a , w)}  + \int \mu(a, z) Q(dz) \ .\]
As noted by \cite{kennedy2016continuous},  $E_{0}\left[ \xi_{\mu, g, Q_0}(Y, A, W) \mid A = a\right] = \theta_0(a)$ if \emph{either} $\mu = \mu_0$ or $g = g_0$. They used this fact to motivate an estimator $\theta_{n,h}(a)$ of $\theta_0(a)$, defined as the local linear regression with bandwidth $h > 0$ of the pseudo-outcomes $\xi_{\mu_n, g_n, Q_n}(Y_1, A_1, W_1), \xi_{\mu_n, g_n, Q_n}(Y_2, A_2, W_2), \dotsc, \xi_{\mu_n, g_n, Q_n}(Y_n, A_n, W_n)$ on $A_1,A_2, \dotsc, A_n$, where $\mu_n$ is an estimator of $\mu_0$, $g_n$ is an estimator of $g_0$, and $Q_n$ is the empirical distribution function based on $W_1, W_2,\dotsc, W_n$. The study of this nonparametric regression problem is not standard because these pseudo-outcomes are dependent when the nuisance function estimators $\mu_n$ and $g_n$ are estimated from the data. Nevertheless, \cite{kennedy2016continuous} showed that their estimator is consistent if either $\mu_n$ or $g_n$ is consistent. Additionally, under regularity conditions, they showed that if both nuisance estimators converge fast enough and the bandwidth $h_n^*$ tends to zero at rate $n^{-1/5}$, then $n^{2/5}[\theta_{n,h_n^*}(a) - \theta_0(a)] \indist N(b_0(a), v_0(a))$, where $b_0(a)$ is an asymptotic bias depending on the second derivative of $\theta_0$, and $v_0(a)$ is an asymptotic variance.


In our setting, $\theta_0$ is known to be monotone. Therefore, instead of using a local linear regression to estimate the conditional mean of the pseudo-outcomes, it is natural to consider as an estimator the isotonic regression of the pseudo-outcomes on $A_1,A_2, \dotsc, A_n$. Using the GCM representation of isotonic regression stated in the previous section, we can summarize our estimation procedure as follows:
\begin{enumerate}
\item Construct estimators $\mu_n$ and $g_n$ of $\mu_0$ and $g_0$, respectively.
\item For each $a$ in the unique values of $A_1, A_2,\dotsc, A_n$, compute and set 
\begin{equation}\Gamma_{n}(a)  :=\frac{1}{n}\sum_{i=1}^nI_{(-\infty, a]}(A_i)\left[\frac{ Y_i - \mu_n(A_i,W_i)}{g_n(A_i, W_i)}\right] + \frac{1}{n^2} \sum_{i=1}^n \sum_{j=1}^n I_{(-\infty, a]}(A_i) \mu_n(A_i, W_j) \ . \label{gamma_one_step} \end{equation}
\item Compute the GCM $\overline\Psi_n$ of the set of points $\left\{ (0,0)\right\} \cup \left\{ \left(F_n(A_i), \Gamma_n(A_i) \right) : i = 1, 2, \dotsc, n \right\}$ over $[0,1]$.
\item Define $\theta_n(a)$ as the left derivative of $\overline\Psi_n$ evaluated at $F_n(a)$.
\end{enumerate}

As in the work of \cite{kennedy2016continuous}, while the proposed estimator $\theta_n$ can be defined as an isotonic regression, the asymptotic properties of our estimator do not appear to simply follow from classical results for isotonic regression because the pseudo-outcomes depend on the estimators $\mu_n$, $g_n$ and $Q_n$, which themselves depend on all the observations. However, $\theta_n$ is of generalized Grenander-type, and thus the asymptotic results of \cite{westling2018monotone} can be used to study its asymptotic properties. To see that $\theta_n$ is a generalized Grenander-type estimator, we define $\psi_P := \theta_P \circ  F_P^{-1}$ and note that since $\theta_P$ and $F_P^{-1}$ are increasing, so is $\psi_P$. Therefore, the primitive function $\Psi_P(t) := \int_0^t \psi_P(u) du = \int_{-\infty}^{F_P^{-1}(t)} \theta_P(v) F_P(dv)$ is convex. Next, we define $\Gamma_P := \Psi_P \circ F_P$, so that $\Gamma_P(a) = \int_{-\infty}^{a} \theta_{P} (u)F_P(du) = \iint_{-\infty}^{a} \mu_P(u, w) F_P(du) Q_P(dw)$. The parameter $\Gamma_P(a_0)$ is pathwise differentiable at $P$ in $\s{M}$ for each $a_0$, and its nonparametric efficient influence function $ \phi_{\mu_P, g_P, F_P, Q_P,a_0}^*$ can be computed to be
\begin{align*}
(y,a,w)\mapsto I_{(-\infty, a_0]}(a) \left[\frac{ y - \mu_P(a,w)}{g_P(a, w)}\right] +\int_{-\infty}^{a_0}  \mu_P(u, w) \, F_P(du)  + I_{(-\infty, a_0]}(a) \theta_{P}(a) - 2\Gamma_{P}(a_0)\ .
 \end{align*}
Denoting by $P_n$ any estimator of $P_0$ compatible with estimators $\mu_n$, $g_n$, $F_n$ and $Q_n$ of $\mu_0$, $g_0$, $F_0$ and $Q_0$, respectively,  the one-step estimator of $\Gamma_0(a)$ is given by $\Gamma_{n}(a) := \Gamma_{\mu_n, F_n, Q_n}(a) + \frac{1}{n}\sum_{i=1}^n\phi_{\mu_n, g_n, F_n, Q_n,a}^*(O_i)$, where we define $\Gamma_{\mu_n, F_n, Q_n}(a) :=  \iint_{-\infty}^{a} \mu_n(u, w) F_n(du) Q_n(dw)$. This one-step estimator is equivalent to that defined in \eqref{gamma_one_step}. We then define $\Psi_n := \Gamma_n \circ F_n^{-}$ for $F_n^{-}$ the empirical quantile function of $A$ as our estimator of $\Psi_0$, and $\psi_n$ as the left derivative of the GCM of $\Psi_n$. Thus, we find that $\theta_n = \psi_n \circ F_n$ is the estimator defined in steps 1--4. This form of the estimator was described in \cite{westling2018monotone}, where it was briefly discussed as one of several examples of a general strategy for nonparametric monotone inference.  

If $\theta_0(a)$ were only known to be monotone on a fixed sub-interval $\s{A}_0 \subset \s{A}$, we would define $F_P(a) := P(A \leq a \mid A \in \s{A}_0)$ as the marginal distribution function restricted to $\s{A}_0$, and $F_n$ as its empirical counterpart. Similarly, $I_{(-\infty, a]}(A_i)$ in \eqref{gamma_one_step} would be replaced with $I_{(-\infty, a]\cap \s{A}_0}(A_i)$. In all other respects, our estimation procedure would remain the same.

Finally, as alluded to earlier, we observe that the proposed estimator generalizes classical isotonic regression in a way we now make precise. If it is known that $A$ is independent of $W$ (Condition 1), so that $g_0(a ,w ) = 1$ for all supported $(a,w)$, we may take $g_n = 1$. If, furthermore, it is known that $Y$ is independent of $W$ given $A$ (Condition 2), then we may construct $\mu_n$ such that  $\mu_n(a, w) = \mu_n(a)$ for all supported $(a, w)$. Inserting $g_n = 1$ and any such $\mu_n$ into \eqref{gamma_one_step}, we obtain that $\Gamma_n(a)  =\frac{1}{n}\sum_{i=1}^n I_{(-\infty, a]}(A_i) Y_i$ and thus that $\theta_n(a) = r_n(a)$ for each $a$. Hence, in this case, our estimator reduces to least-squares isotonic regression.

\section{Theoretical properties}\label{theoretical}

\subsection{Invariance to strictly increasing exposure transformations}

An important feature of the proposed estimator is that, as with the isotonic regression estimator, it is invariant to any strictly increasing transformation of $A$. This is a desirable property because the scale of a continuous exposure is often arbitrary from a statistical perspective. For instance, if $A$ is temperature, whether $A$ is measured in degrees Fahrenheit,  Celsius or Kelvin does not change the information available. In particular, if the parameters $\theta_0$ and $\theta_0^*$ correspond to using as exposure $A$ and $H(A)$, respectively, for $H$ some strictly increasing transformation, then $\theta_0$ and $\theta_0^*$ encode exactly the same information about the effect of $A$ on $Y$ after adjusting for $W$. It is therefore natural to expect any sensible estimator to be invariant to the scale on which the exposure is measured. 

Setting $X := H(A)$ for a strictly  increasing  function $H : \s{A} \to \d{R}$, we first note that the function $\theta_0^*:x\mapsto E_0\left[ E_0\left(Y \mid  X= x, W\right)\right] = \theta_0\circ H^{-1}(x)$ is non-decreasing. Next, we define $\mu_0^*(x, w) := E_0\left(Y \mid X = x, W = w\right)$ and $g_0^*(x,w) =\pi_0^*(x\mid w)/f_0^*(x)$, where $\pi^*_0(x\mid w)$ is the evaluation at $x$ of the conditional density function of $X$ given $W=w$ and $f^*_0$ is the marginal density function of $X$ under $P_0$, and we denote by $\mu_n^*$ and $g_n^*$ estimators of $\mu_0^*$ and  $g_0^*$, respectively. The estimation procedure defined in the previous section but using exposure $X$ instead of $A$ then leads to estimator $\theta_n^*(x) :=\psi_n^* \circ F_n^*(x)$, where $F_n^* := F_n \circ H^{-1}$ is the empirical distribution function based on $X_1,X_2, \dotsc, X_n$, and $\psi_n^*$ is the left derivative of the GCM of $\Psi_n^* := \Gamma_n^* \circ F_n^{*-}$ for 
\begin{align*}
\Gamma_n^*(x)\ :=&\ \ \frac{1}{n}\sum_{i=1}^n\left\{I_{(-\infty, x]}(X_i) \left[\frac{ Y_i - \mu_n^*(X_i,W_i)}{g_n^*(X_i, W_i)}\right] + \int_{-\infty}^{x}\mu_n^*(x, W_i) \, F_n^*(dx)\right\} \\
 =&\ \ \frac{1}{n}\sum_{i=1}^n\left\{ I_{(-\infty, H^{-1}(x)]}(A_i)  \left[\frac{ Y_i - \mu_n^*(H(A_i),W_i)}{g_n^*(H(A_i), W_i)}\right] + \int_{-\infty}^{H^{-1}(x)} \mu_n^*(H(a), W_i) \, F_n(da)\right\}\ .
 \end{align*}
If it is the case that $\mu_n^*(H(a), w) = \mu_n(a, w)$ and $g_n^*(H(a), w) = g_n(a, w)$, implying that nuisance estimators $\mu_n$ and $g_n$ are themselves invariant to strictly increasing transformation of $A$, then we have that $\Gamma_n^* = \Gamma_n\circ H^{-1}$, and so, $\Psi_n^* = \Gamma_n \circ H^{-1} \circ H \circ F_n = \Psi_n$. It follows then that $\theta_n^* = \theta_n \circ H^{-1}$. In other words, the proposed estimator $\theta_n$ of $\theta_0$ is invariant to any strictly increasing transformation of the exposure variable.

We note that it is easy to ensure that $\mu_n^*(H(a), w) = \mu_n(a, w)$ and $g_n^*(H(a), w) = g_n(a, w)$. Set $U := F_n(A)$, which is also equal to $F_n^*(X)$, and let $\bar\mu_n(u,w)$ be an estimator of the conditional mean of $Y$ given $(U,W)=(u,w)$. Then, taking $\mu_n(a,w) := \bar\mu_n(F_n(a), w)$, we have that $\mu_n^*(x,w) := \bar\mu_n(F_n^*(x), w)$ satisfies the desired property. Similarly, letting $\bar{g}_n(u, w)$ be an estimator of the conditional density of $U = u$ given $W = w$, and setting $g_n(a, w) := \bar{g}_n(F_n(a), w)$, we may take  $g_n^*(x, w) := \bar{g}_n(F_n^*(x), w)$.

\subsection{Consistency}

We now provide sufficient conditions under which consistency of $\theta_n$ is guaranteed. Our conditions require controlling the uniform entropy of certain classes of functions. For a uniformly bounded class of functions $\s{F}$, a finite discrete probability measure $Q$, and any $\varepsilon > 0$, the $\varepsilon$-covering number $N(\varepsilon, \s{F}, L_2(Q))$ of $\s{F}$ relative to the $L_2(Q)$ metric is the smallest number of $L_2(Q)$-balls of radius less than or equal to $\varepsilon$ needed to cover $\s{F}$. The uniform $\varepsilon$-entropy of $\s{F}$ is then defined as $\log \sup_Q N(\varepsilon, \s{F}, L_2(Q))$, where the supremum is taken over all finite discrete probability measures. For a thorough treatment of covering numbers and their role in empirical process theory, we refer readers to \cite{van1996weak}.

Below, we state three sufficient conditions we will refer to in the following theorem.
\begin{description}[style=multiline,leftmargin=1cm]
\item[(A1)] There exist constants $C, \delta, K_{0},K_{1},K_{2}\in(0,+\infty)$ and $V\in[0, 2)$ such that, almost surely as $n \to \infty$, $\mu_n$ and $g_n$ are contained in classes of functions  $\s{F}_{0}$ and $\s{F}_{1}$, respectively, satisfying: 
\begin{enumerate}[(a)]
\item $|\mu| \leq K_{0}$ for all $\mu \in \s{F}_{0}$, and $K_{1} \leq g \leq K_{2}$ for all $g \in \s{F}_{1}$;
\item $\log \sup_Q N(\varepsilon, \s{F}_{0}, L_2(Q)) \leq C\varepsilon^{-V/2}$ and $\log \sup_Q N(\varepsilon, \s{F}_{1}, L_2(Q)) \leq C\varepsilon^{-V}$ for all $\varepsilon \leq \delta$.
 \end{enumerate}
\item[(A2)] There exist $\mu_{\infty} \in \s{F}_{0}$ and $g_{\infty} \in \s{F}_{1}$ such that $P_0(\mu_n - \mu_{\infty})^2 \inprob 0$ and $P_0(g_n - g_{\infty})^2  \inprob 0$.
\item[(A3)] There exist subsets $\s{S}_1, \s{S}_2$ and $\s{S}_3$ of $\s{A} \times\s{W}$ such that $P_0(\s{S}_1 \cup \s{S}_2 \cup \s{S}_3) = 1$ and:
\begin{enumerate}[(a)]
\item $\mu_{\infty}(a,w) = \mu_0(a,w)$ for all $(a,w) \in \s{S}_1$;
\item $g_{\infty}(a,w) = g_0(a,w)$ for all $(a,w) \in \s{S}_2$;
\item $\mu_{\infty}(a,w) = \mu_0(a,w)$ and $g_{\infty}(a,w) = g_0(a,w)$ for all $(a,w) \in \s{S}_3$.
\end{enumerate}
\end{description}
Under these three conditions, we have the following result.
\begin{thm}[Consistency]
If conditions (A1)--(A3) hold, then $\theta_n(a) \inprob \theta_0(a)$ for any value $a \in \s{A}$ such that $F_0(a) \in (0,1)$, $\theta_0$ is continuous at $a$, and $F_0$ is strictly increasing in a neighborhood of $a$. If $\theta_0$ is uniformly continuous and $F_0$ is strictly increasing on $\s{A}$, then $\sup_{a \in \s{A}_0} | \theta_n(a) - \theta_0(a) | \inprob 0$ for any bounded strict subinterval $\s{A}_0\subsetneq\s{A}$.
\label{thm:consistency}\end{thm}

We note that in the pointwise statement of Theorem~\ref{thm:consistency}, $F_0(a)$ is required to be in the interior of $[0,1]$, and similarly, the uniform statement of Theorem~\ref{thm:consistency} only covers strict subintervals of $\s{A}$. This is due to the well-known boundary issues with Grenander-type estimators. Various remedies have been proposed in particular settings, and it would be interesting to consider these in future work (see, e.g., \citealp{woodroofe1993penalized, balabdaoui2011grenander, kulikov2006}).

Condition (A1) requires that $\mu_n$ and $g_n$ eventually be contained in uniformly bounded function classes that are small enough for certain empirical process terms to be controlled. This condition is easily satisfied if, for instance, $\s{F}_{0}$ and $\s{F}_{1}$ are parametric classes. It is also satisfied for many infinite-dimensional function classes. Uniform entropy bounds for many such classes may be found in Chapter 2.6 of \cite{van1996weak}. We note that there is an asymmetry between the entropy requirements for $\s{F}_{0}$ and $\s{F}_{1}$ in part (b) of (A1). This is due to the term $\iint_{-\infty}^{a} \mu_n(u, w) F_n(du) \, Q_n(dw)$ appearing in $\Gamma_n(a)$. To control this term, we use an upper bound of the form $\int_0^1 \log \sup_Q N(\varepsilon, \s{F}_{0}, L_2(Q)) d\varepsilon$ from the theory of empirical $U$-processes \citep{nolan1987uprocess} -- this contrasts with the uniform entropy integral $\int_0^1 [\log \sup_Q N(\varepsilon, \s{F}, L_2(Q))]^{1/2} d\varepsilon$ that bounds ordinary empirical processes indexed by a uniformly bounded class $\s{F}$. In Section~\ref{sec:cv}, we consider the use of cross-fitting to avoid the entropy conditions in (A1).

Condition (A2) requires that $\mu_n$ and $g_n$ tend to limit functions $\mu_{\infty}$ and $g_{\infty}$, and condition (A3) requires that either $\mu_{\infty}(a,w) = \mu_0(a,w)$ or $g_{\infty}(a,w) = g_0(a,w)$ for $(F_0 \times Q_0)$-almost every $(a,w)$. If either (i) $\s{S}_1$ and $\s{S}_3$ are null sets or (ii) $\s{S}_2$ and $\s{S}_3$ are null sets, then condition (A3) is known simply as \emph{double-robustness} of the estimator $\theta_n$ relative to the nuisance functions $\mu_0$ and $g_0$: $\theta_n$ is consistent as long as $\mu_\infty=\mu_0$ or $g_\infty=g_0$. Doubly-robust estimators are at this point a mainstay of causal inference and have been studied for over two decades (see, e.g., \citealp{robins1994estimation, rotnitzky1998semiparametric, scharfstein1999adjusting, van2003unified, neugebauer2005prefer, bang2005doubly}). However, (A3) is more general than classical double-robustness, as it allows neither $\mu_n$ nor $g_n$ to tend to their true counterparts over the whole domain, as long as at least one of $\mu_n$ or $g_n$ tends to the truth for almost every point in the domain.

\subsection{Convergence in distribution}

We now study the convergence in distribution of $n^{1/3}[\theta_n(a) - \theta_0(a)]$ for fixed $a$. We first define for any square-integrable functions $h_1, h_2 : \s{A} \times \s{W} \to \d{R}$, $\varepsilon >0$ and $\s{S} \subseteq \s{A} \times \s{W}$ the pseudo-distance
\begin{equation}d(h_1, h_2; a, \varepsilon, \s{S}) := \left[\sup_{|u - a| \leq \varepsilon} E_{0} \left\{ I_{\s{S}}(u ,W) \left[ h_1(u, W) - h_2(u, W)\right]^2 \right\}\right]^{1/2}.\label{eq:discrep}\end{equation}
We also denote by $\sigma_0^2(a, w)$ the conditional variance $E_{0} \left\{ \left[Y - \mu_0(A, W)\right]^2 \middle| A = a, W =w\right\}$ of $Y$ given $A = a$ and $W = w$ under $P_0$. Below, we will refer to these two additional conditions:
\begin{description}[style=multiline,leftmargin=1cm]
\item[(A4)] There exists $\varepsilon_0 > 0$ such that: 
\begin{enumerate}[(a)]
\item  $\max\{d(\mu_n, \mu_{\infty};a, \varepsilon_0, \s{S}_1), d(g_n, g_{\infty};a, \varepsilon_0, \s{S}_2)\} = \fasterthan(n^{-1/3})$; 
\item $\max\{d(\mu_n, \mu_{\infty};a, \varepsilon_0, \s{S}_2), d(g_n, g_{\infty};a, \varepsilon_0, \s{S}_1)\} = \fasterthan(1)$;
\item $ d(\mu_n, \mu_{\infty};a, \varepsilon_0, \s{S}_3)d(g_n, g_{\infty};a, \varepsilon_0, \s{S}_3)  = \fasterthan(n^{-1/3})$.
\end{enumerate}
\item[(A5)] $F_0, \mu_0, \mu_{\infty}, g_0, g_{\infty}$ and $\sigma_0^2$ are continuously differentiable in a neighborhood of $a$ uniformly over $w \in \s{W}$.
\end{description} Under conditions introduced so far, we have the following distributional result.
\begin{thm}[Convergence in distribution]
 If conditions (A1)--(A5) hold, then
 \[ n^{1/3}\left[ \theta_{n}(a) - \theta_0(a)\right] \indist \left[\frac{4 \theta_0'(a) \kappa_0(a)}{f_0(a)}\right]^{1/3} \d{W}\ ,\] for any $a \in \s{A}$ such that $F_0(a)\in(0,1)$, where $\d{W}$ follows the standard Chernoff distribution and
 \[\kappa_0(a) := E_{0} \left\{ E_{0} \left[ \left\{ \left[\frac{Y - \mu_{\infty}(a,W)}{g_{\infty}(a,W)}\right]  + \theta_{\infty}(a) - \theta_0(a)\right\}^2 \middle| A = a, W\right] g_0(a, W)\right\}\] with $\theta_\infty(a)$ denoting $\int \mu_\infty(a,w)Q_0(dw)$. 
\label{thm:dose_response}
\end{thm} 

We note that the limit distribution in Theorem~\ref{thm:dose_response} is the same as that of the standard isotonic regression estimator up to a scale factor. As noted above, when either (i) $Y$ and $W$ are independent given $A$ or (ii) $A$ is independent of $W$, the functions $\theta_0$ and $r_0$ coincide. As such, we can directly compare the respective limit distributions of $n^{1/3}\left[\theta_n(a)-\theta_0(a)\right]$ and $n^{1/3}\left[r_n(a)-r_0(a)\right]$ under these conditions. When both $\mu_{\infty} = \mu_0$ and $g_\infty = g_0$, $r_n(a)$ is asymptotically more concentrated than $\theta_n(a)$ in scenario (i), and less concentrated  in scenario (ii). This is analogous to findings in linear regression, where including a covariate uncorrelated with the outcome inflates the standard error of the estimator of the coefficient corresponding to the exposure, while including a covariate correlated with the outcome but uncorrelated with the exposure deflates its standard error.


 Condition (A4) requires that, on the set $\s{S}_1$ where $\mu_n$ is consistent but $g_n$ is not, $\mu_n$ converges faster than $n^{-1/3}$ uniformly in a neighborhood of $a$, and similarly for $g_n$ on the set $\s{S}_2$. On the set $\s{S}_3$ where both $\mu_n$ and $g_n$ are consistent, only the product of their rates of convergence must be faster than $n^{-1/3}$. Hence, a non-degenerate limit theory is available as long as at least one of the nuisance estimators is consistent at a rate faster than $n^{-1/3}$, even if the other nuisance estimator is inconsistent. This suggests the possibility of performing doubly-robust inference for $\theta_0(a)$, that is, of constructing confidence intervals and tests based on $\theta_n(a)$ with valid calibration even when one of $\mu_0$ and $g_0$ is inconsistently estimated. This is explored in Section \ref{inference}. Finally, as in Theorem~\ref{thm:consistency}, we allow that neither $\mu_n$ nor $g_n$ be consistent everywhere, as long as for $(F_0 \times Q_0)$-almost every $(a,w)$ at least one of $\mu_n$ or $g_n$ is consistent.

We remark that if it is known that $\mu_n(a,\cdot)$ is consistent for $\mu_0(a,\cdot)$ in an $L_2(Q_0)$ sense at rate faster than $n^{-1/3}$, the isotonic regression of the plug-in estimator $\theta_{\mu_n}(a):=\int \mu_n(a,w)Q_n(dw)$ -- which can be equivalently obtained by setting $g_n(a,\cdot)=+\infty$ in the construction of $\theta_n(a)$ -- achieves a faster rate of convergence to $\theta_0(a)$ than does $\theta_n(a)$. This might motivate an analyst to use $\theta_{\mu_n}(a)$ rather than $\theta_n(a)$ in such a scenario. However, the consistency of $\theta_{\mu_n}(a)$ hinges entirely on the fact that $\mu_\infty=\mu_0$, and in particular, $\theta_{\mu_n}(a)$ will be inconsistent if $\mu_\infty\neq\mu_0$, even if $g_\infty=g_0$. Additionally, the estimator $\theta_{\mu_n}(a)$ may not generally admit a tractable limit theory upon which to base the construction of valid confidence intervals, particularly when machine learning methods are used to build $\mu_n$.

\subsection{Grenander-type estimation without domain transformation}

As indicated earlier, the isotonic regression estimator based on estimated pseudo-outcomes coincides with a generalized Grenander-type estimator for which the marginal exposure empirical distribution function is used as domain transformation. An alternative estimator could be constructed via Grenander-type estimation without the use of any domain transformation. Specifically, we let $a_-, a_+ \in \d{R}$ be fixed, and we define $\Theta_0(a) = \int_{a_-}^{a} \theta_0(u) du$. Under regularity conditions, for $a \leq a_+$, the one-step estimator of $\Theta_0(a)$ given by \[\Theta_n(a):=\frac{1}{n}\sum_{i=1}^n\left\{I_{(a_-, a]}(A_i)\left[\frac{ Y_i - \mu_n(A_i,W_i)}{\pi_n(A_i, W_i)}\right] + \int_{a_-}^a \mu_n(u, W_i)du\right\}\] is asymptotically efficient, where $\pi_n$ is an estimator of $\pi_0$, the conditional density of $A$ given $W$ under $P_0$. The left derivative of the GCM of $\Theta_n$ over $[a_-, a_+]$ defines an alternative estimator $\bar\theta_n(a)$.

It is natural to ask how $\bar\theta_n$ compares to the estimator $\theta_n$ we have studied thus far. First, we note that, unlike $\theta_n$, $\bar\theta_n$ neither generalizes the classical isotonic regression estimator nor is invariant to strictly increasing transformations of $A$. Additionally, utilizing the transformation $F_0$ fixes $[0,1]$ as the interval over which the GCM should be performed. If $\s{A}$ is known to be a bounded set, $[a_-, a_+]$ can be taken as the endpoints of $\s{A}$, but otherwise  the domain $[a_-, a_+]$ must be chosen in defining $\bar\theta_n$. Turning to an asymptotic analysis, using the results of \cite{westling2018monotone}, it is possible to establish conditions akin to (A1)--(A5) under which $n^{1/3}\left[\bar\theta_n(a) - \theta_0(a)\right] \indist \left[4 \theta_0'(a) \bar\kappa_0(a)\right]^{1/3} \d{W}$ with scale parameter 
 \[\bar\kappa_0(a):= E_{0} \left[ E_{0} \left\{  \left[\frac{Y - \mu_{\infty}(A,W)}{\pi_{\infty}(A \mid W)}\right]^2 \middle| A = a, W\right\} \pi_0(a \mid W)\right] ,\] 
 where $\pi_\infty$ is the limit of $\pi_n$ in probability. We denote by $[4\tau_0(a)]^{1/3}$ and $[4\bar\tau_0(a)]^{1/3}$ the limit scaling factors of $n^{1/3}\left[\theta_n(a) - \theta_0(a)\right]$ and  $n^{1/3}\left[\bar\theta_n(a) - \theta_0(a)\right]$, respectively. If $g_{\infty}=\pi_{\infty}/f_0$ and $\mu_{\infty} = \mu_0$, then $\tau_0(a) = \bar\tau_0(a)$, and $n^{1/3}\left[\theta_n(a) - \theta_0(a)\right]$  and $n^{1/3}\left[\bar\theta_n(a) - \theta_0(a)\right]$ have the same limit distribution. If instead $g_{\infty} = \pi_{\infty} / f_0 = g_0$ but $\mu_{\infty} \neq \mu_0$, this is no longer the case. In fact, we can show that
\begin{align*}
\tau_0(a)\ &=\ \theta_0'(a)E_{0} \left[ \frac{E_0\{[Y-\mu_\infty(a,W)]^2\mid A=a,W\}}{\pi_0(a \mid W)}\right] - \theta_0'(a) \frac{\{\theta_\infty(a)- \theta_0(a)\}^2}{f_0(a)}\\
&\leq\  \theta_0'(a) E_{0} \left[ \frac{E_0\{[Y-\mu_\infty(a,W)]^2\mid A=a,W\}}{\pi_0(a \mid W)}\right]\ =\ \bar{\tau}(a)\ .
\end{align*}
Hence, when the outcome regression estimator $\mu_n$ is inconsistent, gains in efficiency are achieved by utilizing the transformation, and the relative gain in efficiency is directly related to the amount of asymptotic bias in the estimation of $\mu_{0}$.

\subsection{Discrete domains}

In some circumstances, the exposure $A$ is discrete rather than continuous. Our estimator works equally well in these cases, since, as we highlight below, it turns out to then be asymptotically equivalent to the well-studied augmented IPW (AIPW) estimator. As a result, the large-sample properties of our estimator can be derived from the large-sample properties of the AIPW estimator, and asymptotically valid inference can be obtained using standard influence function-based techniques.

Suppose that $\s{A} = \{a_1 <  a_2 < \cdots < a_m\}$ and $f_{0,j} := P_0(A = a_j) > 0$ for all $j \in \{1, 2, \dotsc, m\}$ and $\sum_{j=1}^m f_{0,j} = 1$. Our estimation procedure remains the same with one exception: in defining $g_0 := \pi_0 / f_0$, we now take $\pi_0$ to be the conditional probability $\pi_0(a_j \mid w) := P_0(A = a_j \mid W = w)$ rather than the corresponding conditional density, and we take $f_0$ as the marginal probability $f_0(a_j) := P_0( A= a_j) = f_{0,j}$ rather than the corresponding marginal density. We then set $g_{n} := \pi_n / f_n$ as the estimator of $g_0$, where $\pi_n$ is any estimator of $\pi_0$ and $f_{n}(a_j) := n_j / n$ for $n_j := \sum_{i=1}^n I(A_i = a_j)$.  In all other respects, our estimation procedure is identical to that defined previously. With these definitions, we denote by $\xi_{n,i}$ the estimated pseudo-outcome for observation $i$. Our estimator is then the isotonic regression of $\xi_{n,1},\xi_{n,2}, \dotsc, \xi_{n,n}$ on $A_1,A_2,  \dotsc, A_n$. However, since for each $i$ there is a unique $j$ such that $A_i = a_j$, this is equivalent to performing isotonic regression of $\theta_n^{\dagger}(a_1),\theta_n^\dagger(a_2), \dotsc, \theta_n^{\dagger}(a_m)$ on $a_1, a_2, \dotsc, a_m$, where $\theta_n^{\dagger}(a_j) := n_j^{-1}\sum_{i=1}^n I_{\{a_j\}}(A_i) \xi_{n,i}$. It is straightforward to see that \[\theta_n^{\dagger}(a_j) = \frac{1}{n} \sum_{i=1}^n \left\{I_{\{a_j\}}(A_i) \left[\frac{Y_i - \mu_n(a_j, W_i)}{\pi_n(a_j \mid W_i)} \right]+ \mu_n(a_j, W_i) \right\}\ ,\] which is exactly the AIPW estimator of $\theta_0(a_j)$. Therefore, in this case, our estimator reduces to the isotonic regression of the classical AIPW estimator constructed separately for each element of the exposure domain. 

The large-sample properties of $\theta_n^{\dagger}$, including doubly-robust consistency and convergence in distribution at the regular parametric rate $n^{-1/2}$, are well-established \citep{robins1994estimation}. Therefore, many properties of $\theta_n$ in this case can be determined using the results of \cite{westling2018projection}, which studied the behavior of the isotonic correction of an initial estimator. In particular, $\max_{a \in \s{A}}| \theta_n(a) - \theta_0(a)| \leq \max_{a \in \s{A}}| \theta_n^{\dagger}(a) - \theta_0(a)|$ as long as $\theta_0$ is non-decreasing on $\s{A}$. Uniform consistency of $\theta_n^{\dagger}$ over $\s{A}$ thus implies uniform consistency of $\theta_n$. Furthermore, if $\theta_0$ is strictly increasing on $\s{A}$ and $\{n^{1/2}[ \theta_n^{\dagger}(a) - \theta_0(a)] : a \in \s{A} \}$ converges in distribution, then $\max_{a \in \s{A}}| \theta_n(a) - \theta_n^{\dagger}(a)| = \fasterthan\left(n^{-1/2}\right)$, so that large-sample standard errors for $\theta_n^\dagger$ are also valid for $\theta_n$. If $\theta_0$ is not strictly increasing on $\s{A}$ but instead has flat regions, then $\theta_n$ is more efficient than $\theta_0$ on these regions, and confidence intervals centered around $\theta_n$ but based upon the limit theory for $\theta_n^\dagger$ will be conservative.

\subsection{Large-sample results for causal effects}

In many applications, in addition to the causal dose response curve $a \mapsto m_0(a)$ itself, causal effects of the form $(a_1, a_2) \mapsto m_0(a_1) - m_0(a_2)$ are of scientific interest as well. Under the identification conditions discussed in Section~\ref{sec:param} applied to each of $a_1$ and $a_2$, such causal effects are identified with the observed-data parameter $\theta_0(a_1) - \theta_0(a_2)$. A natural estimator for such a causal effect in our setting is $\theta_n(a_1) - \theta_n(a_2)$. If the conditions of Theorem~\ref{thm:consistency} hold for both $a_1$ and $a_2$, then the continuous mapping theorem implies that $\theta_n(a_1) - \theta_n(a_2) \inprob \theta_0(a_1) - \theta_0(a_2)$. However, since Theorem~\ref{thm:dose_response} only provides marginal distributional results, and thus does not describe the joint convergence of $Z_n(a_1, a_2) := (n^{1/3}[\theta_n(a_1) - \theta_0(a_1)],n^{1/3}[\theta_n(a_2) - \theta_0(a_2)])$, it cannot be used to determine the large-sample behavior of $n^{1/3}\left\{ \left[ \theta_n(a_1) - \theta_n(a_2) \right] - \left[ \theta_0(a_1) - \theta_0(a_2) \right] \right\}$. The following result demonstrates that such joint convergence can be expected under the aforementioned conditions, and that the bivariate limit distribution of $Z_n(a_1,a_2)$ has independent components.
\begin{thm}[Joint convergence in distribution]
 If conditions (A1)--(A5) hold for $a \in \{a_1, a_2\} \subset \s{A}$ and $F_0(a_1), F_0(a_2) \in (0,1)$, then $Z_n(a_1, a_2)$ converges in distribution to $(\left[4\tau_0(a_1) \right]^{1/3} \d{W}_1,  \left[4\tau_0(a_2)\right]^{1/3} \d{W}_2)$, where $\d{W}_1$ and $\d{W}_2$ are independent standard Chernoff distributions and the scale parameter $\tau_0$ is as defined in Theorem~\ref{thm:dose_response}.
\label{thm:joint_conv}
\end{thm} 
Theorem~\ref{thm:joint_conv} implies that, under the stated conditions, $n^{1/3}\left\{ \left[ \theta_n(a_1) - \theta_n(a_2) \right] - \left[ \theta_0(a_1) - \theta_0(a_2) \right] \right\}$ converges in distribution to $\left[4\tau_0(a_1) \right]^{1/3} \d{W}_1 - \left[4\tau_0(a_2)\right]^{1/3} \d{W}_2$.

\subsection{Use of cross-fitting to avoid empirical process conditions}\label{sec:cv}

Theorems~\ref{thm:consistency} and~\ref{thm:dose_response} reveal that the statistical properties of $\theta_n$ depend on the nuisance estimators $\mu_n$ and $g_n$ in two important ways. First, we require in condition (A1) that $\mu_n$ or $g_n$ fall in small enough classes of functions, as measured by metric entropy, in order to control certain empirical process remainder terms. Second, we require in conditions (A2)--(A3) that at least one of $\mu_n$ or $g_n$ be consistent almost everywhere (for consistency), and in condition (A4)  that the product of their rates of convergence be faster than $n^{-1/3}$ (for convergence in distribution). In observational studies, researchers can rarely specify a priori correct parametric models for $\mu_0$ and $g_0$. This motivates use of data-adaptive estimators of these nuisance functions in order to meet the second requirement. However, such estimators often lead to violations of the first requirement, or it may be onerous to determine that they do not. Thus, because it may be difficult to find nuisance estimators that are both data-adaptive enough to meet required rates of convergence and fall in small enough function classes to make empirical process terms negligible, simultaneously satisfying these two requirements can be challenging in practice.

In the context of asymptotically linear estimators, it has been noted that cross-fitting nuisance estimators can resolve this challenge by eliminating empirical process conditions \citep{zheng2011cvtmle, belloni2018uniform, kennedy2019incremental}. We therefore propose employing cross-fitting of $\mu_n$ and $g_n$ in the estimation of $\Gamma_0$ in order to avoid entropy conditions in Theorems~\ref{thm:consistency} and~\ref{thm:dose_response}. Specifically, we fix $V \in \{2, 3, \dotsc, n/2\}$ and suppose that the indices $\{1, 2, \dotsc, n\}$ are randomly partitioned into $V$ sets $\s{V}_{n,1}, \s{V}_{n,2}, \dotsc, \s{V}_{n,V}$. We assume for convenience that $N := n / V$ is an integer and that $|\s{V}_{n,v}| = N$ for each $v$, but all of our results hold as long as $\max_v n / |\s{V}_{n,v}| = \bounded(1)$. For each $v \in \{1, 2, \dotsc, V\}$, we define $\s{T}_{n,v} := \{ O_i : i \notin \s{V}_{n,v}\}$ as the \emph{training set} for fold $v$, and denote by $\mu_{n,v}$ and $g_{n,v}$ the nuisance estimators constructed using only the observations from $\s{T}_{n,v}$. We then define pointwise the cross-fitted estimator $\Gamma_n^\circ$ of $\Gamma_0$ as 
\begin{equation}\Gamma_{n}^\circ(a)  :=\frac{1}{V} \sum_{v=1}^V \left\{ \frac{1}{N}\sum_{i \in \s{V}_{n,v}} I_{(-\infty, a]}(A_i) \left[\frac{ Y_i - \mu_{n,v}(A_i,W_i)}{g_{n,v}(A_i, W_i)} \right]+ \frac{1}{N^2} \sum_{i, j \in \s{V}_{n,v}} I_{(-\infty, a]}(A_i) \mu_{n,v}(A_i, W_j) \right\}\ . \label{gamma_one_step_cv} \end{equation}
Finally, the cross-fitted estimator $\theta_n^\circ$ of $\theta_0$ is constructed using steps 1--4 outlined in Section~\ref{sec:defn}, with $\Gamma_n$ replaced by $\Gamma_n^\circ$.

As we now demonstrate, utilizing the cross-fitted estimator $\theta_n^\circ$ allows us to avoid the empirical process condition (A1b). We first introduce the following two conditions, which are analogues of conditions (A1) and (A2).
\begin{description}[style=multiline,leftmargin=1cm]
\item[(B1)] There exist constants $C', \delta', K_{0}',K_{1}',K_{2}', K_3'\in(0,+\infty)$ such that, almost surely as $n \to \infty$ and for all $v$, $\mu_{n,v}$ and $g_{n,v}$ are contained in classes of functions $\s{F}_0'$ and $\s{F}_1'$, respectively, satisfying:
\begin{enumerate}[(a)]
\item $|\mu| \leq K_{0}'$ for all $\mu \in \s{F}_{0}'$, and $K_{1}' \leq g \leq K_{2}'$ for all $g \in \s{F}_{1}'$;
\end{enumerate}
and $\sigma_0^2(a,w) \leq K_3'$ for almost all $a,w$.
\item[(B2)] There exist $\mu_{\infty} \in \s{F}_0'$ and $g_{\infty} \in \s{F}_1'$ such that $\max_v P_0(\mu_{n,v} - \mu_{\infty})^2 \inprob 0$ and $\max_v P_0 (g_{n,v} - g_{\infty})^2  \inprob 0$.
\end{description}
We then have the following analogue of Theorem~\ref{thm:consistency} establishing consistency of the cross-fitted estimator $\theta_n^\circ$.
\begin{thm}[Consistency of the cross-fitted estimator]
If conditions (B1)--(B2) and (A3) hold, then $\theta_n^\circ(a) \inprob \theta_0(a)$ for any $a \in \s{A}$ such that $F_0(a) \in (0,1)$, $\theta_0$ is continuous at $a$, and $F_0$ is strictly increasing in a neighborhood of $a$. If $\theta_0$ is uniformly continuous and $F_0$ is strictly increasing on $\s{A}$, then $\sup_{a \in \s{A}_0} | \theta_n^\circ(a) - \theta_0(a) | \inprob 0$ for any bounded strict subinterval $\s{A}_0\subsetneq\s{A}$.
\label{thm:consistency_cv}\end{thm}
For convergence in distribution, we introduce the following analogue of condition (A4).
\begin{description}[style=multiline,leftmargin=1cm]
\item[(B4)] There exists $\varepsilon_0 > 0$ such that: 
\begin{enumerate}[(a)]
\item  $\max_v \max\{d(\mu_{n,v}, \mu_{\infty};a, \varepsilon_0, \s{S}_1), d(g_{n,v}, g_{\infty};a, \varepsilon_0, \s{S}_2)\} = \fasterthan(n^{-1/3})$; 
\item $\max_v \max\{d(\mu_{n,v}, \mu_{\infty};a, \varepsilon_0, \s{S}_2), d(g_{n,v}, g_{\infty};a, \varepsilon_0, \s{S}_1)\} = \fasterthan(1)$;
\item $ \max_v d(\mu_{n,v}, \mu_{\infty};a, \varepsilon_0, \s{S}_3)d(g_{n,v}, g_{\infty};a, \varepsilon_0, \s{S}_3)  = \fasterthan(n^{-1/3})$.
\end{enumerate}
\end{description}
We then have the following analogue of Theorem~\ref{thm:dose_response} for the cross-fitted estimator $\theta_n^\circ$.
\begin{thm}[Convergence in distribution for the cross-fitted estimator]
 If conditions (B1), (B2), (A3), (B4), and (A5) hold, then  $n^{1/3}\left[ \theta_{n}^\circ(a) - \theta_0(a)\right] \indist \left[4\tau_0(a)\right]^{1/3} \d{W}$ for any $a \in \s{A}$ such that $F_0(a)\in(0,1)$.
\label{thm:dose_response_cv}
\end{thm}
The conditions of Theorems~\ref{thm:consistency_cv} and~\ref{thm:dose_response_cv} are analogous to those of Theorems~\ref{thm:consistency} and~\ref{thm:dose_response}, with the important exception that the entropy condition (A1b) is no longer required. Therefore, the estimators $\mu_{n,v}$ and $g_{n,v}$ may be as data-adaptive as one desires without concern for empirical process terms, as long as they satisfy the boundedness conditions stated in (B1).

\section{Construction of confidence intervals}\label{inference}

\subsection{Wald-type confidence intervals}\label{sec:ci_effect}

The distributional results of Theorem~\ref{thm:dose_response} can be used to construct a confidence interval for $\theta_0(a)$. Since the limit distribution of $n^{1/3}\left[\theta_n(a)-\theta_0(a)\right]$ is symmetric around zero, a Wald-type construction seems  appropriate. Specifically, writing $\tau_0(a):= \theta_0'(a) \kappa_0(a) / f_0(a)$ and denoting by $\tau_n(a)$ any consistent estimator of $\tau_0(a)$, a Wald-type $1 - \alpha$ level asymptotic confidence interval for $\theta_0(a)$ is given by
\[ \left[ \theta_n(a) - \left[\frac{4 \tau_n(a)}{n}\right]^{1/3} q_{1-\alpha / 2}, \ \theta_n(a)+ \left[\frac{4 \tau_n(a)}{n}\right]^{1/3} q_{1-\alpha / 2} \right], \]
where $q_p$ denotes the $p^{th}$ quantile of $\d{W}$. Quantiles of the standard Chernoff distribution have been numerically computed and tabulated on a fine grid \citep{groeneboom2001jcgs}, and are readily available in the statistical programming language \texttt{R}. Estimation of $\tau_0(a)$ involves, either directly or indirectly, estimation of $\theta_0'(a)/f_0(a)$ and $\kappa_0(a)$. We focus first on the former.

We note that $\theta_0'(a) / f_0(a)= \psi_0'(F_0(a))$ with $\psi_0 := \theta_0\circ F_0^{-1}$.  This suggests that we could either estimate $\theta_0'$ and $f_0$ separately and consider the ratio of these estimators, or that we could instead estimate $\psi_0'$ directly and compose it with the estimator of $F_0$ already available. The latter approach has the desirable property that the resulting scale estimator is invariant to strictly monotone transformations of the exposure. As such, this is the strategy we favor. To estimate $\psi_0'$, we recall that the estimator $\psi_n$ from Section~\ref{estimator} is a step function and is therefore not differentiable. A natural solution consists of computing the derivative of a smoothed version of $\psi_n$. We have found local quadratic kernel smoothing of points $\{ (u_j, \psi_n(u_j)): j = 1, 2,\dotsc K\}$, for $u_j$ the midpoints of the jump points of $\psi_n$, to work well in practice.

Theorem~\ref{thm:joint_conv} can  be used to construct Wald-type confidence intervals for causal effects of the form $\theta_0(a_1) -\theta_0(a_2)$. We first construct estimates $\tau_{n}(a_1)$ and $\tau_n(a_2)$ of the scale parameters $\tau_0(a_1)$ and $\tau_0(a_2)$, respectively, and then compute an approximation $\bar{q}_{n,1-\alpha/2}$ of the $(1-\alpha/2)$-quantile of $\left[4\tau_n(a_1) \right]^{1/3} \d{W}_1 - \left[4\tau_n(a_2)\right]^{1/3} \d{W}_2$, where $\d{W}_1$ and $\d{W}_2$ are independent Chernoff distributions, using Monte Carlo simulations, for example. An asymptotic $1-\alpha$ level Wald-type confidence interval for $\theta_0(a_1) - \theta_0(a_2)$ is then $\theta_n(a_1) - \theta_n(a_2) \pm \bar{q}_{n,1-\alpha/2}n^{-1/3}$.

In the next two subsections, we discuss different strategies for estimating the scale factor $\kappa_0(a)$.

\subsection{Scale estimation relying on consistent nuisance estimation}

We first consider settings in which both $\mu_n$ and $g_n$ are consistent estimators, that is, $g_{\infty} = g_0$ and $\mu_\infty = \mu_0$. In such cases, we have that $\kappa_0(a) = E_{0} \left[ \sigma_0^2(a, W) / g_0(a, W)\right]$ with $\sigma^2_0(a,w)$ denoting the conditional variance $E_0\{[Y-\mu_0(a,W)]^2\mid A=a,W=w\}$. Any regression technique could be used to estimate the conditional expectation of $Z_n:=[Y - \mu_n(A,W)]^2$ given $A$ and $W$, yielding an estimator $\sigma_n^2(a,w)$ of $\sigma_0^2(a,w)$. A plug-in estimator of $\kappa_0(a)$ is then given by
\[ \kappa_{n}(a) := \frac{1}{n} \sum_{i=1}^n \frac{\sigma_n^2(a, W_i)}{g_n(a, W_i)}\ .\]
Provided $\mu_n$, $g_n$ and $\sigma_n^2$ are consistent estimators of $\mu_0$, $g_0$ and $\sigma_0^2$, respectively, $\kappa_{n}(a)$ is a consistent estimator of $\kappa_0(a)$. We note that in the special case of a binary outcome, the fact that $\sigma_0^2(a,w) = \mu_0(a,w)[1-\mu_0(a,w)]$ motivates the use of $\mu_n(a,w) [1 - \mu_n(a,w)]$ as estimator $\sigma_n^2(a,w)$, and thus eliminates the need for further regression beyond the construction of $\mu_n$ and $g_n$. In practice,  we typically recommend the use of an ensemble method -- for example, the SuperLearner \citep{vanderlaan2007super} -- to combine a variety of regression techniques, including machine learning techniques, to minimize the risk of inconsistency of $\mu_n$, $g_n$ and $\sigma_n^2$. 

\subsection{Doubly-robust scale estimation}

As noted above, Theorem~\ref{thm:dose_response} provides the limit distribution of $n^{1/3}\left[\theta_n(a)-\theta_0(a)\right]$ even if one of the nuisance estimators is inconsistent, as long as the consistent nuisance estimator converges fast enough. We now show how we may capitalize on this result to provide a doubly-robust estimator of $\kappa_0(a)$. Since $\psi_n$ is itself a doubly-robust estimator of $\psi_0$, so will be the proposed estimator $\psi_n'$ of $\psi_0'$ and hence also of the resulting estimator $\tau_n(a)$ of $\tau_0(a)$. This contrasts with the estimator of $\kappa_0(a)$ described in the previous section, which required the consistency of both $\mu_n$ and $g_n$.

To construct an estimator of $\kappa_0(a)$ consistent even if either $\mu_{\infty}\neq \mu_0$ or $g_{\infty}\neq g_0$, we begin by noting that $\kappa_0(a) = \lim_{h \downarrow 0} E_{0} \left[ K_h\left( F_0(A) - F_0(a)\right) \eta_\infty(Y, A, W) \right]$, where $K_h:u\mapsto h^{-1}K(uh^{-1})$ for some bounded density function $K$ with bounded support, and we have defined
\[ \eta_\infty:(y, a, w)\mapsto \left[ \frac{y - \mu_{\infty}(a, w)}{g_{\infty}(a, w)} + \theta_{\infty}(a) - \theta_0(a)\right]^2.\] Setting $\theta_{\mu_n}(a) := \int \mu_n(a, w) Q_n(dw)$ with $Q_n$ the empirical distribution based on $W_1,W_2,\ldots,W_n$, we define $\kappa^*_{n,h}(a) := \frac{1}{n}\sum_{i=1}^n K_h\left( F_n(A_i) - F_n(a)\right)\eta_n(Y_i, A_i, W_i)$ with $\eta_n$ obtained by substituting $\mu_\infty$, $g_\infty$, $\theta_\infty$ and $\theta_0$ by $\mu_n$, $g_n$, $\theta_{\mu_n}$ and $\theta_n$, respectively, in the definition of $\eta_\infty$. Under conditions (A1)--(A5), it can be shown that $\kappa^*_{n,h_n}(a) \inprob \kappa_0(a)$ by standard kernel smoothing arguments for any sequence $h_n \to 0$. In particular,  $\kappa^*_{n,h_n}(a)$ is consistent under the general form of doubly-robustness specified by condition (A3).

To determine an appropriate value of the bandwidth $h$ in practice, we propose the following empirical criterion. We first define the integrated scale $\gamma_0 := \int \kappa_0(a) F_0(da)$, and construct the estimator $\gamma_n(h) := \int \kappa_{n,h}(a) F_n(da)$ for any candidate $h>0$. Furthermore, we observe that $\gamma_0= E_{0}\left[\eta_\infty(Y, A, W)\right]$, which suggests the use of the empirical estimator $\bar{\eta}_n:=\frac{1}{n} \sum_{i=1}^n  \eta_n(Y_i, A_i, W_i)$. This motivates us to define $h_n^* := \argmin_h \left[\gamma_n(h)-\bar{\eta}_n\right]^2 $, that is, the value of $h$ that makes $\gamma_n(h)$ and $\bar{\eta}_n$ closest. The proposed doubly-robust estimator of $\kappa_0(a)$ is thus $\kappa_{n, \n{DR}}(a) := \kappa_{n,h_n^*}(a)$.

We make two final remarks regarding this doubly-robust estimator of $\kappa_0(a)$. First, we note that this estimator only depends on $A$ and $a$ through the ranks $F_n(A)$ and $F_n(a)$. Hence, as before, our estimator is invariant to strictly monotone transformations of the exposure $A$.  Second, we note that if $\mu_n(a, w) = \mu_n(a)$ does not depend on $w$ and $g_n = 1$, $\kappa_{n,\n{DR}}(a)$ tends to the conditional variance $\n{Var}_0(Y \mid A = a)$, which is precisely the scale parameter appearing in standard isotonic regression.


\subsection{Confidence intervals via sample splitting}\label{sec:samplesplit}

As an alternative, we note here that the sample-splitting method recently proposed by \cite{banerjee2019divide} could also be used to perform inference. Specifically, to implement their approach in our context, we randomly split the sample into $m$ subsets of roughly equal size, perform our estimation procedure on each subset to form subset-specific estimates $\theta_{n,1}, \theta_{n,2}, \dotsc, \theta_{n,m}$, and then define $\bar\theta_{n,m}(a) := \frac{1}{m} \sum_{j=1}^m \theta_{n,j}(a)$. \cite{banerjee2019divide} demonstrated that if $m > 1$ is fixed, then under mild conditions $\bar\theta_{n,m}(a)$ has strictly better asymptotic mean squared error than $\theta_n(a)$, and that for moderate $m$,
\begin{equation} \left[ \bar\theta_{n,m}(a) - \frac{\sigma_{n,m}(a)}{\sqrt{m} n^{1/3}} t_{1-\alpha/2, m-1}, \ \bar\theta_{n,m}(a) + \frac{\sigma_{n,m}(z)}{\sqrt{m}n^{1/3} } t_{1-\alpha/2, m-1}\right] \label{eq:pooled_ci}\end{equation}
forms an asymptotic $1-\alpha$ level confidence interval for $\theta_0(a)$, where $\sigma_{n,m}^2(a) := \frac{1}{m-1} \sum_{j=1}^m [ \theta_{n,j}(a) - \bar\theta_{n,m}(a)]^2$ and $t_{1-\alpha/2, m-1}$ is the $(1-\alpha/2)$-quantile of the $t$-distribution with $m-1$ degrees of freedom.

\section{Numerical studies}\label{numerical}

In this section, we perform numerical experiments to assess the performance of the proposed estimators of $\theta_0(a)$ and of the three approaches for constructing confidence intervals, which we also compare to that of the local linear estimator and associated confidence intervals proposed in \cite{kennedy2016continuous}.

In our experiments, we simulate data as follows. First, we generate $W \in \d{R}^4$ as a vector of four independent standard normal variates. A natural next step would be to generate $A$ given $W$. However, since our estimation procedures requires estimating the conditional density of $U:=F_0(A)$ given $W$, we instead generate $U$ given $W$, and then transform $U$ to obtain $A$. This strategy makes it easier to construct correctly-specified parametric nuisance estimators in the context of these simulations. Given $W=w$, we generate $U$ from the distribution with conditional density function $\bar{g}_0(u\mid w) = I_{[0,1]}(u)\{ \lambda(w) + 2u[1 - \lambda(w)]\}$ for $\lambda(w) := 0.1 + 1.8 \expit(\beta^\top w)$. We note that $\bar{g}_0(u\mid w) \geq 0.1$ for all $u\in [0,1]$ and $w \in \d{R}^4$, and also, that $\int \bar{g}_0(u\mid w) Q_0(dw) = I_{[0,1]}(u)$, so that $U$ is marginally uniform. We then take  $A$ to be the evaluation at $U$ of the quantile function of an equal-weight mixture of two normal distributions with means $-2$ and $2$ and standard deviation 1, which implies that $A$ is marginally distributed according to this bimodal normal mixture. Finally, conditionally upon $A=a $ and $W=w$, we simulate $Y$ as a Bernoulli random variate with conditional mean function given by $\mu_0(a, w) :=  \expit\left(\gamma_1^\top \underbar{w} + \gamma_2^\top \underbar{w} a +  \gamma_3 a^2\right)$, where $\underbar{w}$ denotes $(1, w)$. We set $\beta = (-1, -1, 1, 1)^\top$, $\gamma_1 = (-1, -1,-1,1,1)^\top$, $\gamma_2 = (3, -1,-1,1,1)^\top$ and $\gamma_3 = 3$ in the experiments we report on.

We estimate the true confounder-adjusted dose-response curve $\theta_0$ using the causal isotonic regression estimator $\theta_n$, the local linear estimator of \cite{kennedy2016continuous}, and the sample-splitting version of $\theta_n$ proposed by \cite{banerjee2019divide} with $m=5$ splits. For the local linear estimator, we use the data-driven bandwidth selection procedure proposed in Section~3.5 of \cite{kennedy2016continuous}. We consider three settings in which either both $\mu_n$ and $g_n$ are consistent; only $\mu_n$ consistent; and only $g_n$ consistent. To construct a consistent estimator $\mu_n$, we use a correctly specified logistic regression model, whereas to construct a consistent estimator $g_n$, we use a maximum likelihood estimator based on a correctly specified parametric model. To construct an inconsistent estimator $\mu_n$, we still use a logistic regression model but omit covariates $W_3$, $W_4$ and all interactions. To construct an inconsistent estimator $g_n$, we posit the same parametric model  as before but omit $W_3$ and $W_4$. We construct pointwise confidence intervals for $\theta_0$ in each setting using the Wald-type construction described in Section~\ref{inference} using both the plug-in and doubly-robust estimators of $\kappa_0(a)$. We expect intervals based on the doubly-robust estimator of $\kappa_0(a)$ to provide asymptotically correct coverage rates for $\theta_0(a)$ for each of the three settings, but only expect asymptotically correct coverage rates in the first setting when the plug-in estimator of $\kappa_0(a)$ is used. We construct pointwise confidence intervals for the local linear estimator using the procedure proposed in \cite{kennedy2016continuous}, and for the sample splitting procedure using the procedure discussed in Section~\ref{sec:samplesplit}. We consider the performance of these inferential procedures for values of $a$ between $-3$ and $3$.

The left panel of Figure~\ref{fig:sim_example} shows a single sample path of the causal isotonic regression estimator based on a sample of size $n=5000$ and consistent estimators $\mu_n$ and $g_n$. Also included in that panel are asymptotic 95\% pointwise confidence intervals constructed using the doubly-robust estimator of $\kappa_0(a)$. The right panel shows the unadjusted isotonic regression estimate based on the same data and corresponding 95\% asymptotic confidence intervals.  The true causal and unadjusted  regression curves are shown in red. We note that $\theta_0(a) \neq r_0(a)$ for $a \neq 0$, since the relationship between $Y$ and $A$ is confounded by $W$, and indeed the unadjusted regression curve does not have a causal interpretation. Therefore, the marginal isotonic regression estimator will not be consistent for the true causal parameter. In this data-generating setting, the causal effect of $A$ on $Y$ is larger in magnitude than the marginal effect of $A$ on $Y$ in the sense that $\theta_0(a)$ has greater variation over values of $a$ than does $r_0(a)$.

\begin{figure}[ht]
\centering
\includegraphics[width=\linewidth]{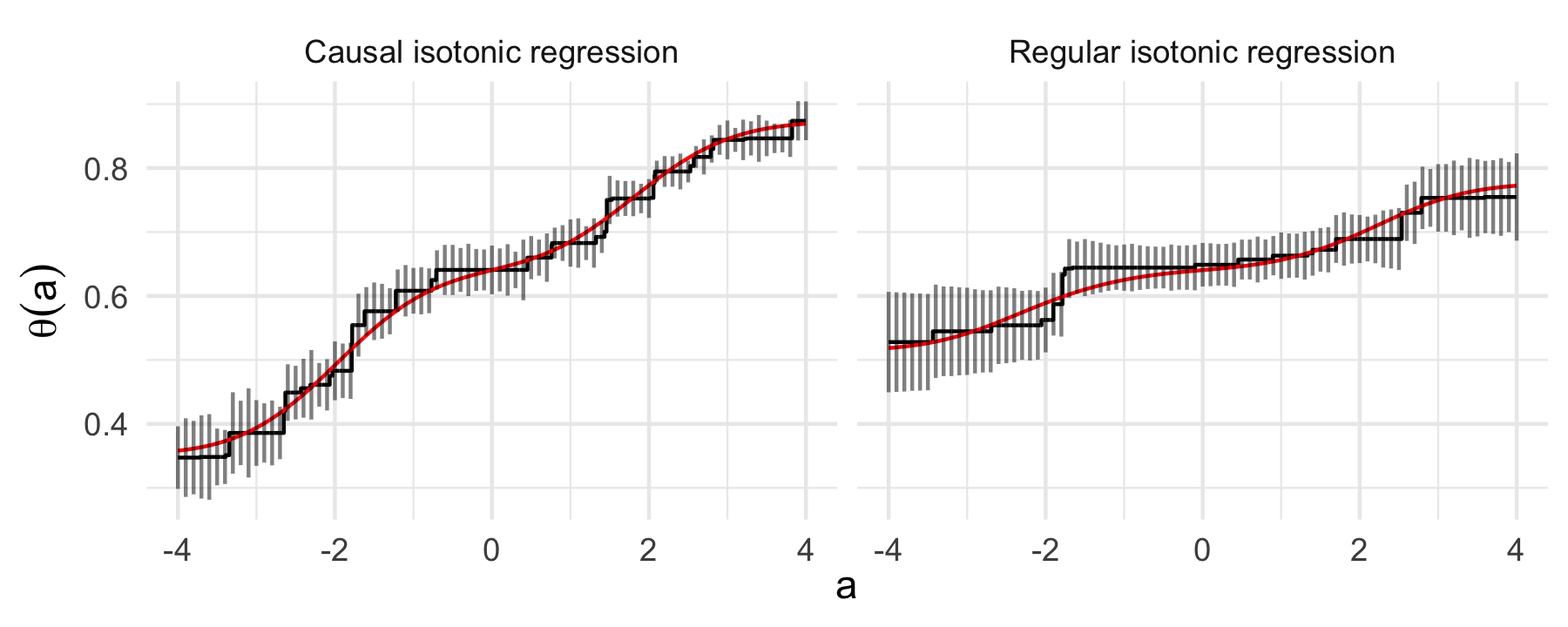}
\caption{Causal isotonic regression estimate using consistent nuisance estimators $\mu_n$ and $g_n$ (left), and regular isotonic regression estimate (right). Pointwise 95\% confidence intervals constructed using the doubly-robust estimator are shown as vertical bars. The true functions are shown in red.}
\label{fig:sim_example}
\end{figure}

We perform 1000 simulations, each with $n\in\{500, 1000, 2500, 5000\}$ observations. Figure~\ref{fig:se} displays the empirical standard error of the three considered estimators over these 1000 simulated datasets as a function of $a$ and for each value of $n$. We first note that the standard error of the local linear estimator is smaller than that of $\theta_n$, which is expected due to the faster rate of convergence of the local linear estimator. The sample splitting procedure also reduces the standard error of $\theta_n$. Furthermore, the standard deviation of the local linear estimator appears to decrease faster than $n^{-1/3}$, whereas the standard deviation of the estimators based on $\theta_n$ do not, in line with the theoretical rates of convergence of these estimators. We also note that inconsistent estimation of the propensity has little impact on the standard errors of any of the estimators, but inconsistent estimation of the outcome regression results in slightly larger standard errors.


\begin{figure}[ht]
\centering
\includegraphics[width=\linewidth]{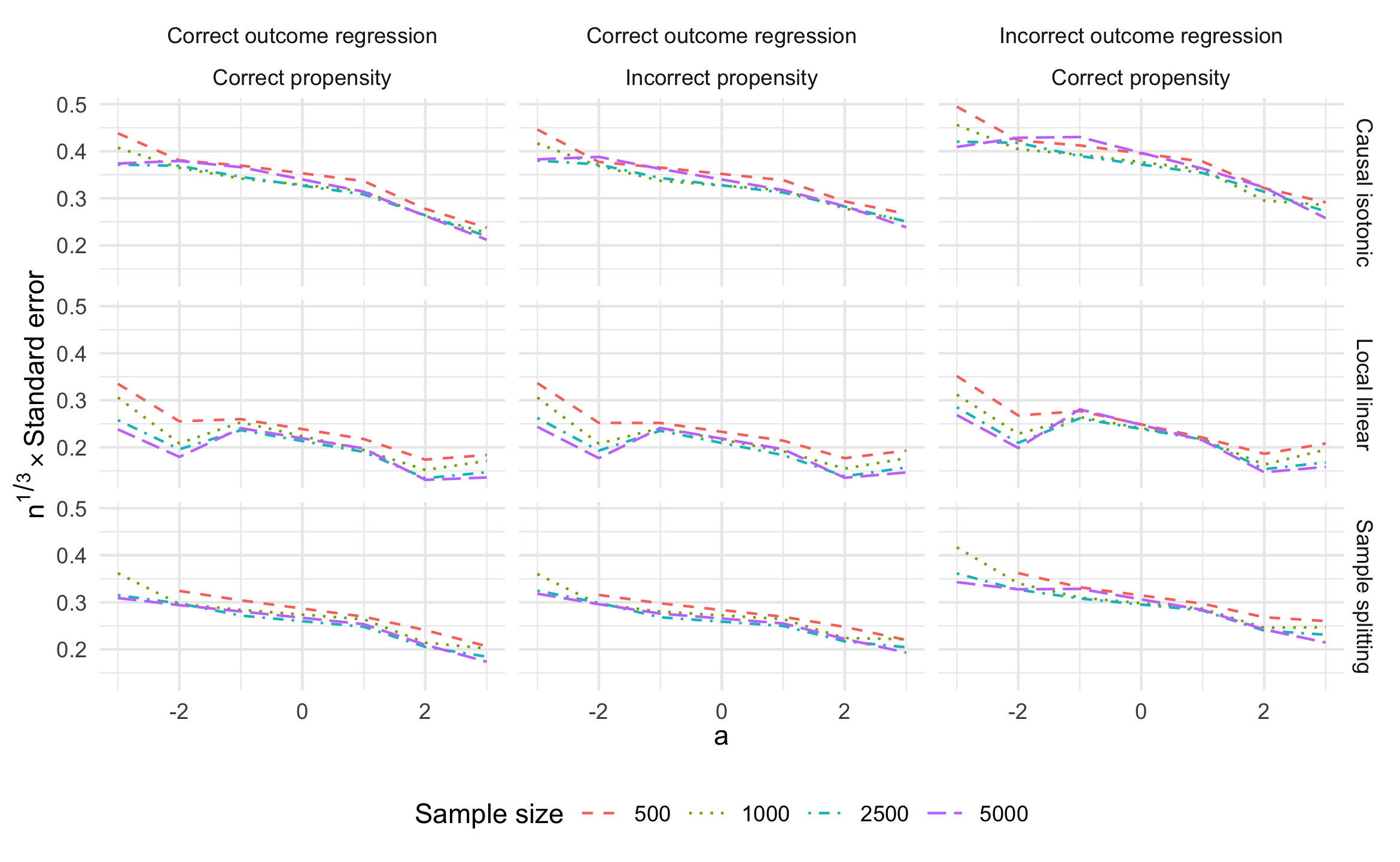}
\caption{Standard error of the three estimators scaled by $n^{1/3}$ as a function of $n$ for different values of $a$ and in contexts in which $\mu_n$ and $g_n$ are either consistent or inconsistent, computed empirically over 1000 simulated datasets of different sizes.}
\label{fig:se}
\end{figure}

Figure~\ref{fig:bias} displays the absolute bias of the three estimators. For most values of $a$, the estimator $\theta_n$ proposed here has smaller absolute bias than the local linear estimator, and its absolute bias decreases faster than $n^{-1/3}$. The absolute bias of the local linear estimator depends strongly on $a$, and in particular is largest where the second derivative of $\theta_0$ is large in absolute value, agreeing with the large-sample theory described in \cite{kennedy2016continuous}. The sample splitting estimator has larger absolute bias than $\theta_n$ because it inherits the bias of $\theta_{n/m}$. The bias is especially large for values of $a$ in the tails of the marginal distribution of $A$.

\begin{figure}[ht]
\centering
\includegraphics[width=\linewidth]{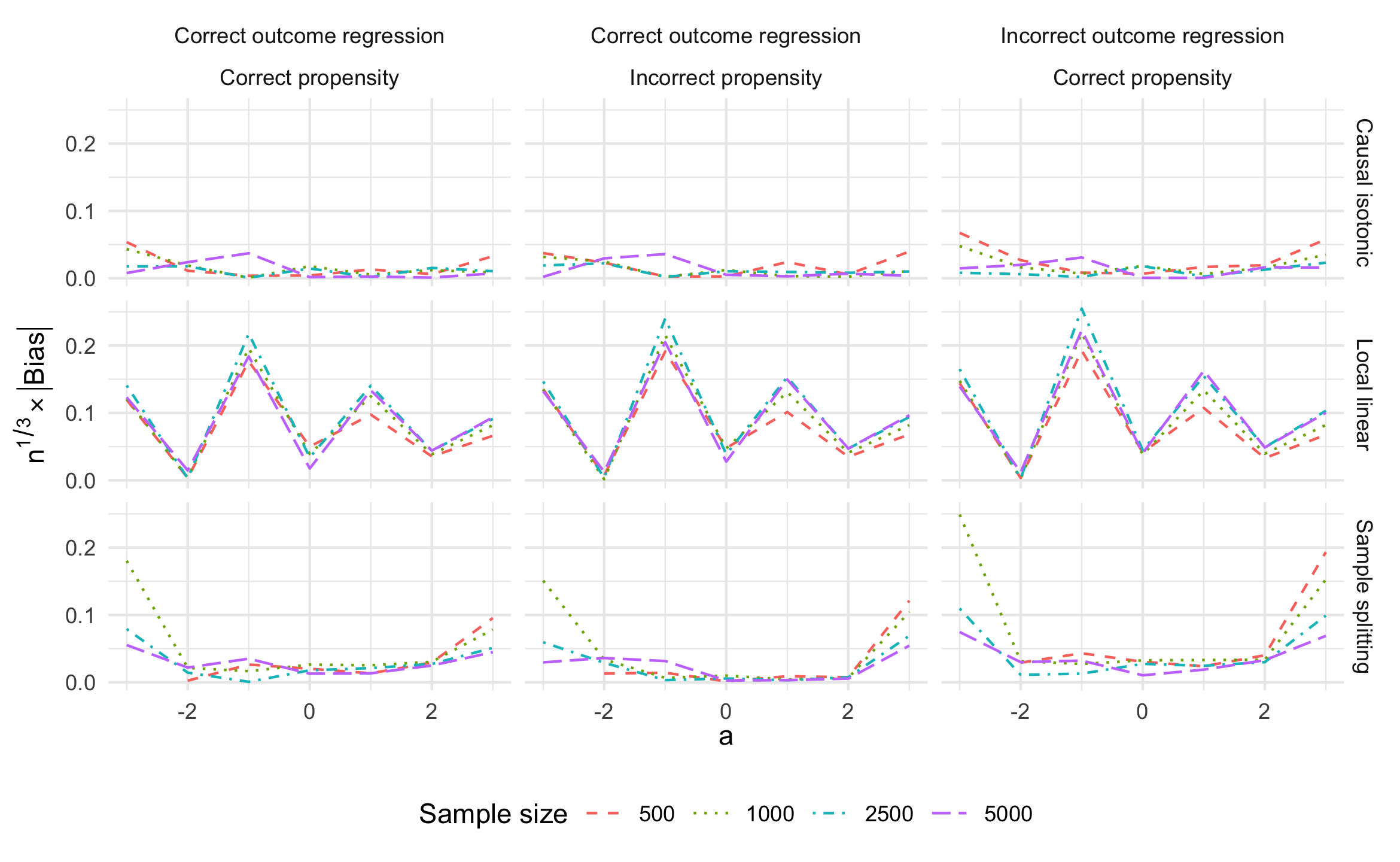}
\caption{Absolute bias of the three estimators scaled by $n^{1/3}$ as a function of $n$ for different values of $a$ and in contexts in which $\mu_n$ and $g_n$ are either consistent or inconsistent, computed empirically over 1000 simulated datasets of different sizes.}
\label{fig:bias}
\end{figure}

Figure~\ref{fig:coverages} shows the empirical coverage of nominal 95\% pointwise confidence intervals for a range of values of $a$ for the four methods considered. For both the plug-in and doubly-robust intervals centered around $\theta_n$, the coverage improves as $n$ grows, especially for values of $a$ in the tails of the marginal distribution of $A$.  Under correct specification of outcome and propensity regression models, the plug-in method attains close to nominal coverage rates for $a$ between $-3$ and $3$ by $n=1000$. When the propensity estimator is inconsistent, the plug-in method still performs well in this example, although we do not expect this to always be the case. However, when $\mu_n$ is inconsistent, the plug-in method is very conservative for positive values of $a$. The doubly-robust method attains close to nominal coverage for large samples as long as one of $g_n$ or $\mu_n$ is consistent. Compared to the plug-in method, the doubly-robust method requires larger sample sizes to achieve good coverage, especially for extreme values of $a$.  This is because the doubly-robust estimator of $\kappa_0(a)$ has a slower rate of convergence than does the plug-in estimator, as demonstrated by box plots of these estimators provided in Supplementary Material. 

The confidence intervals associated with the local linear estimator have poor coverage for values of $a$ where the bias of the estimator is large, which, as mentioned above, occurs when the second derivative of $\theta_0$ is large in absolute value. Overall, the sample splitting estimator has excellent coverage, except perhaps for values of $a$ in the tails of the marginal distribution of $A$ when $n$ is small or moderate, in which case the coverage is near 90\%.

\begin{figure}[h!]
\centering
\includegraphics[width=\linewidth]{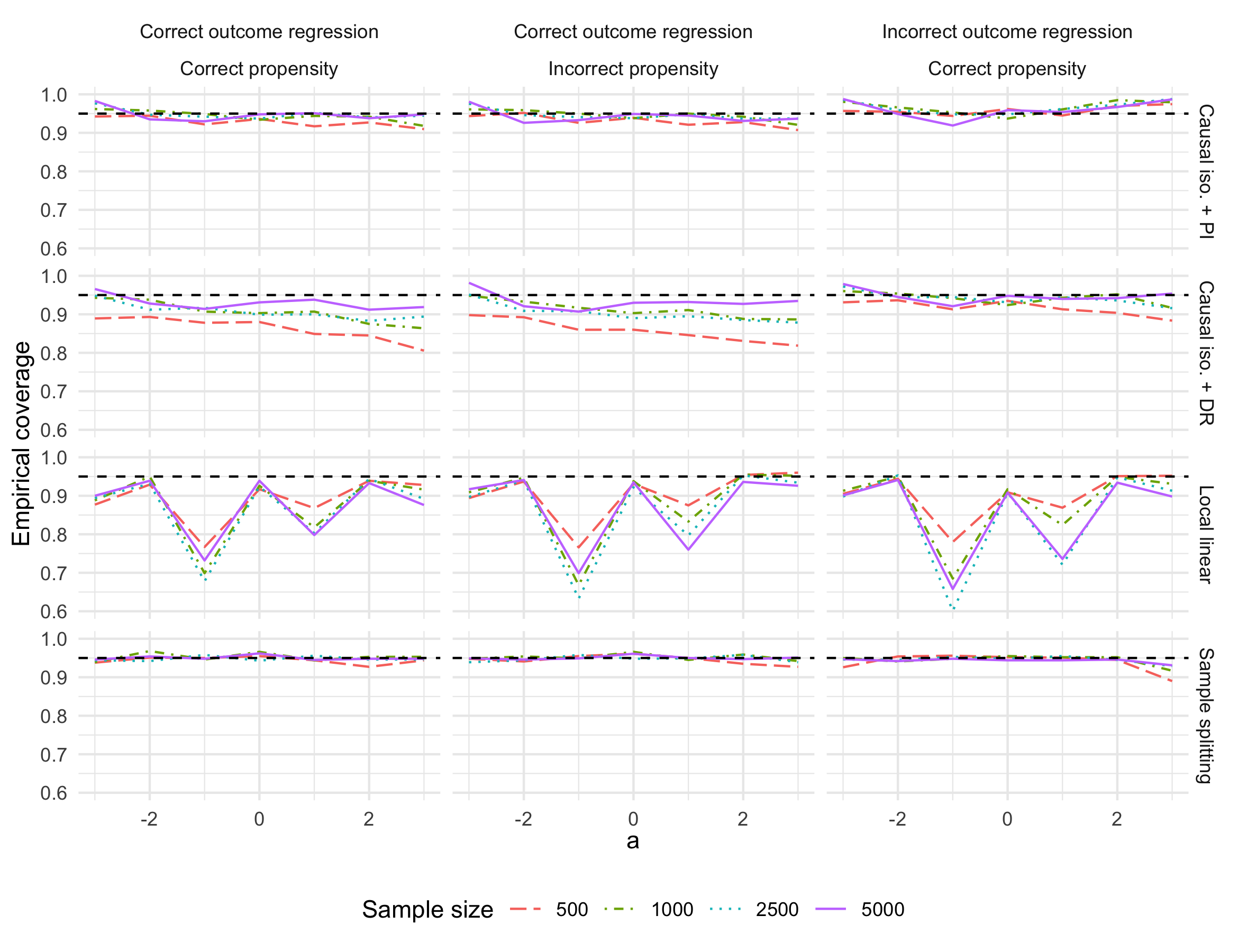}
\caption{Observed coverage of pointwise 95\% confidence intervals using $\theta_n$ and the plug-in method (top row), $\theta_n$ and eht doubly-robust method (second row), the local linear estimator and associated intervals (third row), and the sample splitting estimator (bottom row), considered for different values of $a$ and computed empirically over 1000 simulated datasets of different sizes. Columns indicate whether $\mu_n$ and $g_n$ is consistent or not. Black dashed lines indicate the nominal coverage rate.}
\label{fig:coverages}
\end{figure}

We also conducted a small simulation study to illustrate the performance of the proposed procedures when machine learning techniques are used to construct $\mu_n$ and $g_n$. To consistently estimate $\mu_0$, we used a Super Learner \citep{vanderlaan2007super} with a library consisting of generalized linear models, multivariate adaptive regression splines, and generalized additive models. To consistently estimate $g_0$, we used the method proposed by \cite{diaz2011super} with covariate vector $(W_1, W_2,W_3, W_4)$. To produce inconsistent estimators $\mu_n$ or $g_n$, we used the same estimators but omitted covariates $W_1$ and $W_2$. We also considered the estimator $\theta_n^\circ$ obtained via cross-fitting these nuisance parameters, as discussed in Section~\ref{sec:cv}, as well as the local linear estimator. Due to computational limitations, we performed 1000 simulations at sample size $n=1000$ only. Figure~\ref{fig:coverages_ml} shows the coverage of nominal 95\% confidence intervals. The plug-in intervals achieve very close to nominal coverage under consistent estimation of both nuisances, and also achieve surprisingly good coverage rates when the propensity is inconsistently estimated. The plug-in intervals are somewhat conservative when the outcome regression is inconsistently estimated. The doubly-robust method is anti-conservative under inconsistent estimation of both nuisances and also when the propensity is inconsistently estimated, with coverage rates mostly between 90 and 95\%. Good coverage rates are also achieved when the outcome regression is inconsistently estimated. These results suggest that the doubly-robust intervals may require larger sample sizes to achieve good coverage, particularly when machine learning estimators are used for $\mu_n$ and $g_n$. The plug-in intervals appear to be relatively robust to moderate misspecification of models for the nuisance parameters in smaller samples. Histograms of the estimators of $\kappa_0(a)$ and $\psi_0'(a)$ are provided in the Supplementary Material. Confidence intervals based on the local linear estimator show a similar pattern as in the previous simulation study, undercovering where the second derivative of the true function is large in absolute value. Cross-fitting had little impact on coverage.

As noted above, we found in our numerical experiments that the plug-in estimator of the scale parameter was surprisingly robust to inconsistent estimation of the nuisance parameters, while its doubly-robust estimator was anti-conservative even when the nuisance parameters were estimated consistently. This phenomenon can be explained in terms of the bias and variance of the two proposed scale estimators. On one hand, under inconsistent estimation of any nuisance function, the plug-in estimator of the scale parameter is biased, even in large samples. However, its variance decreases relatively quickly with sample size, since it is a simple empirical average of estimated functions. On the other hand, the doubly-robust estimator is asymptotically unbiased, but its variance decreases much slower with sample size. These trends can be observed in the figures provided in the Supplementary Material. In sufficiently large samples, the doubly-robust estimator is expected to outperform the plug-in estimator in terms of mean squared error when one of the nuisances is inconsistently estimated. However, the sample size required for this trade-off to significantly affect confidence interval coverage depends on the degree of inconsistency. While we did not see this tradeoff occur at the sample sizes used in our numerical experiments, we expect the benefits of the doubly-robust confidence interval construction to become apparent in smaller samples in other settings.

\begin{figure}[h!]
\centering
\includegraphics[width=\linewidth]{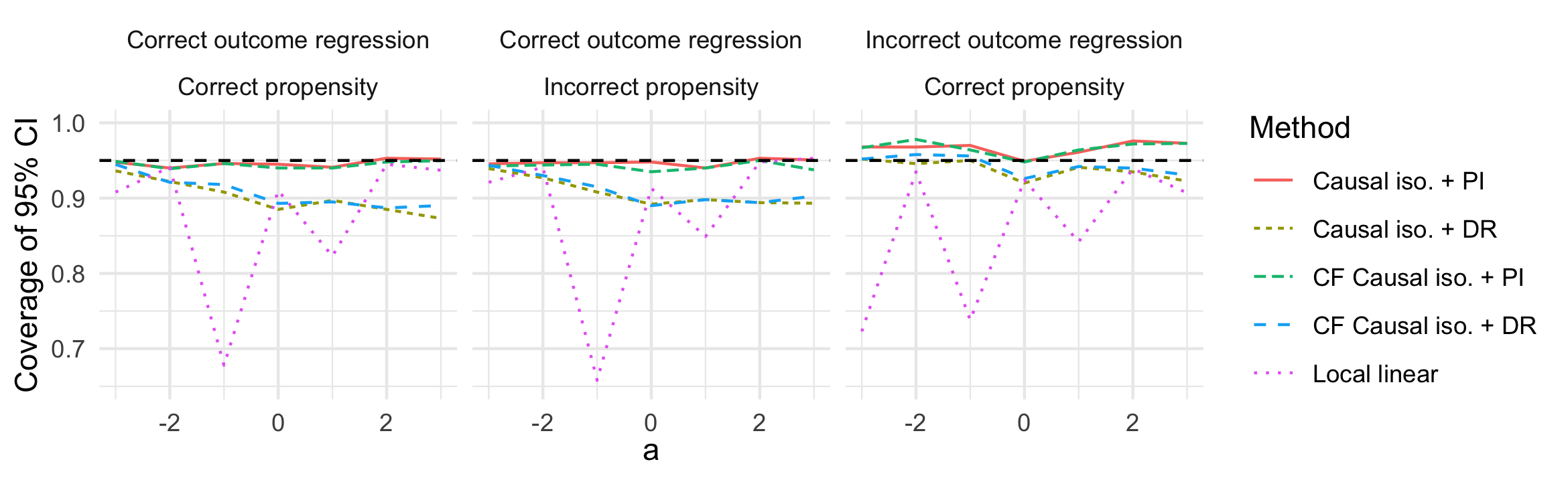}
\caption{Observed coverage of pointwise 95\% doubly-robust and plug-in confidence intervals using machine learning estimators based on simulated data including $n=1000$ observations. Columns indicate whether $\mu_n$ and $g_n$ are consistent or not. Black dashed lines indicate the nominal coverage rate. CF stands for cross-fitted; PI for plug-in; DR for doubly-robust.}
\label{fig:coverages_ml}
\end{figure}

\section{BMI and T-cell response in HIV vaccine studies}\label{bmi}

The scientific literature indicates that, for several vaccines, obesity or BMI is inversely associated with immune responses to vaccination (see, e.g.\ \citealp{sheridan2012obesity, young2013obesity, jin2015multiple, painter2015weight,liu2017influences}). Some of this literature has investigated potential mechanisms of how obesity or higher BMI might lead to impaired immune responses. For example, \cite{painter2015weight} concluded that obesity may alter cellular immune responses, especially in adipose tissue, which varies with BMI. \cite{sheridan2012obesity} found that obesity is associated with decreased CD8+ T-cell activation and decreased expression of functional proteins in the context of influenza vaccines. \cite{liu2017influences} found that obesity reduced Hepatitis B immune responses through ``leptin-induced systemic and B cell intrinsic inflammation, impaired T cell responses and lymphocyte division and proliferation." Given this evidence of a monotone effect of BMI on immune responses, we  used the methods presented in this paper to assess the covariate-adjusted relationship between BMI and CD4+ T-cell responses using data from a collection of clinical trials of candidate HIV vaccines. We present the results of our analyses here.

In \cite{jin2015multiple}, the authors compared the compared the rate of CD4+ T cell response to HIV peptide pools among low (BMI $<25$) medium ($25 \leq$ BMI $<30$) and high (BMI $\geq 30$) BMI participants, and they found that low BMI participants had a statistically significantly greater response rate than high BMI participants using Fisher's exact test. However, such a marginal assessment of the relationship between BMI and immune response can be misleading because there are known common causes, such as age and sex, of both BMI and immune response. For this reason, \cite{jin2015multiple} also performed a logistic regression of the binary CD4+ responses against sex, age, BMI (not discretized), vaccination dose, and number of vaccinations. In this adjusted analysis, they found a significant association between BMI and CD4+ response rate after adjusting for all other covariates (OR: 0.92; 95\% CI: 0.86, 0.98; $p$=0.007). However, such an adjusted odds-ratio only has a formal causal interpretation under strong parametric assumptions. As discussed in Section~\ref{sec:param}, the covariate-adjusted dose-response function $\theta_0$ is identified with the causal dose-response curve without making parametric assumptions, and is therefore of interest for understanding the continuous covariate-adjusted relationship between BMI and immune responses. 

 We note that there is some debate in the causal inference literature about whether exposures such as BMI have a meaningful interpretation in formal causal modeling. In particular, some researchers suggest that causal models should always be tied to hypothetical randomized experiments (see, e.g., \citealp{bind2017bridging}), and it is difficult to imagine a hypothetical randomized experiment that would assign participants to levels of BMI. From this perspective, it may therefore not be sensible to interpret $\theta_0(a)$ in a causal manner in the context of this example. Nevertheless, as discussed in the introduction, we contend that $\theta_0(a)$ is still of interest. In particular, it provides a meaningful summary of the relationship between BMI and immune response accounting for measured potential confounders. In this case, we interpret $\theta_0(a)$ as the  probability of immune response in a population of participants with BMI value $a$ but sex, age, vaccination dose, number of vaccinations, and study with a similar distribution to that of the entire study population.
 
We pooled data from the vaccine arms of 11 phase I/II clinical trials, all conducted through the HIV Vaccine Trials Network (HVTN). Ten of these trials were previously studied in the analysis presented in \cite{jin2015multiple}, and a detailed description of the trials are contained therein. The final trial in our pooled analysis is HVTN 100, in which 210 participants were randomized to receive four doses of the ALVAC-HIV vaccine (vCP1521). The ALVAC-HIV vaccine, in combination with an AIDSVAX boost, was found to have statistically significant vaccine efficacy against HIV-1 in the RV-144 trial conducted in Thailand \citep{rerks2009vaccination}. CD4+ and CD8+ T-cell responses to HIV peptide pools were measured in all 11 trials using validated intracellular cytokine staining at HVTN laboratories. These continuous responses were converted to binary indicators of whether there was a significant change from baseline using the method described in \cite{jin2015multiple}. We analyzed these binary responses at the first visit following administration of the last vaccine dose--either two or four weeks after the final vaccination  depending on the trial. After accounting for missing responses from a small number of participants, our analysis datasets consisted of a total of $n=439$ participants for the analysis of CD4+ responses and $n=462$ participants for CD8+ responses. Here, we focus on analyzing CD4+ responses; we present the analysis of CD8+ responses in Supplementary Material.

We assessed the relationship between BMI and T-cell response by estimating the covariate-adjusted dose-response function $\theta_0$ using our cross-fitted estimator $\theta_n^\circ$,  the local linear estimator, and the sample-splitting version of our estimator with $m=5$ splits. We adjusted for sex, age, vaccination dose, number of vaccinations, and study. We estimated  $\mu_0$ and $g_0$ as in the machine learning-based simulation study described in Section~\ref{numerical}, and constructed confidence intervals for our estimator using both the plug-in and doubly-robust estimators described above.

Figure~\ref{fig:tcell_responses} presents the estimated probability of a positive CD4+ T-cell response as a function of BMI for BMI values between the 0.05 and 0.95 quantile of the marginal empirical distribution of BMI using our estimator (left panel), the local linear estimator (middle panel), and the sample-splitting estimator (right panel). Pointwise 95\% confidence intervals are shown as dashed/dotted lines. The three methods found qualitatively similar results. We found that the change in probability of CD4+ response appears to be largest for BMI $< 20$ and BMI $>30$. We estimated the probability of having a positive CD4+ T-cell response, after adjusting for potential confounders, to be 0.52 (95\% doubly-robust CI: 0.44--0.59) for a BMI of 20, 0.47 (0.42--0.52) for a BMI of 25, 0.47 (0.32--0.62) for a BMI of 30, and 0.29 (0.12--0.47) for a BMI of 35. We estimated the difference between these probabilities for BMIs of 20 and 35 to be 0.22 (0.03--0.41).


\begin{figure}[h!]
\centering
\includegraphics[width=6.5in]{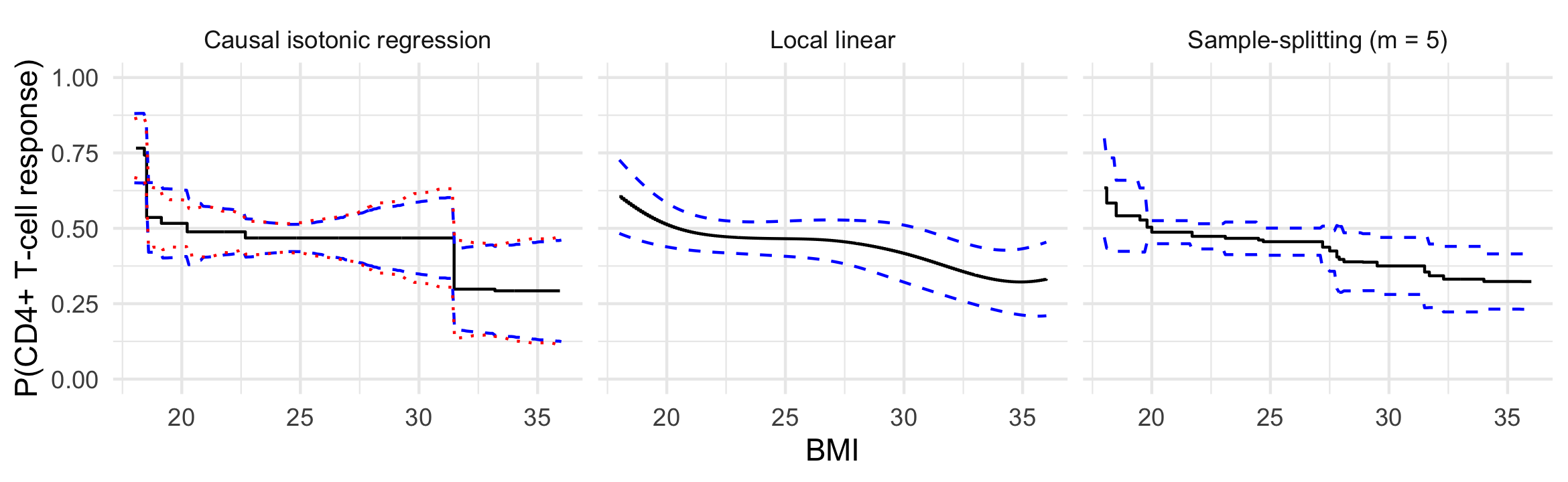}
\caption{Estimated probabilities of CD4+ T-cell response and 95\% pointwise confidence intervals as a function of BMI, adjusted for sex, age, number of vaccinations received, vaccine dose, and study. The left panel displays the estimator proposed here, the middle panel the local linear estimator of \cite{kennedy2016continuous}, and the right panel the sample-splitting version of our estimator with $m=5$ splits. In the left panel, the blue dashed lines are confidence intervals based on the plug-in estimator of the scale parameter, and the dotted lines are based on the doubly-robust estimator of the scale parameter.}
\label{fig:tcell_responses}
\end{figure}

\section{Concluding remarks}\label{discussion}


The work we have presented in this paper lies at the interface of causal inference and shape-constrained nonparametric inference, and there are natural future directions building on developments in either of these areas. Inference on a monotone causal dose-response curve when outcome data are only observed subject to potential coarsening, such as censoring, truncation, or missingness, is needed to increase the applicability of our proposed method. To tackle such cases, it appears most fruitful to follow the general primitive strategy described in \cite{westling2018monotone} based on a revised causal identification formula allowing such coarsening.

It would be useful to develop tests of the monotonicity assumption, as \cite{durot2003test} did for regression functions. Such a test could likely be developed by studying the large-sample behavior of $\|\bar\Psi_n  -\Psi_n \|_p$ under the null hypothesis that $\theta_0$ is monotone, where $\Psi_n$ and $\bar\Psi_n$ are the primitive estimator and its greatest convex minorant as defined in Section~\ref{sec:defn}. Such a result would likely permit testing with a given asymptotic size when $\theta_0$ is strictly increasing, and asymptotically conservative inference otherwise. It would also be useful to develop methods for uniform inference. Uniform inference is difficult in this setting due to the fact that $\{ n^{1/3}[\theta_n(a) - \theta_0(a)]: a \in \s{A} \}$ does not convergence weakly as a process in the space $\ell^\infty(\s{A})$ of bounded functions on $\s{A}$ to a tight limit process. Indeed, Theorem~\ref{thm:joint_conv} indicates that  $\{ n^{1/3}[\theta_n(a) - \theta_0(a)]: a \in \s{A} \}$ converges to an independent white noise process, which is not tight, so that this convergence is not useful for constructing uniform confidence bands. Instead, it may be possible to extend the work of \cite{durot2012} to our setting (and other generalized Grenander-type estimators) by  demonstrating that $\log n \left[ (n / \log n)^{1/3} \sup_{a \in \s{A}_n} |\theta_n(a) - \theta_0(a)| / \alpha_0 - c_n\right]$ converges in distribution to a non-degenerate limit for some constant $\alpha_0$ depending upon $P_0$, a deterministic sequence $c_n$, and a suitable sequence of subsets $\s{A}_n$ increasing to $\s{A}$. Developing procedures for uniform inference and tests of the monotonicity assumption are important areas for future research.

An alternative approach to estimating a causal dose-response curve is to use local linear regression, as \cite{kennedy2016continuous} did. As is true in the context of estimating classical univariate functions such as density, hazard, and regression functions, there are certain trade-offs between local linear smoothing and monotonicity-based methods. On the one hand, local linear regression estimators exhibit a faster $n^{-2/5}$ rate of convergence whenever optimal tuning rates are used and the true function possesses two continuous derivatives. However, the limit distribution involves an asymptotic bias term depending on the second derivative of the true function, so that confidence intervals based on optimally-chosen tuning parameters provide asymptotically correct coverage only for a \emph{smoothed} parameter rather than the true parameter of interest. In contrast, monotonicity-constrained estimators such as the estimator proposed here exhibit an $n^{-1/3}$ rate of convergence whenever the true function is strictly monotone and possesses one continuous derivative, do not require choosing a tuning parameter, are invariant to strictly increasing transformations of the exposure, and their limit theory does not include any asymptotic bias (as illustrated by Theorem~\ref{thm:dose_response}). We note that both estimators achieve the optimal rate of convergence for pointwise estimation of a univariate function under their respective smoothness constraints. In our view, the ability to perform asymptotically valid inference using a monotonicity-constrained estimator is one of the most important benefits of leveraging the monotonicity assumption rather than using smoothing methods. This advantage was evident in our numerical studies when comparing the isotonic estimator proposed here and the local linear method of \cite{kennedy2016continuous}. Under-smoothing can be used to construct calibrated confidence intervals using kernel-smoothing estimators, but performing adequate under-smoothing in practice is challenging.

The two methods for pointwise asymptotic inference we presented require estimation of the derivative $\theta_0'(a)$ and the scale parameter $\kappa_0(a)$. We found that the plug-in estimator of $\kappa_0(a)$ had low variance but possibly large bias depending on the levels of inconsistency of $\mu_n$ and $g_n$, and that its doubly-robust estimator instead had high variance but low bias as long as either $\mu_n$ or $g_n$ is consistent. In practice, we found the low variance of the plug-in estimator to often outweigh its bias, resulting in better coverage rates for intervals based on the plug-in estimator of $\kappa_0(a)$, especially in samples of small and moderate sizes. Whether a doubly-robust estimator of $\kappa_0(a)$ with smaller variance can be constructed is an important question to be addressed in future work. We found that sample splitting with as few as $m=5$ splits provided doubly-robust coverage, and the sample splitting estimator also had smaller variance than the original estimator, at the expense of some additional bias.

It would be even more desirable to have inferential methods that do not require estimation of additional nuisance parameters or sample splitting. Unfortunately, the standard nonparametric bootstrap is not generally consistent in Grenander-type estimation settings, and although alternative bootstrap methods have been proposed, to our knowledge, all such proposals require the selection of critical tuning parameters \citep{kosorok2008bootstrap, sen2010bootstrap}. Likelihood ratio-based inference for Grenander-type estimators has proven fruitful in a variety of contexts (see, e.g. \citealp{banerjee2001ratio, groeneboom2015nonparametric}), and extending such methods to our context is also an area of significant interest in future work.

\vspace{0.1in}
\singlespacing
\bibliographystyle{chicago}
\bibliography{../../monotone_dose_response.bib}

\clearpage
\doublespacing

\section*{Supplementary material: technical results}

We will use the notation $Pf$ to refer to $\int f dP$ for any probability measure $P$ and $P$-integrable function $f$. We will denote by $\d{P}_n$ the empirical distribution based on $O_1,O_2,\ldots,O_n$, so that $\d{P}_n f := \frac{1}{n} \sum_{i=1}^n f(O_i)$. We will denote by $\d{G}_n$ the empirical process $n^{1/2}(\d{P}_n - P_0)$. Finally, we will say that $a \lesssim b$ if there exists a $c < \infty$ such that $a \leq cb$. Below, for brevity, we will refer to \cite{westling2018monotone} as WC.

Throughout the Supplementary Material, we will refer to $a_0$ as any element of $\s{A}$ at which we evaluate functions such as $\theta_0$, $\theta_n$, $\Gamma_0$ or $\Gamma_n$. We will reserve $a$ for arguments to integrands and influence functions.

\subsection*{Supporting lemmas}

Before proceeding to proofs for Theorems 1 and 2, we state three lemmas that we will use. First, we derive a first-order expansion of $\Gamma_n(a_0)$ that we will rely upon.  
We define $\phi_{\infty, a_0} := \phi_{\mu_\infty, g_\infty, a_0}$ with
\begin{align*}
\phi_{\mu, g, a_0}(y, a, w)\ :=&\ \ I_{(-\infty, a_0]}(a) \left[\frac{ y - \mu(a,w)}{g(a,w)}+ \int \mu( a, \tilde{w})  Q_0(d\tilde{w})\right]\\
&\ \ \quad+ \int_{-\infty}^{a_0}\mu(a, w)  F_0(da) - \iint_{-\infty}^{a_0} \mu(a, \tilde{w}) F_0(da) Q_0(d\tilde{w})\ ,\\
\phi_{\mu, g, a_0}'(y,a,w)\ :=&\ \ I_{(-\infty, a_0]}(a)\left[  \frac{ y - \mu(a,w)}{g(a,w)} + \int \mu( a, \tilde{w}) \,Q_0(d\tilde{w})\right] + \int_{-\infty}^{a_0} \mu(\tilde{a}, w) \, F_0(d\tilde{a})\ .
\end{align*}
and $\phi_{\infty, a_0}^*:=\phi_{\infty, a_0} - \Gamma_0(a_0)$. We also define
\begin{align*}
\gamma_{\mu, a_0}(o_i, o_j)\ :=&\ \ I_{(-\infty, a_0]}(a_i) \mu(a_i, w_j) + I_{(-\infty, a_0]}(a_j) \mu(a_j, w_i) \\
&\ \ \quad- \int \left[ I_{(-\infty, a_0]}(a_i) \mu(a_i, w) + I_{(-\infty, a_0]}(a_j)\mu(a_j, w)\right]Q_0(dw) \\
&\ \ \quad- \int_{-\infty}^{a_0} \left[ \mu(a, w_i) + \mu(a, w_j)\right]F_0(da) +  2\iint I_{(-\infty, a_0]}(a) \mu(a, w) F_0(da) Q_0(dw)\ .
\end{align*}
We then have the following first-order expansion.
\begin{lemma}
If condition (A3) holds, then $\Gamma_{n}(a_0) -  \Gamma_0(a_0) = \d{P}_n\phi_{\infty, a_0}^* + R_{n,a_0}$, where we have defined $R_{n,a_0} := R_{n,a_0,1}+R_{n,a_0,2}+R_{n,a_0,3}$ with
\begin{align*}
R_{n,a_0,1}\ &:=\  \iint_{-\infty}^{a_0} \left[\mu_n(a, w) - \mu_0(u, w)\right] \left[ 1 - \frac{g_0(a, w)}{g_n(a,w)}\right] F_0(da) Q_0(dw)\ ,\\
R_{n,a_0,2}\ &:=\ (\d{P}_n - P_0)(\phi_{\mu_n, g_n, a_0}'  - \phi_{\mu_{\infty}, g_{\infty}, a_0}')\ , \\
R_{n,a_0,3}\ &:=\ \frac{1}{2n^2} \sum_{i \neq j} \gamma_{\mu_n, a_0}(O_i, O_j)+\frac{1}{2n^{3/2}}\d{G}_n \gamma_{\mu_n, a_0}+ \frac{1}{n}E_{0} \left[I_{(-\infty, a_0]}(A) \mu_n(A, W) \left[ 1- \frac{1}{g_0(A, W)} \right]\right].
\end{align*}
\end{lemma}
\begin{proof}
We define
\begin{align*}
\phi_{n,a_0}(y, a, w)\ :=&\ \  I_{(-\infty, a_0]}(a) \left[  \frac{ y - \mu_n(a,w)}{g_n(a,w)}+ \int \mu_n( a, \tilde{w}) Q_n(d\tilde{w})\right]\\
&\ \ \quad+ \int_{-\infty}^{a_0} \mu_n(a, w) F_n(da) -  \iint_{-\infty}^{a_0} \mu_n(a, \tilde{w}) F_n(da) Q_n(d\tilde{w})\ ,
\end{align*}
so that $\Gamma_{n}(a_0) = \d{P}_n \phi_{n,a_0}$. By (A3), we have that
\[ P_0 \phi_{\infty,a_0}  = \iint_{-\infty}^{a_0}\left[\mu_{\infty}(a,w) - \mu_0(a,w)\right] \left[1 - \frac{g_0(a,w)}{g_{\infty}(a,w)}\right] F_0(da) \, Q_0(dw) + \Gamma_0(a_0) = \Gamma_0(a_0)\ . \]
Thus, we have the expansion $\Gamma_{n}(a_0) -  \Gamma_0(a_0) = \d{P}_n\phi_{\infty, a_0}^* + R_{n,a_0}$ for $R_{n,a_0} := (\d{P}_n - P_0)(\phi_{n,a_0} - \phi_{\infty, a_0})+ P_0 \phi_{n,a_0} - \Gamma_0(a_0)$. By adding and subtracting terms and rearranging, we can write  $R_{n,a_0}$ as follows:
\begin{align*}
R_{n,a_0} &= \iiint  I_{(-\infty, a_0]}(a) \left[  \frac{ y - \mu_n(a,w)}{g_n(a,w)} - \frac{ y - \mu_\infty(a,w)}{g_\infty(a,w)} \right] (\d{P}_n - P_0)(dy, da, dw) \\
&\qquad + \iint_{-\infty}^{a_0} \mu_n(a,w) \left[F_n(da) Q_n(dw) - F_0(da) Q_0(dw)\right]\\
&\qquad -  \iint_{-\infty}^{a_0} \mu_\infty(a,w) \left[F_n(da) Q_0(dw) +F_0(da) Q_n(dw)  - 2F_0(da) Q_0(dw)\right] \\
&\qquad + \iint_{-\infty}^{a_0} \left\{ \left[\mu_0(a,w) - \mu_n(a,w)\right] \frac{g_0(a,w)}{g_n(a,w)} +  \left[\mu_n(a,w) - \mu_0(a,w)\right] \right\}  F_0(da) Q_0(dw) \\
&= \iiint  I_{(-\infty, a_0]}(a) \left[  \frac{ y - \mu_n(a,w)}{g_n(a,w)} - \frac{ y - \mu_\infty(a,w)}{g_\infty(a,w)} \right] (\d{P}_n - P_0)(dy, da, dw) \\
&\qquad + \iint_{-\infty}^{a_0} \mu_n(a,w) \left[F_n(da) Q_n(dw)  - F_n(da) Q_0(dw) - F_0(da) Q_n(dw) + F_0(da) Q_0(dw)\right]\\
&\qquad +  \iint_{-\infty}^{a_0}\left[   \mu_n(a,w)  - \mu_\infty(a,w)\right] \left[F_n(da) Q_0(dw) +F_0(da) Q_n(dw)  - 2F_0(da) Q_0(dw)\right] \\
&\qquad + \iint_{-\infty}^{a_0} \left[\mu_n(a,w) - \mu_0(a,w)\right] \left[ 1- \frac{g_0(a,w)}{g_n(a,w)} \right] F_0(da) Q_0(dw) \ .
\end{align*}
The sum of the first and third lines in the preceding display can be expressed as $(\d{P}_n - P_0)(\phi_{\mu_n, g_n, a_0}'  - \phi_{\mu_{\infty}, g_{\infty}, a_0}')$. Therefore, we can decompose the remainder term $R_{n,a_0}$ into $R_{n,a_0,1}+R_{n,a_0,2}+R_{n,a_0,3}$ as claimed, where
\begin{align*}
R_{n,a_0,3}\ :=\ \iint_{-\infty}^{a_0} \mu_{n}(a, w)(F_n - F_0)(da) (Q_n - Q_0)(dw)\ .
\end{align*}
Furthermore, $R_{n,a_0,3}$ can be rewritten as claimed by adding and subtracting terms.
\end{proof}

Lemma~\ref{lemma:marginal_class} below indicates that the entropy of a uniformly bounded class over a product space, when marginalized over one component of the product space with respect to a fixed probability measure, is bounded above by the entropy of the original class.
\begin{lemma}\label{lemma:marginal_class}
Let $\s{F}$ be a uniformly bounded class of functions $f : \s{Z}_1 \times \s{Z}_2 \to \d{R}$, with $|f| \leq K < \infty$ for all $f \in \s{F}$. Let $R$ be a fixed probability measure on $\s{Z}_2$, and define $\s{F}^* := \{ z_1 \mapsto \int f(z_1, z_2)R(dz_2) : f \in \s{F}\}$. Then, we have that
\[\sup_{Q} N(\varepsilon K, \s{F}^*, L_2(Q))\ \leq\ \sup_{Q} N(\varepsilon K/2, \s{F}, L_2(Q))\ .\]
\end{lemma}
\begin{proof}
The statement follows immediately from Lemma 5.2 of \cite{van2006survival} by taking $r=s=t=2$.
\end{proof}

The final lemma concerns so-called \emph{degenerate U-processes}, and is a slight extension of Theorem 6 of \cite{nolan1987uprocess}. A $P_0$-degenerate $U$-process for a class of functions $\s{F}$ is defined as a sum of the form $\left\{S_n(f) : f \in \s{F}\right\}$, where 
\[ S_n(f) := \sum_{1 \leq i \neq j \leq n} f(O_i, O_j)\ ,\]
and where each $f \in \s{F}$ is a function from $\s{O} \times \s{O} \to \d{R}$ satisfying that: (i) $f$ is symmetric in its arguments, meaning that $f(o, \tilde{o}) = f(\tilde{o}, o)$ for all $o, \tilde{o} \in \s{O}$, and (ii) $\int f(o, \tilde{o})P_0(d\tilde{o}) = 0$ for all $o \in \s{O}$. For such processes, we have the following result.
\begin{lemma}\label{lemma:u_process}
Suppose $\left\{S_n(f) : f \in \s{F}\right\}$ be a $P_0$-degenerate $U$-process. If $F$ is an envelope function for $\s{F}$, then we have that
\[ \frac{1}{[n(n-1)]^{1/2}}E_0\left[ \sup_{f \in \s{F}} |S_n(f) | \right]\ \lesssim\ \|F \|_{P_0 \times P_0, 2} \int_0^1 \left[ 1 + \log  \sup_Q N(\varepsilon \|F\|_{Q, 2},\s{F}, L_2(Q)) \right] d\varepsilon  \ .\]
\end{lemma}
\begin{proof}
We let $\d{T}_n f := \tfrac{1}{n(n-1)}\sum_{i \neq j} f(O_i, O_j)$, and also define $\vartheta_n := \tfrac{1}{4} \sup_{f \in \s{F}}\| f\|_{\d{T}_n, 2}$, $\tau_n := \| F \|_{\d{T}_n,2}$ and $J_n(s) := \int_0^s \log N(\varepsilon, \s{F}, d_{\d{T}_n, 2, F})\, d\varepsilon$, where 
\[ d_{\d{T}_n, 2, F}(f, g) := \left[ \frac{\d{T}_n (f - g)^2}{\d{T}_n F ^2} \right]^{1/2} = \frac{\| f - g \|_{\d{T}_n, 2}}{\|F\|_{\d{T}_n,2}} \ .\] 
Theorem 6 of \cite{nolan1987uprocess} then states that
\[ \frac{1}{[n(n-1)]^{1/2}}E_0\left[ \sup_{f \in \s{F}} |S_n(f) | \right]\ \lesssim\ E_0\left[ \vartheta_n + \tau_n J_n(\vartheta_n / \tau_n)\right] \ . \]
Now, we note that 
\[J_n(s)\ =\ \int_0^s \log N(\varepsilon \|F\|_{\d{T}_n, 2},\s{F}, L_2(\d{T}_n))\, d\varepsilon\ \leq\ \int_0^s\sup_Q \log N(\varepsilon \|F\|_{Q, 2},\s{F}, L_2(Q))\, d\varepsilon
\ ,\]
where the supremum is taken over all finite, discrete $Q$ such that $Q F > 0$. Next, since $\vartheta_n \leq \tau_n$, we have
\[  E_0\left[ \vartheta_n + \tau_n J_n(\vartheta_n / \tau_n)\right]\ \leq\ E_0\left(\tau_n\right)\left[ 1 +  \int_0^1\sup_Q \log N(\varepsilon \|F\|_{Q, 2},\s{F}, L_2(Q))\, d\varepsilon \right] .\]
By Jensen's inequality, we have that $E_0\left(\tau_n\right) \leq \| F \|_{P_0 \times P_0, 2}$, which then implies the claimed result.

\end{proof}

\subsection*{Proof of Theorem 1}

We use Theorem 1 of WC for both the pointwise and uniform consistency statements. Since $F_n$ is the empirical distribution function, $\sup_{a_0 \in \s{A}}| F_n(a_0) - F_0(a_0)| \inprob 0$ by the Glivenko-Cantelli Theorem. Hence, we only need to show that $\sup_{a_0 \in \s{A}}|\Gamma_{n}(a_0) - \Gamma_0(a_0)| \inprob 0$. 

We first establish that $\{ \phi_{\infty, a_0}^* : a_0 \in \s{A}\}$ is a $P_0$-Donsker class. The class $\{o\mapsto I_{(-\infty, a_0]}(a): a_0 \in \s{A}\}$ is a VC class and hence also $P_0$-Donsker. Since $\mu_{\infty}$ is a bounded, fixed function, $\{o\mapsto I_{(-\infty, a_0]}(a) \mu_{\infty}(a, w) : a_0 \in \s{A}\}$ is also $P_0$-Donsker, which implies that $\{ o\mapsto\int_{-\infty}^{a_0} \mu_{\infty}(a, w) F_0(da) : a_0 \in \s{A}\}$ is $P_0$-Donsker by Lemma~\ref{lemma:marginal_class}. Hence, by the permanence properties of Donsker classes, we find that $\{ \phi_{\infty, a_0}^* : a_0 \in \s{A}\}$ is a $P_0$-Donsker class and thus that $\sup_{a_0 \in \s{A}} |\d{P}_n \phi_{\infty,a_0}^*| = \bounded(n^{-1/2})$.

We first focus on studying remainder term $R_{n,a_0,1}$, which can be uniformly bounded by
\begin{align*}
\sup_{a_0 \in\s{A}}|R_{n,a_0, 1}|\ \leq&\ \ \iint_{\s{S}_1}\left|\mu_{n}(a, w) - \mu_{\infty}(u, w)\right| \left| 1 - \frac{g_0(a, w)}{g_{n}(a,w)}\right| F_0(da) Q_0(dw) \\
&\ \ \quad+ \iint_{\s{S}_2}\left|\mu_{n}(a, w) - \mu_0(u, w)\right| \left| 1 - \frac{g_{\infty}(a, w)}{g_{n}(a,w)}\right| F_0(da)Q_0(dw)\\
&\ \ \quad+ \iint_{\s{S}_3}\left|\mu_{n}(a, w) - \mu_\infty(u, w)\right| \left| 1 - \frac{g_{\infty}(a, w)}{g_{n}(a,w)}\right| F_0(da)Q_0(dw)\\
\leq&\ \ K_1^{-1}\left[P_0(\mu_{n} - \mu_{\infty})^2 P_0(1 - g_0/g_{n})^2\right]^{1/2} +K_1^{-1}\left[P_0(\mu_{n} - \mu_{0})^2 P_0(1 - g_{\infty}/g_{n})^2\right]^{1/2} \\
&\ \ \quad+K_1^{-1}\left[P_0(\mu_{n} - \mu_{\infty})^2 P_0(1 - g_{\infty}/g_{n})^2\right]^{1/2}\ .
\end{align*}
By assumption, $P_0(\mu_{n} - \mu_{\infty})^2 = \fasterthan(1)$, and since $g_{n}$ is eventually bounded uniformly above and away from zero almost surely,  $P_0(1 - g_{\infty}/g_{n})^2 = \fasterthan(1)$ as well. Also, $ P_0(1 - g_0/g_{n})^2 = \bounded(1)$ and $P_0(\mu_{n} - \mu_{0})^2 = \bounded(1)$ since $\mu_{n}$, $g_{n}$, $\mu_0$ and $g_0$ are all bounded for $n$ large enough. Hence, $\sup_{a_0 \in\s{A}}|R_{n,a_0, 1}|  = \fasterthan(1)$.

For the remainder term $R_{n,a_0,2}$, we define the stochastic process $\{\d{G}_n \phi_{\mu, g, a_0}' : \mu \in \s{F}_{0}, g \in \s{F}_{1}, a_0 \in \s{A} \}$. We will use  Lemma 4 of WC to establish that $\sup_{a_0 \in \s{A}} |n^{1/2}R_{n,a_0,2}| = \fasterthan(1)$. In their notation, we set $\s{U} := \s{A}$, equipped with the usual Euclidean norm, and $\s{F} = \s{F}_{0} \times \s{F}_{1}$, equipped with the product $L_2(P_0)$ semi-metric $d((\mu, g), (\tilde{\mu}, \tilde{g})) = [P_0(\mu - \tilde{\mu})^2]^{1/2} + [P_0(g-\tilde{g})^2]^{1/2}$. Application of this result requires showing that the process is uniformly asymptotically $\rho$-equicontinuous for $\rho$ the product semi-metric. This would be implied if the class $\{\phi_{\mu, g, a_0}' : \mu \in \s{F}_{0}, g \in \s{F}_{1}, a_0 \in \s{A}\}$ were $P_0$-Donsker. Note that condition (A1) implies that $\s{F}_{0}$ and $\s{F}_{1}$ are $P_0$-Donsker classes by Theorem 2.5.2 of \cite{van1996weak}. Since $\{o\mapsto I_{(-\infty, a_0]}(a) : a_0 \in \s{A}\}$ is a $P_0$-Donsker, as established above, the classes $\{o\mapsto \int I_{(-\infty, a_0]}(a) \mu( a, \tilde{w})Q_0(d\tilde{w}) : \mu \in \s{F}_{0}, a_0 \in \s{A}\}$ and $\{o\mapsto \int_{-\infty}^{a_0}\mu(a, w) F_0(da)  : \mu \in\s{F}_{0}, a_0 \in \s{A}\}$ are also $P_0$-Donsker by Lemma~\ref{lemma:marginal_class}. Since $\s{F}_{1}$ is bounded below, the class $\{o\mapsto I_{(-\infty, a_0]}(a) [y - \mu(a, w)] / g(a,w) : \mu \in \s{F}_{0}, g \in \s{F}_{1}, a_0 \in \s{A}\}$ is also $P_0$-Donsker. This then yields that the original class is $P_0$-Donsker. The second requirement of Lemma 4 of WC is satisfied by assumption.

Finally, we analyze the remainder term $R_{n,a_0,3}$, which itself has three components, as decomposed before the presentation of Lemma~\ref{lemma:marginal_class}. Its second component is an ordinary empirical process involving function classes discussed in the preceding paragraph. Using these results yields the second component to be $\bounded(n^{-3/2})$.  Its third sub-component is a bias term which, in view of the uniform boundedness of $\mu_n$, is $\bounded(n^{-1})$. Its first sub-component is a $P_0$-degenerate $U$-process as defined above, to which we will apply Lemma~\ref{lemma:u_process}. The function $\gamma_{\mu_n, a}$ is contained in the class $\big\{(a_1, w_1, a_2, w_2) \mapsto \gamma_{\mu, a_0}(a_1, w_1, a_2, w_2) : a_0 \in \s{A} ,\mu \in \s{F}_{0}\big\}$. As we discuss in more detail below, by Lemma~\ref{lemma:marginal_class} and Lemma 5.1 of \cite{van2006survival}, and in view of condition (A1), this class has uniform entropy bounded up to a constant by $\varepsilon^{-V/2}-\log\varepsilon$ relative to a constant envelope. Therefore, Lemma~\ref{lemma:u_process} implies that
\[ E_0 \left[ \sup_{\mu \in \s{F}_{0}, a_0 \in \s{A}} \left| \sum_{i \neq j}\gamma_{\mu, a_0}(O_i, O_j)\right|\right] \lesssim [n (n-1)]^{1/2} \ .\]
Therefore, the first sub-component of $R_{n,a_0,3}$ is $\bounded(n^{-1})$. Thus, we have that $\sup_{a_0 \in \s{A}} |R_{n,a_0,3}| = \bounded(n^{-1})$. 

In conclusion, we have shown that conditions (A1)--(A3) imply that all three remainder terms are controlled, so that $\sup_{a_0 \in \s{A}} | \Gamma_n(a_0) - \Gamma_0(a_0) | \inprob 0$. \qed

\subsection*{Proof of Theorem 2}

We will use Theorem 4 of WC to establish Theorem 2 stated in the main text. In what follows, we verify conditions (B1)--(B5) and (A4)--(A5) of WC, which we refer to as (WC.B1), (WC.B2) and so on.

\paragraph{Conditions (WC.B1) and (WC.B2).} Define pointwise $I_{a_0, u}(a) := I_{(-\infty, a_0 + u]}(a) - I_{(-\infty, a_0]}(a)$ and $g_{a_0,u}(o) := [\phi_{\infty, a_0 + u}^*(o) - \phi_{\infty,a_0}^*(o)] - \theta_0(a) I_{a_0, u}(a)$. Since $F_0$ is by assumption strictly increasing at $a$, we then have that
\begin{align*}
g_{a_0,u}(o)\  =&\ \ I_{a_0,u}(a) \left[ \frac{y - \mu_{\infty}(a,w)}{g_{\infty}(a,w)}  + \theta_{\infty}(a) - \theta_0(a)\right] + \int I_{a_0,u}(v) \mu_{\infty}(v, w)F_0(dv) \\
&\ \ \quad-[ \Gamma_{\infty}(a_0 + u) -\Gamma_{\infty}(a_0)]- [ \Gamma_0(a_0 + u) -\Gamma_0(a_0)] + [F_0(a_0 + u) - F_0(a_0)]\ ,
\end{align*}
where we define $\Gamma_{\infty}(a_0):=\int_{-\infty}^{a_0} \theta_{\infty}(a)F_0(da)$.

The class $\s{I}_R = \{o\mapsto I_{a_0, u}(a) : |u| \leq R\}$ is a VC class of functions by a slight extension of Example 2.6.1 of \cite{van1996weak}. Its envelope function is $J_{a_0, u}:a\mapsto I_{[0,R]}(|a - a_0|)$, and hence, we have that $\sup_Q \log N(\varepsilon \| J_R\|_{Q,2}, \s{I}_R, L_2(Q)) \lesssim  -\log(\varepsilon)$ by Theorem 2.6.7 of \cite{van1996weak}. The class $\{o\mapsto \int I_{a_0,u}(v) \mu(v, w) F_0(dv)  : |u|\leq R\}$ thus satisfies the same inequality by Lemma~\ref{lemma:marginal_class}. The classes $\{\Gamma_{\infty}(a_0 + u) -\Gamma_{\infty}(a_0) : |u| \leq R\}$, $\{\Gamma_0(a_0 + u) -\Gamma_0(a_0) : |u| \leq R\}$ and $\{ F_0(a_0 + u) - F_0(a_0): |u| \leq R\}$ are sets of constants not depending on the data, bounded up to a constant by $R$ for $R$ small enough since $\Gamma_0$  and $F_0$ are continuously differentiable in a neighborhood of $a_0$. Hence, they also have uniform entropy bounded up to a constant by $-\log(\varepsilon)$. Finally, the class $\s{G}_R$ is a linear combination of the above classes, and so, by Lemma 5.1 of \cite{van2006survival}, $\s{G}_R$ satisfies that $\sup_Q \log N(\varepsilon \| G_R\|_{Q,2}, \s{G}_R, L_2(Q)) \lesssim -\log(\varepsilon)$ as well. This verifies condition (WC.B1).

Since $\Gamma_0$, $\Gamma_{\infty}$ and $F_0$ are continuously differentiable in a neighborhood of $a_0$, an envelope function for the class $\s{G}_R = \{ g_{a_0, u} : |u| \leq R\}$ is
\[ G_R:o\mapsto J_{a_0, R}(a)\left| \frac{y - \mu_{\infty}(a,w)}{g_{\infty}(a,w)}  + \theta_{\infty}(a) - \theta_0(a)\right| + \int J_{a_0,R}(v) |\mu_{\infty}(v, w)| F_0(dv)  + K_1R\]
for some $0<K_1<+\infty$. Using the triangle inequality on $\|G_R\|_{P_0, 2}$, we first note that 
\[ E_{0} \left\{ J_{a_0, R}(A) \left[\frac{Y - \mu_{\infty}(A,W)}{g_{\infty}(A,W)}\right]^2 \right\} = E_{0} \left[ J_{a_0, R}(A) \left\{\frac{\sigma_0^2(A,W) + [ \mu_\infty(A,W) - \mu_0(A,W)]^2}{g_{\infty}(A,W)^2}\right\} \right]  \leq K_2 R \]
for some $0<K_2<+\infty$ by the boundedness of $\sigma_0^2$, $1/g_{\infty}$, $\mu_{\infty}$, $\mu_0$ and the conditional density $\pi_0$ in a neighborhood of $a_0$ uniformly over almost every $w$ under $Q_0$. Similar bounds hold for the other terms, yielding that $P_0 G_R^2 \lesssim R$ for all $R$ small enough, as required.

For the second requirement of (WC.B2), we note that $0\leq G_R(o) \leq J_{a_0, R} (|y|/C_1 + C_2) + C_3R$ for all $R$ small enough and some constants $0<C_1, C_2, C_3<+\infty$. By assumption, and in view of properties of probability densities, for all $R$ small enough and for all $\varepsilon > 0$, there is a $C_0$ such that $P_0[ J_{a_0, R}(A)|Y| > C_0]< \varepsilon$. This implies that for any $\eta > 0$, $P_0 G_R^2 I_{(\eta / R, \infty)}(G_R) < \varepsilon R$ for all $R$ small enough.

\paragraph{Condition (WC.B3).} Next, we need to study the covariance $\Sigma(s, t) := P_0 [\phi_{\infty,s}^* - \theta_0(a_0)\gamma_s^*]  [\phi_{\infty,t}^* - \theta_0(a_0)\gamma_t^*]$ for $s,t$ near $a_0$, where $\gamma_s^* :o\mapsto I_{(-\infty, s]}(a) - F_0(s)$, and where we may ignore any terms in the covariance function that are continuously differentiable in a neighborhood of $(a_0,a_0)$. We thus have
\[\phi_{\infty,s}^*(o) - \theta_0(a_0)\gamma_s^*(o)\ =\   [ \phi_{\infty,s}(o) - \Gamma_{\infty}(s) - \Gamma_0(s)] - \theta_0(a_0)[ I_{(-\infty, s]}(a)- F_0(s) ] \ . \]
Since $\Gamma_{\infty}$, $\Gamma_0$ and $F_0$ are continuously differentiable in a neighborhood of $a_0$, expanding $\Sigma(s, t)$, it is straightforward to see that we may focus on 
\begin{align*}
&\hspace{-0.25in}E_{0} \left[ \phi_{\infty,s}(O)  - \theta_0(a_0)I_{(-\infty, s]}(A) \right]  \left[ \phi_{\infty,t}(O) - \theta_0(a_0)I_{(-\infty, t]}(A)  \right] \\
&=\ E_{0} \left\{ I_{(-\infty, s \wedge t]} (A) \left[ \frac{Y - \mu_{\infty}(A,W)}{g_{\infty}(A,W)}  + \theta_{\infty}(A) - \theta_0(a_0)\right]^2 \right\}\\
&\ \ \ \ +E_{0} \left\{ I_{(-\infty, s]}(A) \left[ \frac{\mu_0(A, W) - \mu_{\infty}(A,W)}{g_{\infty}(A,W)}  + \theta_{\infty}(A) - \theta_0(a_0)\right] \right\} \int_{-\infty}^t \mu_{\infty}(a, W) F_0(da) \\
&\ \ \ \ + E_{0} \left\{I_{(-\infty, t]}(A)\left[ \frac{\mu_0(A,W) - \mu_{\infty}(A,W)}{g_{\infty}(A,W)}  + \theta_{\infty}(A) - \theta_0(a_0)\right] \right\}\int_{-\infty}^s \mu_{\infty}(a, W) F_0(da) \\
&\ \ \ \ +E_{0}\left[  \int_{-\infty}^s \mu_{\infty}(a, W) \,F_0(da)\int_{-\infty}^t \mu_{\infty}(a, W) F_0(da) \right].
\end{align*}
The bottom three lines are continuously differentiable in $(s, t)$  in a neighborhood of $(a_0, a_0)$ since $\mu_{\infty}$, $\mu_0$, $g_{\infty}$ and $g_0$ are all continuous in a neighborhood of $a_0$, uniformly over almost every $w$ under $Q_0$. As such, they do not contribute to the scale parameter of the limit.

By Fubini's theorem, the first line can be rewritten as
\begin{align*}
&\int_{-\infty}^{s \wedge t} \int E_{0} \left\{ \left[ \frac{Y - \mu_{\infty}(A,W)}{g_{\infty}(A,W)}  + \theta_{\infty}(a) - \theta_0(a_0)\right]^2 \middle| A = a, W = w\right\} g_0(a, w) Q_0(dw) F_0(da) \ .
\end{align*}
In view of (A5), this satisfies (WC.B3), and so, the limit distribution is $\left[4\theta_0'(a) \tilde\kappa_0(a) / f_0(a)^2\right]^{1/3} \d{W}$, where 
\begin{align*}
\tilde\kappa_0(a) &:= E_{0} \left[ E_{0} \left\{ \left[ \frac{Y - \mu_{\infty}(A,W)}{g_{\infty}(A,W)}  + \theta_{\infty}(A) - \theta_0(A)\right]^2 \middle| A = a, W = w\right\} g_0(a, W)\right] f_0(a)\ .
\end{align*}
We can thus simplify the scale factor $[4\theta_0'(a) \tilde\kappa_0(a) / f_0(a)^2]^{1/3}$ to $[4\theta_0'(a) \kappa_0(a) / f_0(a)]^{1/3}$, where $\kappa_0(a)$ is as defined in the statement of Theorem 2.

\paragraph{Conditions (WC.B4) and (WC.B5).}  Defining 
\[K_{n,j}(\delta) := n^{2/3} \sup_{|u|\leq \delta n^{-1/3}}  \left| R_{n,a +u,j} - R_{n,a,j} \right|,\]
for each $j$, we must show that $K_{n,j}(\delta) \inprob 0$ for all $\delta$ small enough and that, for some $\beta\in(1,2)$, $\delta\mapsto\delta^{-\beta}E[ K_{n,j}(\delta)]$ is decreasing for all $\delta$ small enough and $n$ large enough. For $K_{n,1}(\delta)$, by Fubini's theorem and taking supremum bounds, for $n$ large enough and $\delta$ small enough, we find that
\begin{align*}
 K_{n,1}(\delta)\ &\lesssim\ \delta n^{1/3} \sup_{|a - a_0| \leq \varepsilon_0} E_{0} \left[ \left| \mu_{n}(a, W) - \mu_0(a, W)\right| \left|g_{n}(a,W) - g_0(a, W) \right| \right] \\
 &=\  \delta n^{1/3} \sup_{|a - a_0| \leq \varepsilon_0} E_{0} \left[ I_{\s{S}_1}(a,W)\left| \mu_{n}(a, W) - \mu_\infty(a, W)\right| \left|g_{n}(a,W) - g_0(a, W) \right| \right] \\
 &\quad\quad + \delta n^{1/3} \sup_{|a - a_0| \leq \varepsilon_0} E_{0} \left[ I_{\s{S}_2}(a,W)\left| \mu_{n}(a, W) - \mu_0(a, W)\right| \left|g_{n}(a,W) - g_\infty(a, W) \right| \right] \\
&\quad\quad + \delta n^{1/3} \sup_{|a - a_0| \leq \varepsilon_0} E_{0} \left[ I_{\s{S}_3}(a,W)\left|\mu_{n}(a, W) - \mu_\infty(a, W)\right| \left|g_{n}(a,W) - g_\infty(a, W) \right| \right]\\
&\lesssim\ \delta n^{1/3} \left[ d(\mu_{n}, \mu_\infty; a_0, \varepsilon_0, \s{S}_1) +d(g_{n},g_\infty; a_0, \varepsilon_0, \s{S}_2) +d(\mu_{n}, \mu_\infty; a_0, \varepsilon_0, \s{S}_3) d(g_{n},g_\infty; a, \varepsilon_0, \s{S}_3)\right]\ .
\end{align*}
Hence, under conditions (A4a), (A4b) and (A4c),  $K_{n,1}(\delta) \inprob 0$ for each $\delta>0$. Furthermore, $\delta\mapsto \delta^{-\beta}E \left[ K_{n,1}(\delta)\right]$ is decreasing for any $\beta\in(1,2)$ by the assumed uniform boundedness of $\mu_{n}$, $g_{n}$, $\mu_\infty$, $g_\infty$, $\mu_0$ and $g_0$.

We will use Theorem 6 of WC to establish negligibility of the empirical process term $K_{n,2}(\delta)$, which requires checking conditions (WC.C1)--(WC.C4). Let $\omega := (\mu, g)$, which is contained in the product class $\s{P} := \s{F}_{0} \times \s{F}_{1}$ almost surely for all $n$ large enough, itself equipped with the semi-metric
\[ d^*:(\omega_1, \omega_2) \mapsto d(\mu_1, \mu_2; a_0, \varepsilon_0, \s{A} \times \s{W}) +d(g_1, g_2; a_0, \varepsilon_0, \s{A} \times \s{W})\ . \]
Next, we define $\s{G}_{R} := \left\{s_{u}(\mu, g)  : | u| \leq R,\mu  \in \s{F}_{0}, g \in \s{F}_{1}\right\}$,
where
\[ s_{u}(\mu, g):o \mapsto  I_{a_0, u}(a)\left[  \frac{ y - \mu(a,w)}{g(a,w)} + \int\mu(a,w)Q_0(dw)\right] + E_{0}[I_{a_0,u}(A) \mu(A,w)]\ .\]
We let $G_{R}$ be the envelope function for $\s{G}_{R}$ obtained by combining the assumed uniform bounds on $\s{F}_{0}$ and $\s{F}_{1}$ along with the natural envelope for $I_{a_0, u}$. Specifically, we have $G_R(y,a,w) =  I_{[0, R]}(|a - a_0|)\left(C_4|y| +C_5\right)$ for some $0<C_4, C_5 < \infty$.  For all $R$ small enough and some $V < 1$, $\s{G}_R$ is a Lipschitz transformation of the following classes:
\begin{itemize}
\item $\s{F}_{0}$, which has uniform entropy bounded up to a constant by $\varepsilon^{-V}$;
\item $\s{F}_{1}$, which has uniform entropy bounded up to a constant by $\varepsilon^{-V}$;
\item  $\{a \mapsto \int \mu(a,w)Q_0(dw) : \mu \in \s{F}_{0}\}$, which has uniform entropy bounded up to a constant by $\varepsilon^{-V}$ in view of Lemma~\ref{lemma:marginal_class};
\item  $\{ I_{a_0, u} :  | u| \leq R\}$, which has polynomial covering number;
\item $\{ w\mapsto \int I_{a_0,u}(a) \mu(a,w)F_0(da) : \mu \in \s{F}_{0}, |u| \leq R\}$, which has uniform entropy bounded up to a constant by $\varepsilon^{-V}-\log \varepsilon $ in view of Lemma 5.1 of \cite{van2006survival} and our Lemma~\ref{lemma:marginal_class};
\item  $\{ w\mapsto \int I_{a_0, u}(a) \mu_{\infty}(a,w)F_0(da):  | u| \leq R\}$, which has polynomial covering number;
\item the singleton class $\{y\}$, with covering number equal to one.
\end{itemize}
Thus, by Lemma 5.1 of \cite{van2006survival}, the $L_2$ covering number of $\s{G}_{R}$ relative to $G_R$ is bounded up to a constant by $\varepsilon^{-V} + \varepsilon^{-V/2}-\log \varepsilon$. Since $V < 2$, $\int_0^1 [\log \sup_Q N(\varepsilon \| G_{R}\|_{Q,2}, \s{G}_{R}, L_2(Q))]^{1/2} d\varepsilon$ is uniformly bounded above for all $R$ small enough with probability tending to one. This establishes (WC.C1). 

Existence of the conditional variance of $Y$ given $(A, W)$ and positivity of $f_0$ in a neighborhood of $a_0$ yields that $P_0 G_R^2 \leq c R$  and that, for any $\varepsilon>0$, there exists $\varepsilon'>0$ such that $P_0[ G_R^2 I_{(\varepsilon' / R, \infty)}(G_R)] \leq \varepsilon R$ for all $R$ small enough. Hence, condition (WC.C2) is satisfied.

Turning to (WC.C3), we note that $\{ P_0\left[ s_{u}(\mu, g) - s_{v}(\mu, g)\right]^2\}^{1/2}$ is bounded above by
\begin{align*}
&\left\{ \int\! \left[\int_{a_0 + v}^{a_0 + u} \mu(a,w) \,F_0(da) \right]^2\, Q_0(dw) \right\}^{1/2}\\
&\qquad+ \left[ \int_{a_0 + v}^{a_0 + u} \iint E_0\left\{\left[  \frac{ Y - \mu(a,w)}{g(a,w)} + \int \mu(a,w)Q_0(dw)\right]^2\middle|A=a,W=w\right\}  g_0(a,w) Q_0(dw) F_0(da)\right]^{1/2},
\end{align*}
and by the finite conditional second moment of $Y$ given $(A,W)$, the boundedness of $g_0$, the uniform boundedness of $\mu$ and $g$, and the positivity of $f_0$ near $a_0$, we find that $P_0[ s_{u}(\mu, g) - s_{v}(\mu, g)]^2 \lesssim |u -v|$ for all $u, v$ in a neighborhood of 0. Similarly, we can bound $\{P_0[ s_{u}(\mu_1, g_1) - s_{u}(\mu_2, g_2)]^2\}^{1/2}$ above by
\begin{align*}
&\left\{ \int \left[\int_{a_0}^{a_0 +v} \{ \mu_1(a, w) - \mu_2(a,w)\} \, F_0(da) \right]^2 \, dQ_0(w) \right\}^{1/2} \\
&\qquad + \left\{ \int_{a_0}^{a_0 +v} \left[\int \{ \mu_1(a, w) - \mu_2(a,w) \}\, Q_0(dw) \right]^2 \, F_0(da) \right\}^{1/2} \\
&\qquad +\left[ \int_{a_0}^{a_0 + v}\!\iint E\left\{\left[\frac{Y-\mu_2(a,w)}{g_1(a,w) g_2(a,w)} \{g_2(a, w) -g_1(a,w )\} \right]^2  \middle| A=a,W=w\right\}g_0(a,w) Q_0(dw) F_0(da) \right]^{1/2} \\ 
&\qquad +\left\{ \int_{a_0}^{a_0 + v}\!\iint \left[\frac{\mu_1(a, w) -\mu_2(a,w )}{g_1(a,w) g_2(a,w)}\right]^2 Q_0(dw) F_0(da) \right\}^{1/2} .
\end{align*}
We find that, for $v$ small enough, this is bounded up to a constant by
\[ |v|^{1/2}\left\{ \sup_{|a- a_0|\leq \varepsilon_0} \left[ E_{0} \{ \mu_1(a, W) - \mu_2(a,W)\} ^2\right]^{1/2} + \sup_{|a- a_0|\leq \varepsilon_0} \left[ E_{0} \{g_1(a, W) - g_2(a,W)\}^2\right]^{1/2} \right\},\]
as required. Finally, (WC.C4) is satisfied by assumption. 

For $K_{n,3}(\delta)$, we first note that (WC.B4) has already been shown to hold in the proof of Theorem 4, since 
\[n^{2/3}\sup_{|u| \leq \delta n^{-1/3}} |R_{n, a_0 +u, 3} - R_{n, a_0, 3}|\ \leq\ 2n^{2/3} \sup_{a_0 \in \s{A}} |R_{n,a_0,3}|\ =\  \bounded(n^{-1/3})\ .\]
We verify (WC.B5) for each of the three sub-components of $K_{n,3}(\delta)$ defined by the three sub-components of $R_{n,a_0,3}$. Due to the assumed boundedness of $|\mu_n|$, the contribution of the third component is bounded for all $\delta$ small enough up to a constant (not depending on $\delta$ or $n$) by $n^{-1/3}P_0\left(|A - a| \leq \delta n^{-1/3}\right) \lesssim n^{-2/3} \delta$, which satisfies (WC.B5). For the second component, by Lemma 4 of WC, $E_0\left[\sup_{|u| \leq \delta n^{-1/3}} | \d{G}_n I_{a, u} \mu_n| \right]\lesssim \delta^{1/2}$, and so, the expectation of the second component is bounded up to a constant by $\delta^{1/2}n^{-1}$ for all $\delta$ small enough and $n$ large enough, which is also sufficient for (WC.B5).

The first component requires controlling $\sum_{i \neq j} \gamma_{\mu_n, a_0, u}^*(O_i, O_j)$, where we define
\begin{align*}
\gamma_{\mu, a_0, u}^*(o_i, o_j)\ &:=\ I_{a_0, u}(a_i) \mu(a_i, w_j) + I_{a_0,u}(a_j) \mu(a_j, w_i) \\
&\qquad - \int \left[ I_{a_0, u}(a_i) \mu(a_i, w) + I_{a_0, u}(a_j) \mu(a_j, w)\right] Q_0(dw) \\
&\qquad - \int  I_{a_0, u}(a)\left[ \mu(a, w_i) + \mu(a, w_j)\right] F_0(da)+  2\iint I_{a_0, u}(a) \mu(a, w) F_0(da) Q_0(dw) \ .
\end{align*}
The function $\gamma_{\mu_n, a_0, u}^*$  falls in the class
$ \s{H}_{\delta} := \left\{\gamma_{\mu, a_0, u}^* : |u| \leq \delta, \mu \in \s{F}_{0} \right\}$.
Thus, $\{ \sum_{i \neq j} \gamma^*(O_i, O_j): \gamma^* \in \s{H}_{\delta} \}$ is a $P_0$-degenerate $U$-process.  By a similar argument as made above, the class $H_{\delta}$ has uniform entropy $\log \sup_Q N(\varepsilon \|H_{\delta}\|_{Q,2}, \s{H}_{\delta}, L_2(Q))$ bounded up to a constant by $\varepsilon^{-V/2}-\log\varepsilon$ relative to the envelope 
\[H_{\delta}:(a_1,w_1, a_2, w_2) \mapsto 2K_{\mu}I_{[0, \delta]}(|a_1- a_0|) + 2K_{\mu}I_{[0, \delta]}(|a_2- a_0|)  + 4K_{\mu} P_0\left(|A - a_0| \leq \delta\right).\]
Since $-V/2 > -1$ and $\| H_{\delta} \|_{P_0 \times P_0,2} \lesssim \delta^{1/2}$, Lemma~\ref{lemma:u_process} yields that
\[ n^{2/3}E_0 \left[ \sup_{\gamma^* \in \s{H}_{\delta}} \left|\frac{1}{n^2} \sum_{i \neq j} \gamma^*(O_i, O_j)\right| \right] \lesssim n^{-1/3}\delta^{1/2}\]
for all $\delta$ small enough. Hence, (WC.B5) is satisfied for this $U$-process term.

\paragraph{Conditions (WC.A4) and (WC.A5).} Condition (WC.A4) is trivially satisfied since the transformation used here is the empirical distribution function. Condition (WC.A5) was established in the proof of Theorem 1 under our conditions (A1)--(A3). We have now checked all the conditions of Theorem 4 of WC and verified the stated limit distribution in the course of checking condition (WC.B3). This concludes the proof. \qed

\subsection*{Proof of Theorem~3}

We first note that $n^{1/3} \left[ \theta_n^\circ(a_1) - \theta_0(a_1)\right] > \eta_1$ and  $n^{1/3} \left[ \theta_n^\circ(a_1) - \theta_0(a_1)\right] > \eta_2$ if and only if $\theta_n^\circ(a_1) > \theta_0(a_1) + n^{-1/3} \eta_1$ and $\theta_n^\circ(a_2) > \theta_0(a_2) + n^{-1/3} \eta_2$.  By Lemma~1 of WC, this holds if and only if the set of inequalities
\begin{align*}
&\sup\argmax_{v_1 \in \s{A}} \left\{ \left[\theta_0(a_1) + n^{-1/3} \eta_1\right] F_n(v_1) - \Gamma_n^\circ(v_1)\right\} < F_n^-\left(F_n(a_1)\right)\\
&\sup\argmax_{v_2 \in \s{A}} \left\{ \left[\theta_0(a_2) + n^{-1/3} \eta_2\right] F_n(v_2) - \Gamma_n^\circ(v_2)\right\} < F_n^-\left(F_n(a_2) \right)
\end{align*} holds true. Standard manipulation of the argmax (see the proof of Theorem~3 of WC) yields that this is equivalent to the set of inequalities
\begin{align*}
&\hat{v}_{n, a_1, \eta_1}:=\sup\argmax_{v_1 \in n^{1/3}(\s{A} - a_1)} \left\{ H_{n,a_1, \eta_1}(v_1)  + \sum_{j=1}^3 S_{n,a_1,\eta_1, j}(v_1) \right\}  <n^{1/3} \left[ F_n^-\left(F_n(a_1) \right) - a_1\right]\\
&\hat{v}_{n, a_2, \eta_2} :=\sup\argmax_{v_2 \in n^{1/3}(\s{A} - a_2)} \left\{ H_{n,a_2, \eta_2}(v_2)  + \sum_{j=1}^3 S_{n,a_2,\eta_2, j}(v_2) \right\}  <n^{1/3} \left[ F_n^-\left(F_n(a_2) \right) - a_2\right], \end{align*}
where we have defined the terms
\begin{align*}
H_{n,a,\eta}(v)\ &:=\ -W_{n,a}(v) + \left[ \eta f_0(a) \right] v - \left[ \tfrac{1}{2} f_0(a) \theta_0'(a) \right] v^2; \\
W_{n,a}(v)\ &:=\ n^{2/3} \left\{ \left[ \Gamma_n^\circ(a + n^{-1/3} v) - \Gamma_n^\circ(a) \right] - \left[ \Gamma_0(a + n^{-1/3} v) - \Gamma_0(a) \right] \right\}; \\
S_{n, a, \eta, 1}(v)\ &:=\ n^{1/3} \eta \left[F_n(a + n^{-1/3} v) - F_0(a +  n^{-1/3}v) \right];\\
S_{n, a, \eta, 2}(v)\ &:=\ n^{1/3}\eta \left[ F_0(a  + n^{-1/3}v) - F_0(a) - f_0(a) (n^{-1/3}v ) \right];\\
S_{n, a, \eta, 3}(v)\ &:=\ -n^{2/3} \left[ M_{0,a}(n^{-1/3}v)- \tfrac{1}{2}f_0(x)\theta_0'(x)(n^{-1/3} v)^{2}\right];\\
M_{0,a}(u)\ &:=\ [\Gamma_0(a + u) - \theta_0(a) F_0(a + u)] - [\Gamma_0(a) - \theta_0(a)F_0(a)] \ .
\end{align*}
We have that $\sup_{|v| \leq M} |S_{n, a, \eta, 1}(v)| = \fasterthan(1)$ for $(a, \eta) \in \{ (a_1, \eta_1), (a_2, \eta_2)\}$ and any $M \in (0, \infty)$ by uniform consistency of $F_n$, and similarly for $S_{n, a, \eta, 2}$ and $S_{n, a, \eta, 3}$ using the continuous differentiability of $F_0$ and differentiability of $\theta_0$ at $a_1$ and $a_2$. See the proof of Theorem~3 of WC for additional details.

The core of the argument is to demonstrate that $W_{n,a_1}$ and $W_{n,a_2}$ converge jointly (as processes) to independent Brownian motions $W_{a_1} = \kappa_0(a_1)^{1/2}Z_1$ and $W_{a_2} = \kappa_0(a_2)^{1/2} Z_2$, where $Z_1$ and $Z_2$ are two independent standard two-sided Brownian motions originating from zero. If this holds, it would follows that \[ \left\{ \begin{pmatrix} H_{n,a_1, \eta_1}(v)  + \sum_{j=1}^3 S_{n,a_1,\eta_1, j}(v) \\ H_{n,a_2, \eta_2}(v)  + \sum_{j=1}^3 S_{n,a_2,\eta_2, j}(v) \end{pmatrix} : |v| \leq M \right\} \rightsquigarrow \left\{ \begin{pmatrix} H_{a_1, \eta_1}(v)  \\ H_{a_2, \eta_2}(v)  \end{pmatrix} : |v| \leq M \right\} \]
in $\ell^{\infty}([-M, M]) \times\ell^{\infty}([-M, M])$ for all $M\in(0,\infty)$,  where $H_{a, \eta}(v) :=  -W_{a}(v) + \left[ \eta f_0(a) \right] v - \left[ \tfrac{1}{2} f_0(a) \theta_0'(a) \right] v^2$.  An adaptation of the argmax continuous mapping theorem (i.e.\ Theorem~3.2.2 of \cite{van1996weak}) and the arguments of Theorem~3 of WC imply that under the stated conditions, \[(\hat{v}_{n, a_1, \eta_1}, \hat{v}_{n, a_2, \eta_2}) \indist (\hat{v}_{a_1, \eta_1}, \hat{v}_{a_2, \eta_2})\ ,\] where $\hat{v}_{a,\eta} := \sup\argmax_{v \in \d{R}} H_{a, \eta}(v)$. Since $H_{a_1, \eta_1}$ and $H_{a_2, \eta_2}$ are independent, so are $\hat{v}_{a_1, \eta_1}$ and $\hat{v}_{a_2, \eta_2}$, and
\begin{align*}
&P_0\left( n^{1/3} \left[ \theta_n^\circ(a_1) - \theta_0(a_1)\right] > \eta_1,\ n^{1/3} \left[ \theta_n^\circ(a_1) - \theta_0(a_1)\right] > \eta_2\right)\\
 &\qquad=\ P_0\left( \hat{v}_{n, a_1, \eta_1} < n^{1/3} \left[ F_n^-\left(F_n(a_1) \right) - a_1\right],\  \hat{v}_{n, a_2, \eta_2} < n^{1/3} \left[ F_n^-\left(F_n(a_2) \right) - a_2\right]\right) \\
 &\qquad\longrightarrow\ P_0\left( \hat{v}_{a_1, \eta_1}  < 0,\ \hat{v}_{a_2, \eta_2}  < 0\right)\ =\ P_0\left( \hat{v}_{a_1, \eta_1}  < 0\right) P_0\left( \hat{v}_{a_2, \eta_2}  < 0\right)  \ .
\end{align*}
From there, standard manipulations of Brownian motion yield the result (see, for example, the proof of Theorem~3 of WC) applied to each $a_1$ and $a_2$ separately.

We now show that $W_{n,a_1}$ and $W_{n,a_2}$ converge jointly as processes to independent Brownian motions $W_{a_1} = \kappa_0(a_1)^{1/2}Z_1$ and $W_{a_2} = \kappa_0(a_2)^{1/2} Z_2$. We note that 
\[ \sup_{|v| \leq M} \left| W_{n,a}(v) -  \d{G}_n n^{1/6}\left( \phi_{\infty, a + v n^{-1/3}}^* - \phi_{\infty, a}^*\right) \right| \inprob 0\]
 for $a \in \{a_1, a_2\}$ by our derivations in the proof of Theorem~2. Furthermore, since $F_0$ is Lipschitz in neighborhoods of $a_1$ and $a_2$, we have 
\[\sup_{|v| \leq M} \left| \d{G}_n n^{1/6}\left( \phi_{\infty, a + v n^{-1/3}}^* - \phi_{\infty, a}^*\right) -  \d{G}_n n^{1/6} \phi_{\infty, a, v n^{-1/3}}^\dagger \right|\inprob 0\]
 for $a \in \{a_1, a_2\}$, where we define $\phi_{\infty, a_0, v n^{-1/3}}^\dagger : (y, a, w) \mapsto I_{a_0, v n^{-1/3}}(a) \left[ \frac{y - \mu_{\infty}(a,w)}{g_{\infty}(a,w)} +\int \mu_\infty(a,\tilde{w}) \, Q_0(d\tilde{w})\right]$.
Now, for all $n > 2M / |a_2 - a_1|$, $\left[a_1 - v n^{-1/3}, a_1 + v n^{-1/3}\right] \cap \left[a_2 - u n^{-1/3}, a_2 + u n^{-1/3}\right] = \emptyset$ for all $u,v$ such that $|u| \leq M$ and $|v| \leq M$. Thus, for all such $n$ and $u,v$, $ \d{G}_n n^{1/6}\phi_{\infty, a_1, u n^{-1/3}}^\dagger$ and $ \d{G}_n n^{1/6}\phi_{\infty, a_2, u n^{-1/3}}^\dagger$ depend on disjoint sets of the observations $O_1, \dotsc, O_n$, which implies that they are independent. This implies that $\{ \d{G}_n n^{1/6} \phi_{\infty, a_1, v n^{-1/3}}^\dagger : |v| \leq M\}$ and $\{ \d{G}_n n^{1/6} \phi_{\infty, a_2, v n^{-1/3}}^\dagger : |v| \leq M\}$ are independent for all $n > 2M / |a_2 - a_1|$, and hence asymptotically independent for all $M \in (0, \infty)$. Furthermore, the proof of Theorem~4 of WC demonstrates that conditions WC.B1--WC.B4 and WC.A4--WC.A5 imply that the processes $\{ \d{G}_n n^{1/6} \phi_{\infty, a_1, v n^{-1/3}}^\dagger  : |v| \leq M\}$ for $a \in \{a_1, a_2\}$ converge marginally as processes in $\ell^{\infty}([-M, M])$ to $W_{a}$. Therefore, by Example 1.4.6 of VW, the two processes converge jointly to independent Brownian motions. We have thus found that
\[ \sup_{|v| \leq M} \left| W_{n,a}(v) -  \d{G}_n n^{1/6} \phi_{\infty, a, v n^{-1/3}}^\dagger \right| \inprob 0\]
 for $a \in \{a_1, a_2\}$, where $\{ \d{G}_n n^{1/6} \phi_{\infty, a_1, v n^{-1/3}}^\dagger : |v| \leq M\}$ and $\{ \d{G}_n n^{1/6} \phi_{\infty, a_2, v n^{-1/3}}^\dagger : |v| \leq M\}$  converge jointly to independent Brownian motions. Therefore, $\{ W_{n,a_1}(v) : |v| \leq M\}$ and $\{ W_{n,a_2}(v) : |v| \leq M\}$ converge jointly to this same limit. 
 \qed
 
 \subsection*{First-order expansion of cross-fitted estimator}

Next, we provide proofs of Theorems~4 and~5 for the estimator $\theta_n^\circ$, which is based upon the cross-fitted nuisance estimators $\mu_{n,v}$ and $g_{n,v}$. We recall that $\s{T}_{n,v}$ is the training set for fold $v$, that is, the subset of observations in $\{O_1, O_2,\dotsc, O_n\}$  used to estimate $\mu_{n,v}$ and $g_{n,v}$, and $\s{V}_{n,v}$ is the vector of indices of the validation set for fold $v$, that is, $\{1, 2,\dotsc, n\} \backslash \{i : O_i \in \s{T}_{n,v}\}$. We note that $\cup_{v=1}^V \s{V}_{n,v} = \{1, 2,\dotsc, n\}$ and $\s{V}_{n,v} \cap \s{V}_{n, u} = \emptyset$ for each $u,v$. We denote by $\d{P}_n^v$ the empirical measure corresponding to the observations with indices in $\s{V}_{n,v}$, and we let $Q_n^v$ and $F_n^v$ denote the marginal empirical measures of $\{ W_i : i \in \s{V}_{n,v}\}$ and $\{A_i : i \in \s{V}_{n,v}\}$.

Before proving our results, we derive a first-order expansion of $\Gamma_{n}^\circ(a_0)$ that we will rely upon. 
\begin{lemma}
If condition (A3) holds, then $\Gamma_{n}^\circ(a_0) -  \Gamma_0(a_0) = \d{P}_n\phi_{\infty, a_0}^* + R_{n,a_0}^\circ$, where $R_{n,a_0}^\circ = R_{n,a_0,1}^\circ+R_{n,a_0,2}^\circ+R_{n,a_0,3}^\circ$ for
\begin{align*}
R_{n,a_0,1}^\circ\ &:=\ \frac{1}{V}\sum_{v=1}^V \iint_{-\infty}^{a_0} \left[\mu_{n,v}(a, w) - \mu_0(u, w)\right] \left[ 1 - \frac{g_0(a, w)}{g_{n,v}(a,w)}\right] F_0(da) Q_0(dw)\ ,\\
R_{n,a_0,2}^\circ\ &:=\ \frac{1}{V}\sum_{v=1}^V (\d{P}_n^v - P_0) ( \phi_{\mu_{n,v}, g_{n,v},a_0,v}' - \phi_{\mu_\infty, g_\infty, a_0}')\ ,\ \ R^\circ_{n,a_0,3}\ :=\  \frac{1}{V}\sum_{v=1}^V R^{\circ,v}_{n,a_0,3}\ ,\\
R_{n,a_0,3}^{\circ,v}\ &:=\ \frac{1}{2N^2} \sum_{\stackrel{i, j \in \s{V}_{n,v}}{i \neq j}} \gamma_{\mu_{n,v}, a_0}(O_i, O_j)+\frac{1}{2N^{3/2}} \d{G}_n^v \gamma_{\mu_{n,v}, a_0}  + \frac{1}{N}E_{0} \left[I_{(-\infty, a_0]}(A) \mu_{n,v}(A, W) \left[ 1 - \frac{1}{g_0(A, W)} \right] \right],
\end{align*}
where $\gamma_{\mu, a_0}$ is as defined in the proof of Theorem~1.
\end{lemma}
\begin{proof}
We define
\begin{align*}
\phi_{n,a_0,v}(y, a, w)\ :=&\ \  I_{(-\infty, a_0]}(a) \left[  \frac{ y - \mu_{n,v}(a,w)}{g_{n,v}(a,w)}+ \int \mu_{n,v}( a, \tilde{w}) Q_n^v(d\tilde{w})\right]\\
&\ \ \quad+ \int_{-\infty}^{a_0} \mu_{n,v}(a, w) F_n^v(da) -  \iint_{-\infty}^{a_0} \mu_{n,v}(a, \tilde{w}) F_{n}^v(da) Q_{n}^v(d\tilde{w})\ ,\\
\phi_{\mu, g, a_0}(y, a, w)\ :=&\ \ I_{(-\infty, a_0]}(a) \left[\frac{ y - \mu(a,w)}{g(a,w)}+ \int \mu( a, \tilde{w})  Q_0(d\tilde{w})\right]\\
&\ \ \quad+ \int_{-\infty}^{a_0}\mu(a, w)  F_0(da) - \iint_{-\infty}^{a_0} \mu(a, \tilde{w}) F_0(da) Q_0(d\tilde{w})\ ,
\end{align*}
so that $\Gamma_{n}^\circ(a_0) = \frac{1}{V} \sum_{v=1}^V\d{P}_n^v \phi_{n,a_0,v}$. Writing $\phi_{\infty, a_0} := \phi_{\mu_\infty, g_\infty, a_0}$, in view of condition (A3), we have that
\[ P_0 \phi_{\infty,a_0}  = \iint_{-\infty}^{a_0}\left[\mu_{\infty}(a,w) - \mu_0(a,w)\right] \left[1 - \frac{g_0(a,w)}{g_{\infty}(a,w)}\right] F_0(da) \, Q_0(dw) + \Gamma_0(a_0) = \Gamma_0(a_0)\ . \]
The expansion follows by adding and subtracting terms.
%
%
\end{proof}

\subsection*{Proof of Theorem~4}

As before, we use Theorem~1 of WC for both the pointwise and uniform consistency statements. We only need to show that $\sup_{a_0 \in \s{A}}|\Gamma_n^\circ(a_0) - \Gamma_0(a_0)| \inprob 0$. 

In the proof of Theorem~1, we established that $\sup_{a_0 \in \s{A}} |\d{P}_n \phi_{\infty,a_0}^*| = \bounded(n^{-1/2})$. Since the analysis of the remainder term $R_{n,a_0,1}^\circ$ is entirely analogous to that provided in the proof of Theorem~1, we begin by looking at the remainder term $R_{n,a_0,2}^\circ$ instead. We define $\s{F}_{n, v} := \{\phi_{\mu_{n,v}, g_{n,v}, a_0}'  - \phi_{\mu_{\infty}, g_{\infty}, a_0}' : a_0 \in \s{A} \}$. We then have $\sup_{a_0 \in \s{A}} \left|R_{n,a_0,2}^\circ\right| \leq n^{-1/2} \max_v \sup_{f \in \s{F}_{n, v}} \left| \d{G}_n^v f\right|$.  We will demonstrate that $E_0 \left[ \sup_{f \in \s{F}_{n,v}} \left| \d{G}_n^v f \right| \right] = \fasterthandet(1)$ using Theorem 2.14.2 of VW. By the tower property,
\[ E_0 \left[\sup_{f \in \s{F}_{n,v}} \left| \d{G}_n^v f \right| \right]  = E_0 \left\{ E_0\left[\sup_{f \in \s{F}_{n,v}} \left| \d{G}_n^v f  \right| \Bigg| \s{T}_{n,v} \right] \right\} \ .\]
Here, the inner expectation is with respect to the distribution of the observations in the validation sample $\s{V}_{n,v}$ given the training sample $\s{T}_{n,v}$, while the outer expectation is with respect to the observations in the training sample. Since $\mu_{n,v}$ and $g_{n,v}$ are constructed only using $\s{T}_{n,v}$, they are fixed when conditioning on $\s{T}_{n,v}$. We note that with probability one, for all $n$ large enough,
\begin{align*} 
&\left| \phi_{\mu_{n,v}, g_{n,v}, a_0}'(o)  - \phi_{\mu_{\infty}, g_{\infty}, a_0}'(o)\right| \\
&=\ \left| I_{(-\infty, a_0]}(a)\left\{  \frac{ y - \mu_{n,v}(a,w)}{g_{n,v}(a,w)}- \frac{ y - \mu_\infty(a,w)}{g_\infty(a,w)} + \int\left[ \mu_{n,v}( a, \tilde{w}) - \mu_{\infty}(a, \tilde{w})\right] \,Q_0(d\tilde{w})\right\} \right.\\
&\qquad\left. + \int_{-\infty}^{a_0} \left[\mu_{n,v}(\tilde{a}, w) - \mu_\infty(\tilde{a},w) \right]\, F_0(d\tilde{a}) \right| \\
&\leq\   I_{(-\infty, a_0]}(a)\left| \frac{ y - \mu_{n,v}(a,w)}{g_{n,v}(a,w)}- \frac{ y - \mu_\infty(a,w)}{g_\infty(a,w)}\right| +I_{(-\infty, a_0]}(a)\int\left| \mu_{n,v}( a, \tilde{w}) - \mu_{\infty}(a, \tilde{w})\right| \,Q_0(d\tilde{w}) \\
&\qquad + \int_{-\infty}^{a_0} \left|\mu_{n,v}(\tilde{a}, w) - \mu_\infty(\tilde{a},w) \right|\, F_0(d\tilde{a}) \\
&\leq\ \left|  \frac{ y - \mu_{n,v}(a,w)}{g_{n,v}(a,w)}- \frac{ y - \mu_\infty(a,w)}{g_\infty(a,w)}\right| +\int\left| \mu_{n,v}( a, \tilde{w}) - \mu_{\infty}(a, \tilde{w})\right| \,Q_0(d\tilde{w})  + \int \left|\mu_{n,v}(\tilde{a}, w) - \mu_\infty(\tilde{a},w) \right|\, F_0(d\tilde{a}) \\
&\leq\ \left| \left[ y - \mu_\infty(a, w) \right] \left[ \frac{1}{g_{n,v}(a,w)} - \frac{1}{g_{\infty}(a,w)}\right]\right| + \frac{1}{g_{n,v}(a,w)} \left| \mu_{n,v}(a,w) - \mu_\infty(a,w)\right| \\
&\qquad +\int\left| \mu_{n,v}( a, \tilde{w}) - \mu_{\infty}(a, \tilde{w})\right| \,Q_0(d\tilde{w})+ \int \left|\mu_{n,v}(\tilde{a}, w) - \mu_\infty(\tilde{a},w) \right|\, F_0(d\tilde{a}) 
 \end{align*}
 for all $a_0 \in \s{A}$. We then define $F_{n,v}$ pointwise by taking $F_{n,v}(o)$ to be the sum of terms on the right-hand side of the last inequality above. $F_{n,v}$ ultimately serves as an envelope function for $\s{F}_{n,v}$, so that 
 by Theorem 2.14.1 of VW we have that, for $n$ large enough,
 \[ E_0\left[\sup_{f \in \s{F}_{n,v}} \left| \d{G}_n^v f  \right|\Bigg| \s{T}_{n,v} \right]\  \leq\ C \| F_{n,v}\|_{P_0, 2} J(1, \s{F}_{n,v}) \]
 for a universal constant $C\in(0,\infty)$, where $J(1, \s{F}_{n,v})$ is the uniform entropy integral as defined in Chapter 2.14 of VW. The class $\s{F}_{n,v}$ is a convex combination of the class $\{ I_{(-\infty, a_0]}(a) : a_0 \in \s{A}\}$, which is well-known to be VC and hence possess polynomial covering numbers, plus the class $\{  \int  I_{(-\infty, a_0]}(a) \mu(a, w) \, F_0(da) : a_0 \in \s{A}\}$ for $\mu =\mu_\infty$ and $\mu = \mu_{n,v}$ (both of which are fixed functions), so in view of Lemma~\ref{lemma:marginal_class} this class also possesses polynomial covering numbers. Thus, $J(1, \s{F}_{n,v})$ is uniformly bounded for all $n$ and $v$.  It follows then that, for some constant $C'\in(0,\infty)$ and large enough $n$,
 \[  E_0 \left[\sup_{f \in \s{F}_{n,v}} \left| \d{G}_n^v f \right| \right] \ \leq\ C' E_0 \left[  \| F_{n,v}\|_{P_0, 2}  \right]. \]
It remains to demonstrate that $\max_v  E_0 \left[  \| F_{n,v}\|_{P_0, 2}  \right] \longrightarrow 0$. We have that $\| F_{n,v}\|_{P_0,2}$ is bounded above by 
\begin{align*}
&3 K_1^{-1} \left\{ \int \left[ \mu_{n,v}(a,w) - \mu_\infty(a, w)\right]^2 \, dP_0(o) \right\}^{1/2}+\left\{\int  \sigma_0^2(a, w) \left[ \frac{1}{g_{n,v}(a,w)} - \frac{1}{g_{\infty}(a, w)}\right]^2 \, dP_0(o) \right\}^{1/2} \\
&\quad\ + \left\{ \int  \left[\mu_0(a, w) - \mu_\infty(a, w) \right]^2\left[ \frac{1}{g_{n,v}(a,w)} - \frac{1}{g_{\infty}(a, w)}\right]^2 \, dP_0(o)\right\}^{1/2} \\
&\leq\ 3 K_1^{-1}  \left\{ \int \left[ \mu_{n,v}(a,w) - \mu_\infty(a, w)\right]^2   \, dP_0(o)\right\}^{1/2}+ (K_0 + K_3) \left\{\int \left[ \frac{1}{g_{n,v}(a,w)} - \frac{1}{g_{\infty}(a, w)}\right]^2 \, dP_0(o)\right\}^{1/2}.
\end{align*}
Both terms tend to zero in probability by condition (B2), and since all involved terms are uniformly bounded by condition (B1), they also tend  to zero in expectation. Therefore, we have that $\max_v  E_0 \left[  \| F_{n,v}\|_{P_0, 2}  \right] \longrightarrow 0$, which implies that $\sup_{a_0 \in \s{A}} \left|R_{n,a_0,2}^\circ\right| = \fasterthan(n^{-1/2})$.

Finally, we analyze the remainder term $R_{n,a_0,3}^\circ$, which itself has three subcomponents, as decomposed before the presentation of Lemma~\ref{lemma:marginal_class}. Its second subcomponent, $\frac{1}{2N^{3/2}} \d{G}_n^v \gamma_{\mu_{n,v}, a_0}$, is an ordinary empirical process and can be analyzed in a manner analogous to that used for $R_{n,a_0,2}^\circ$. Doing so yields that the second subcomponent is $\bounded(n^{-3/2})$. The third subcomponent of $R_{n,a_0,3}^\circ$ is a bias term which, in view of the uniform boundedness of $\mu_{n,v}$ and $g_0^{-1}$, is $\bounded(n^{-1})$. The first subcomponent of $R_{n,a_0,3}^\circ$ is a $P_0$-degenerate $U$-process as defined above. We denote $\s{F}_{n,v}' := \{ \gamma_{\mu_{n,v}, a_0} : a_0 \in\s{A} \}$ and $S_{n,v}(\gamma) :=  \sum_{i, j \in \s{V}_{n,v}.i \neq j} \gamma(O_i, O_j)$. We then have 
$\sup_{a_0 \in\s{A}} \left| \sum_{i, j \in \s{V}_{n,v},i \neq j} \gamma_{\mu_{n,v}, a_0}(O_i,O_j) \right| = \sup_{\gamma \in \s{F}_{n,v}'} \left|S_{n,v}(\gamma) \right|.$
As before, we begin by conditioning on $\s{T}_{n,v}$ so that $\mu_{n,v}$ is a fixed function:
\[ E_{0} \left[\sup_{\gamma \in \s{F}_{n,v}'} \left|S_{n,v}(\gamma) \right| \right] = E_{0} \left\{E_0\left[\sup_{\gamma \in \s{F}_{n,v}'} \left|S_{n,v}(\gamma) \right|\  \middle|\ \s{T}_{n,v} \right] \right\} \ .\]
We apply Lemma~\ref{lemma:u_process} to the inner expectation. First, we note that $\s{F}_{n,v}'$ is a uniformly bounded class of functions by the uniform boundedness of $\mu_{n,v}$. Second, the class $\s{F}_{n,v}'$ can be formed as a sequence of compositions of the class $\{a \mapsto I_{(-\infty, a_0]}(a) : a_0 \in \s{A}\}$, which, as discussed above, has polynomial uniform entropy numbers. This implies that the uniform entropy integral in the upper bound of Lemma~\ref{lemma:u_process} is finite. Therefore, we have that
\[  E_{0} \left[\sup_{\gamma \in \s{F}_{n,v}'} \left|S_{n,v}(\gamma) \right|\ \middle|\ \s{T}_{n, v} \right] \ \lesssim\ [N (N-1)]^{1/2} \]
for some universal constant. Thus, the first subcomponent of $R_{n,a_0,3}^\circ$ is also $\bounded(n^{-1})$, and we conclude that $\sup_{a_0 \in \s{A}} |R_{n,a_0,3}^\circ| = \bounded(n^{-1})$. 

We have now shown that, under conditions (B1)--(B2) and (A3), all three remainder terms are at least $\fasterthan(1)$, and thus, Theorem~1 of WC yields the result. \qed

\subsection*{Proof of Theorem~5}

As before, we use Theorem~4 of WC to establish the result. Verification of the conditions (WC.B1)--(WC.B3) and (WC.A4)--(WC.A5) is identical as in the proof of Theorem~2. Hence, we focus on conditions (WC.B4)--(WC.B5). Specifically, defining 
\[K_{n,j}^\circ(\delta) := n^{2/3} \sup_{|u|\leq \delta n^{-1/3}}  \left| R_{n,a +u,j}^\circ - R_{n,a,j}^\circ \right|,\]
for each $j$, we must show that $K_{n,j}^\circ(\delta) \inprob 0$ for all $\delta$ small enough and that, for some $\beta\in(1,2)$, $\delta\mapsto\delta^{-\beta}E[ K_{n,j}^\circ(\delta)]$ is decreasing for all $\delta$ small enough and $n$ large enough. Verification for the term $K_{n,1}^\circ$ is nearly identical to the analysis presented for $K_{n,1}$ in the proof of Theorem~2. For $K_{n,2}^\circ(\delta)$, we first define 
\[\s{G}_{n,v, R} := \left\{ \left(\phi_{\mu_{n,v}, g_{n,v}, a_0 + u}'  - \phi_{\mu_{\infty}, g_{\infty}, a_0 + u}'\right) -\left( \phi_{\mu_{n,v}, g_{n,v}, a_0}'  - \phi_{\mu_{\infty}, g_{\infty}, a_0}' \right): |u| \leq R\right\}\]
for each $R > 0$, where we have suppressed dependence on $a_0$. We then have that \[K_{n,2}^\circ(\delta) \leq n^{1/6} \max_v \sup_{g \in \s{G}_{n,v, \delta  n^{-1/3}}} \left| \d{G}_n^v g\right|.\] As before, we condition on $\s{T}_{n,v}$, and write 
\[ E_0 \left[ \sup_{g \in \s{G}_{n,v, \delta  n^{-1/3}}} \left| \d{G}_n^v g\right| \right] = E_0\left\{ E_0 \left[ \sup_{g \in \s{G}_{n,v, \delta  n^{-1/3}}} \left| \d{G}_n^v g\right|\ \middle|\ \s{T}_{n,v} \right]  \right\} \ .\]
Thus, $\mu_{n,v}$ and $g_{n,v}$ are fixed with respect to the inner expectation. We note that 
\begin{align*} 
&\left| \left[\phi_{\mu_{n,v}, g_{n,v}, a_0 + u}'(o)  - \phi_{\mu_{\infty}, g_{\infty}, a_0 + u}'(o)\right] -\left[ \phi_{\mu_{n,v}, g_{n,v}, a_0}'(o)  - \phi_{\mu_{\infty}, g_{\infty}, a_0}'(o)\right]\right| \\
&\qquad=\ \left| I_{a_0,u}(a)\left\{  \frac{ y - \mu_{n,v}(a,w)}{g_{n,v}(a,w)}- \frac{ y - \mu_\infty(a,w)}{g_\infty(a,w)} + \int\left[ \mu_{n,v}( a, \tilde{w}) - \mu_{\infty}(a, \tilde{w})\right] \,Q_0(d\tilde{w})\right\} \right.\\
&\qquad\qquad \left. + \int_{a_0}^{a_0 + u} \left[\mu_{n,v}(\tilde{a}, w) - \mu_\infty(\tilde{a},w) \right]\, F_0(d\tilde{a}) \right| \\
&\qquad\leq\   I_{[a_0-u, a_0+u]}(a)\left\{\left| \left[ y - \mu_\infty(a, w) \right] \left[ \frac{1}{g_{n,v}(a,w)} - \frac{1}{g_{\infty}(a,w)}\right]\right| + \frac{1}{g_{n,v}(a,w)} \left| \mu_{n,v}(a,w) - \mu_\infty(a,w)\right| \right\} \\
&\qquad\qquad +I_{[a_0-u, a_0+u]}(a)\int \left| \mu_{n,v}( a, \tilde{w}) - \mu_{\infty}(a, \tilde{w})\right| \,Q_0(d\tilde{w})+ \int_{a_0 - u}^{a_0 + u}\ \left|\mu_{n,v}(\tilde{a}, w) - \mu_\infty(\tilde{a},w) \right|\, F_0(d\tilde{a}) \ .
 \end{align*}
We will take as envelope function $G_{n,v, R}$ for $\s{G}_{n,v,R}$ the sum of terms on the right-hand side of the last inequality above, with $u$ replaced by $R$. We then have by Theorem 2.14.1 of VW that
 \[ E_0\left[\sup_{g \in \s{G}_{n,v, \delta  n^{-1/3}}} \left| \d{G}_n^v g  \right|\ \middle|\ \s{T}_{n,v} \right] \ \leq\ C \| G_{n,v, \delta  n^{-1/3}}\|_{P_0,2} J(1, \s{G}_{n,v, \delta  n^{-1/3}}) \ . \]
The class $\s{G}_{n,v,\delta n^{-1/3}}$ is once again contained in a sequence of Lipschitz transformations of the class $\{ a \mapsto I_{(a_0, a_0 + u]}(a) : u \in \d{R}\}$ and various fixed functions, so that the class has polynomial uniform entropy numbers and $J(1, \s{G}_{n,v, \delta  n^{-1/3}})$ is uniformly bounded for all $n$. We then have for some $C_4 < \infty$ that
\[ E_0 \left[ \sup_{g \in \s{G}_{n,v, \delta  n^{-1/3}}} \left| \d{G}_n^v g\right| \right] \ \leq\ C_4  E_0 \left[  \| G_{n,v, \delta  n^{-1/3}}\|_{P_0,2} \right]\ .\]
By the boundedness condition (B1), for all $n$ large enough, we have that 
\begin{align*}
\| G_{n,v, \delta  n^{-1/3}}\|_{P_0, 2}\ \leq&\ \ (K_0 + K_3) \left\{ \int I_{[a_0-\delta  n^{-1/3}, a_0+\delta  n^{-1/3}]}(a) \left[ \frac{1}{g_{n,v}(a,w)} - \frac{1}{g_{\infty}(a,w)}\right]^2 \, dP_0(o) \right\}^{1/2}\\
&\quad+ 3K_1^{-1} \left\{ \int I_{[a_0-\delta  n^{-1/3}, a_0+\delta  n^{-1/3}]}(a)   \left[  \mu_{n,v}(a,w) - \mu_\infty(a,w)\right]^2 \, dP_0(o) \right\}^{1/2}  \\
\leq&\ \ K_1^{-2}(K_0 + K_3) \left\{ P_0\left ( | A - a_0| \leq \delta  n^{-1/3}\right)\int \left[ \frac{1}{g_{n,v}(a,w)} - \frac{1}{g_{\infty}(a,w)}\right]^2 \, dP_0(o) \right\}^{1/2}\\
&\quad+ 3K_1^{-1} K_0^2\left\{ P_0\left ( | A - a_0| \leq \delta  n^{-1/3}\right) \int  \left[ \mu_{n,v}(a,w) - \mu_\infty(a,w)\right]^2 \, dP_0(o) \right\}^{1/2}  \\
\lesssim&\ \ \delta^{1/2} n^{-1/6} \left\{ \left[ P_0 \left( g_{n,v}^{-1}  -g_\infty^{-1}\right)^2 \right]^{1/2}+ \left[ P_0\left(  \mu_{n,v} - \mu_\infty\right)^2 \right]^{1/2}\right\}  \ .
\end{align*}
Since all terms involved are uniformly bounded, we therefore have 
\[ E_0 \left[ K_{n,2}^\circ(\delta)\right]\ \lesssim\ \delta^{1/2} \max_v \left\{ \left[ P_0 \left( g_{n,v}^{-1}  -g_\infty^{-1}\right)^2 \right]^{1/2}+ \left[ P_0\left(  \mu_{n,v} - \mu_\infty\right)^2 \right]^{1/2}\right\} \ ,\]
so that $K_{n,2}^\circ(\delta) \inprob 0$ for each $\delta > 0$ and $\delta \mapsto \delta^{-\beta} E_0 \left[ K_{n,2}^\circ(\delta)\right]$ is decreasing for any $\beta \in (1,2)$ and all $n$ large enough.

For $K_{n,3}^\circ(\delta)$, we first note that (WC.B4) has already been shown to hold in the proof of Theorem 1, since 
\[n^{2/3}\sup_{|u| \leq \delta n^{-1/3}} |R_{n, a_0 +u, 3}^\circ - R_{n, a_0, 3}^\circ|\ \leq\ 2n^{2/3} \sup_{a_0 \in \s{A}} |R_{n,a_0,3}^\circ|\ =\  \bounded(n^{-1/3})\ .\]
We verify (WC.B5) for each of the three subcomponents of $K_{n,3}^\circ(\delta)$ defined by the three subcomponents of $R_{n,a_0,3}^\circ$. Due to the assumed boundedness of $\mu_{n,v}$ and $g_0^{-1}$, the contribution of the third subcomponent is bounded for all $\delta$ small enough up to a constant (not depending on $\delta$ or $n$) by $n^{-1/3}P_0\left(|A - a| \leq \delta n^{-1/3}\right) \lesssim n^{-2/3} \delta$, which satisfies (WC.B5). For the second subcomponent, which is an ordinary empirical process term, analogous methods to that used for $K_{n,2}^\circ$ can be used to verify (WC.B5). The first subcomponent requires controlling $\sum_{i, j \in \s{V}_{n,v}, i \neq j} \gamma_{\mu_{n,v}, a_0, u}^*(O_i, O_j)$, where we define
\begin{align*}
\gamma_{\mu, a_0, u}^*(o_i, o_j)\ &:=\ I_{a_0, u}(a_i) \mu(a_i, w_j) + I_{a_0,u}(a_j) \mu(a_j, w_i) \\
&\qquad - \int \left[ I_{a_0, u}(a_i) \mu(a_i, w) + I_{a_0, u}(a_j) \mu(a_j, w)\right] Q_0(dw) \\
&\qquad - \int  I_{a_0, u}(a)\left[ \mu(a, w_i) + \mu(a, w_j)\right] F_0(da)+  2\iint I_{a_0, u}(a) \mu(a, w) F_0(da) Q_0(dw) \ .
\end{align*}
Conditioning upon $\s{T}_{n,v}$, $\mu_{n,v}$ becomes fixed, so that the function $\gamma_{\mu_{n,v}, a_0, u}^*$  falls in the class $\s{H}_{\delta,n,v}^\circ := \{\gamma_{\mu_{n,v}, a_0, u}^* : |u| \leq \delta \}$ for all $|u| \leq \delta$. Thus, 
\[ \left\{ \sum_{i, j \in \s{V}_{n,v},i \neq j} \gamma^*(O_i, O_j): \gamma^* \in \s{H}_{\delta, n,v} \right\}\] 
is a $P_0$-degenerate $U$-process conditional on $\s{T}_{n,v}$. The class $H_{\delta, n,v}$ has uniform entropy bounded up to a constant by $-\log\varepsilon$ relative to the envelope 
\[H_{\delta, n, v}:(a_1,w_1, a_2, w_2) \mapsto 2K_{0}I_{[0, \delta]}(|a_1- a_0|) + 2K_{0}I_{[0, \delta]}(|a_2- a_0|)  + 4K_{0} P_0\left(|A - a_0| \leq \delta\right).\]
Since $\| H_{\delta} \|_{P_0 \times P_0,2} \lesssim \delta^{1/2}$, Lemma~\ref{lemma:u_process} yields that
\[n^{2/3}E_0 \left[ \sup_{\gamma^* \in \s{H}_{\delta, n, v}} \left|\frac{1}{N^2} \sum_{{\stackrel{i, j \in \s{V}_{n,v}}{i \neq j}}} \gamma^*(O_i, O_j)\right|\ \middle|\ \s{T}_{n,v}\right]\ \lesssim\  n^{-1/3}\delta^{1/2}\]
for all $\delta$ small enough. Hence, (WC.B5) is satisfied for this $U$-process term. \qed

\clearpage

\section*{Supplementary material: additional simulation results}

Figure~\ref{fig:psiprime} presents boxplots of the estimator $\psi_n'(a)$ of the true derivative $\psi_0'(a)$ for each combination of nuisance function estimators used. The estimators are taken to the one-third power because that is what appears in the estimator of the pointwise confidence intervals. The estimators are roughly centered around the truth (shown in red), except for values of $a$ in the tails of the distribution of $A$.
\begin{figure}[ht!]
\centering
\includegraphics[width=\linewidth]{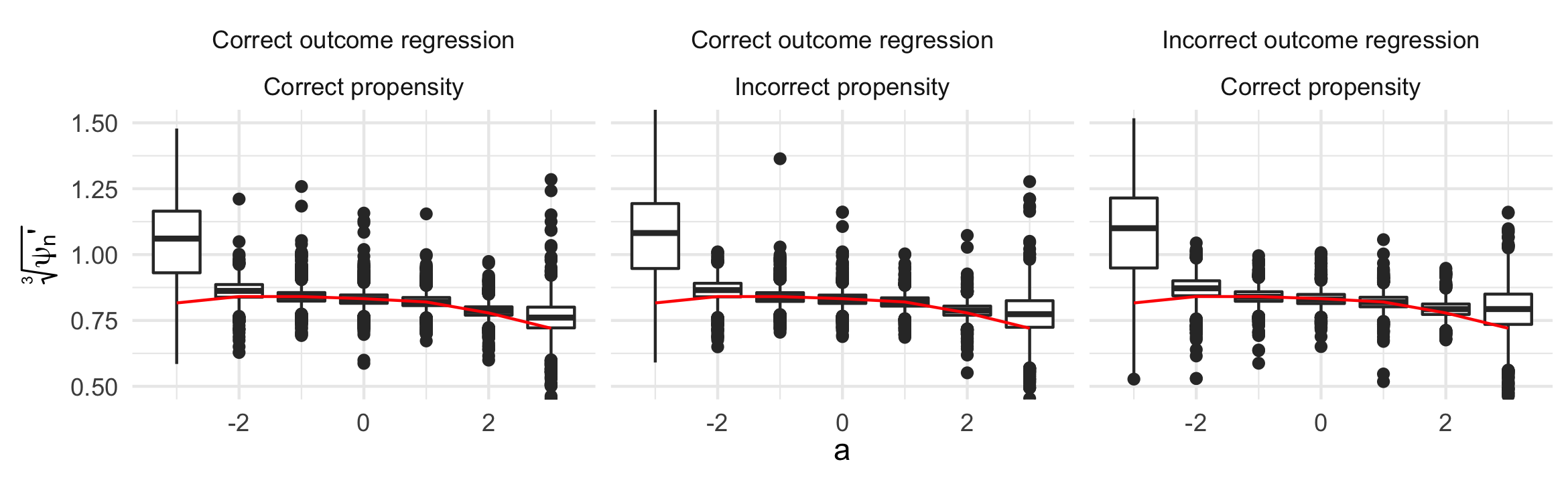}
\caption{Distribution of the estimator $\psi_n'(a)$ of $\psi_0'(a)$ for different values of $a$ over 1000 datasets simulated as described in the text. Red lines show the true values $\psi_0'(a)$.}
\label{fig:psiprime}
\end{figure} 

Figure~\ref{fig:plug_in_kappa} shows histograms of the plug-in estimator of $\kappa_0(a)$. The estimators are taken to the one-third power because that is what appears in the estimator of the pointwise confidence intervals. The estimators are centered around the truth (show in red) when both $\mu_n$ and $g_n$ are consistent, but are biased for some values of $a$ when either $\mu_n$ or $g_n$ is inconsistent.
\begin{figure}[ht!]
\centering
\includegraphics[width=\linewidth]{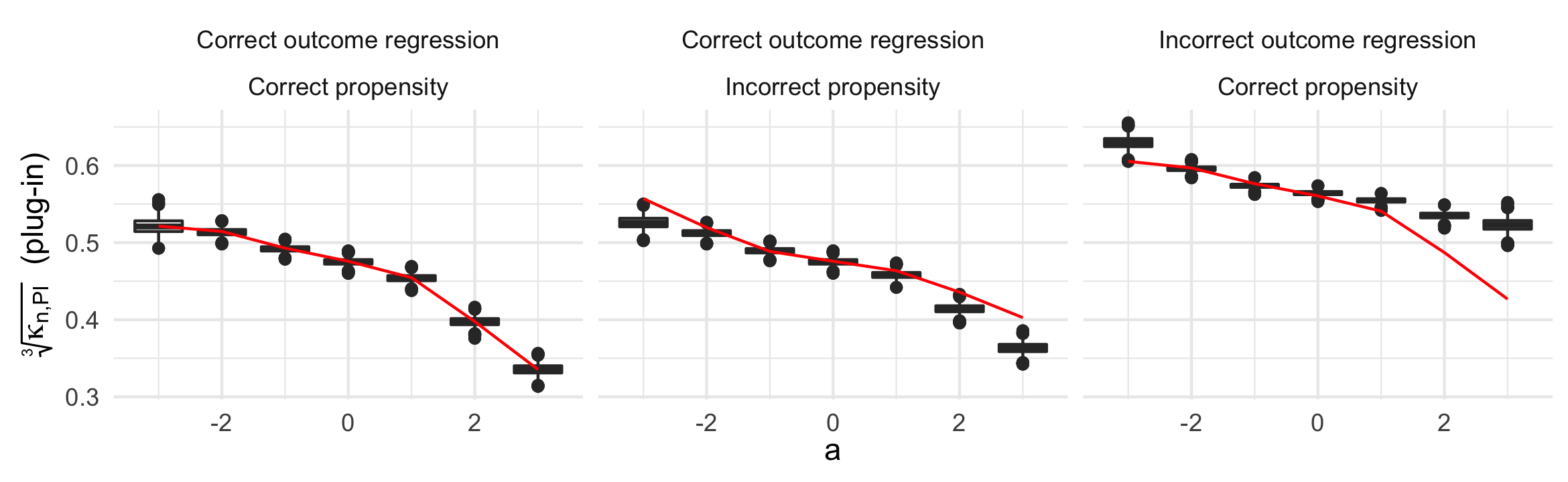}
\caption{Distribution of the plug-in estimator $\kappa_n(a)$ of $\kappa_0(a)$ for different values of $a$  over 1000 datasets simulated as described in the text. Red lines show the true values $\kappa_0(a)$.}
\label{fig:plug_in_kappa}
\end{figure}

Figure~\ref{fig:dr_kappa} shows histograms of the doubly-robust estimator of $\kappa_0(a)$. Once again, the estimators are taken to the one-third power because that is what appears in the estimator of the pointwise confidence intervals. In all settings considered, the estimators are roughly centered around the truth, which is shown in red. However, the spread of the estimator around the true scale is substantially larger than that of the plug-in estimator.
\begin{figure}[ht!]
\centering
\includegraphics[width=\linewidth]{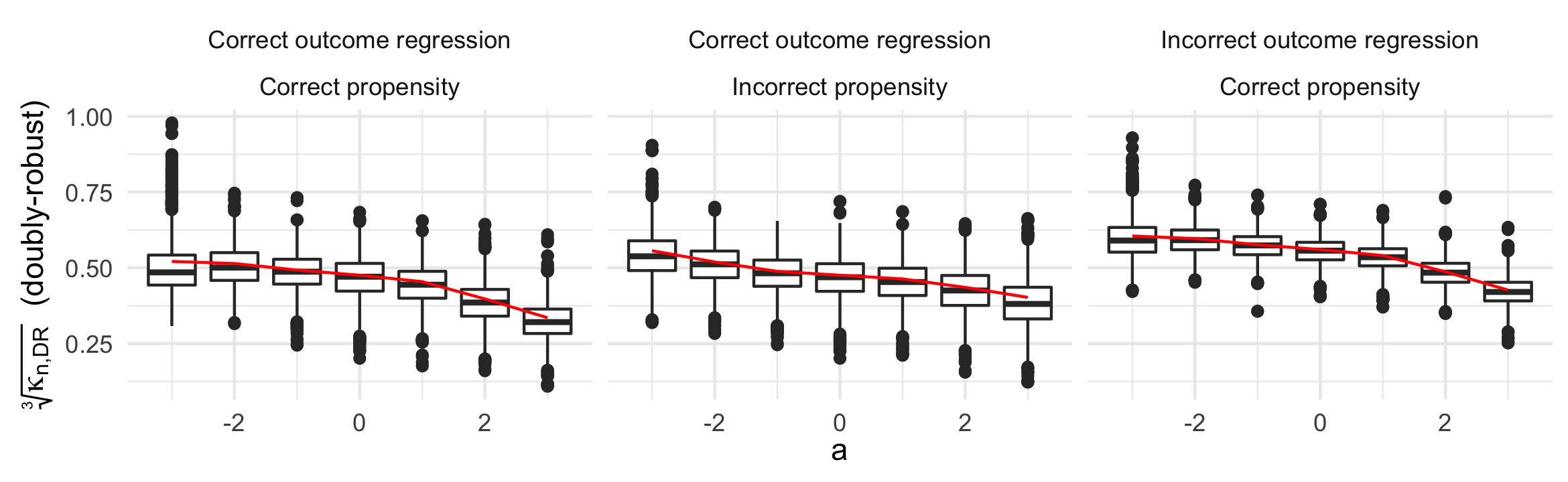}
\caption{Distribution of the doubly-robust estimator $\kappa_n(a)$ of $\kappa_0(a)$ for different values of $a$ over 1000 simulated datasets as described in the text. Red lines show the true values $\kappa_0(a)$.}
\label{fig:dr_kappa}
\end{figure}

\clearpage
\section*{Supplementary material: additional data analyses}

Figure~\ref{fig:tcell_responses} presents the estimated probability of a positive CD8+ T-cell response as a function of BMI for BMI values between the 0.05 and 0.95 quantile of the marginal empirical distribution of BMI using our estimator (left panel), the local linear estimator (middle panel), and the sample-splitting estimator (right panel). Pointwise 95\% confidence intervals are shown as dashed/dotted lines.


\begin{figure}[h!]
\centering
\includegraphics[width=6.5in]{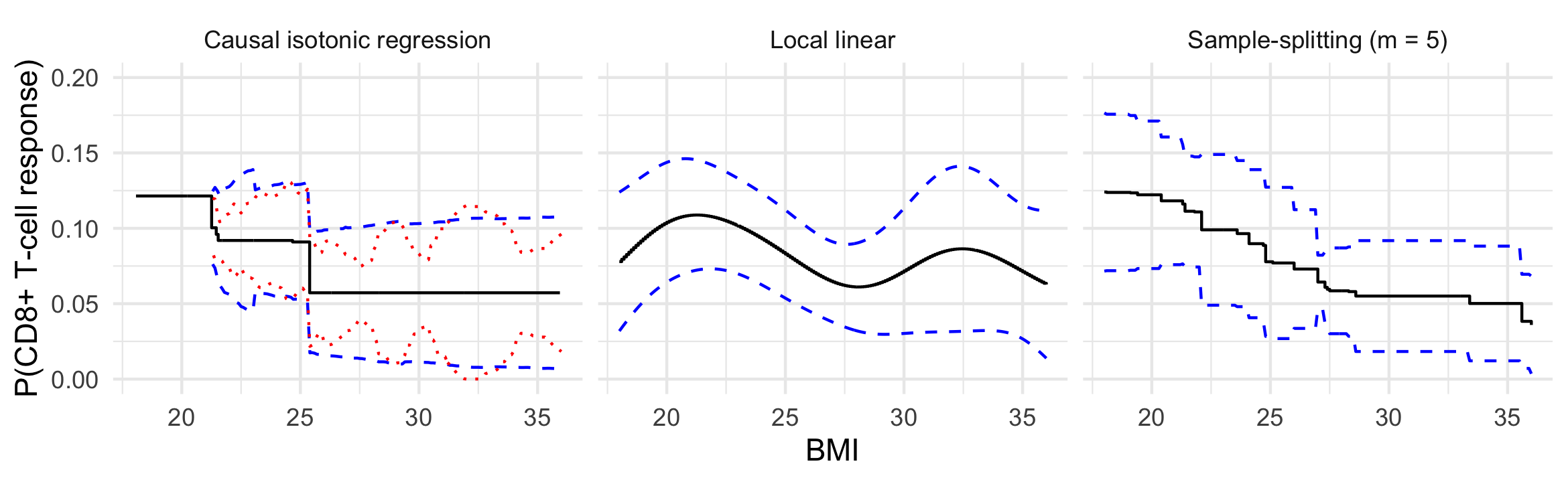}
\caption{Estimated probabilities of CD8+ T-cell response and 95\% pointwise confidence intervals as a function of BMI, adjusted for sex, age, number of vaccinations received, vaccine dose, and study. The left panel displays the estimator proposed here, the middle panel the local linear estimator of \cite{kennedy2016continuous}, and the right panel the sample-splitting version of our estimator with $m=5$ splits. In the left panel, the blue dashed lines are confidence intervals based on the plug-in estimator of the scale parameter, and the dotted lines are based on the doubly-robust estimator of the scale parameter.}
\label{fig:tcell_responses}
\end{figure}

\end{document}